\documentclass{article}
\usepackage[utf8]{inputenc}
\usepackage{amsmath,amssymb,amsfonts,amsthm}
\usepackage{mathrsfs}
\usepackage{indentfirst}
\usepackage[pdftex]{color,graphicx}
\usepackage[pdftex,bookmarks,unicode,colorlinks]{hyperref}
\usepackage{enumerate} 
\usepackage[numbers]{natbib} 
\usepackage{yfonts} %Fraktur/Gothic fonts
\usepackage[toc,page]{appendix}
\usepackage{bm}
\usepackage[normalem]{ulem}

\usepackage{bigfoot}
\DeclareNewFootnote{A}[arabic]
\DeclareNewFootnote{B}[fnsymbol]
\usepackage{etoolbox}
\makeatletter
\patchcmd\maketitle{\def\@makefnmark{\rlap{\@textsuperscript{\normalfont\@thefnmark}}}}{}{}{}
\makeatother
\makeatletter
\def\thanksA#1{% <--- These %'s are necessary for spacing
  \footnotemarkA\protected@xdef\@thanks{\@thanks%
        \protect\footnotetextA[\the \c@footnoteA]{#1}}%
}
\def\thanksB#1{%
  \footnotemarkB%
  \protected@xdef\@thanks{\@thanks%
        \protect\footnotetextB[\the \c@footnoteB]{#1}}%
}
\makeatother

\makeatletter
\newcommand{\overleftsmallarrow}{\mathpalette{\overarrowsmall@\leftarrowfill@}}
\newcommand{\overarrowsmall@}[3]{%
  \vbox{%
    \ialign{%
      ##\crcr
      #1{\smaller@style{#2}}\crcr
      \noalign{\nointerlineskip}%
      $\m@th\hfil#2#3\hfil$\crcr
    }%
  }%
}
\def\smaller@style#1{%
  \ifx#1\displaystyle\scriptstyle\else
    \ifx#1\textstyle\scriptstyle\else
      \scriptscriptstyle
    \fi
  \fi
}
\makeatother
\newcommand{\rev}{\overleftsmallarrow}

\newtheorem{thm}{\bf Theorem} [section]
\newtheorem{lem}[thm]{Lemma}
\newtheorem{prop}[thm]{Proposition}
\newtheorem{cor}[thm]{Corollary}
\newtheorem{Defi}{Definition}

\newtheorem{Assu}{Assumption}

\numberwithin{Exa}{section}
\numberwithin{equation}{section}
\numberwithin{Lem}{section}
\numberwithin{Defi}{section}
\numberwithin{Theo}{section}
\numberwithin{Rem}{section}
\numberwithin{Coro}{section}
\numberwithin{Fig}{section}

\newtheorem{rem}[thm]{Remark}

\usepackage{tikz}
\usetikzlibrary{arrows.meta,bending}
\newcommand{\R}{\mathbb R}

\newcommand{\pt}{\partial}
\newcommand{\E}{\mathbf E}
\renewcommand{\P}{\mathbf P}
\newcommand{\ind}{\mathbf 1}
\newcommand{\e}{\epsilon}

\newcommand{\id}{\mathbf{Id}}
\renewcommand{\d}{\mathrm{d}}

\begin{document}

\title{\bf\Large A study of path measures based on second-order Hamilton--Jacobi equations and their applications \\
in stochastic thermodynamics
%Second-order Hamilton--Jacobi equations at the crossroads of stochastic geometric mechanics and stochastic thermodynamics
}

\author{
\bf\normalsize{
Jianyu Hu\thanksA{Division of Mathematical Sciences, School of Physical and Mathematical Sciences, Nanyang Technological University, 21 Nanyang Link, Singapore 637371. Email: \texttt{jianyu.hu@ntu.edu.sg}},
Qiao Huang\thanksA{School of Mathematics, Southeast University, Nanjing 211189, P.R. China. Email: \texttt{qiao.huang@seu.edu.cn}}
$^{,}$\thanksB{Corresponding author},
Yuanfei Huang\thanksA{Asia Pacific Center for Theoretical Physics, Pohang 37673, Korea, and Department of Mathematics, City University of Hong Kong, Kowloon, Hong Kong SAR. Email: \texttt{yuanfei.huang@apctp.org}},
Jean-Claude Zambrini\thanksA{Department of Mathematics, Faculty of Sciences, Group of Mathematical Physics (GFMUL), University of Lisbon, Campo Grande, Edifício C6, 1749-016, Lisbon, Portugal
. Email: \texttt{jczambrini@fc.ul.pt}}
}
}

\maketitle
\vspace{-0.5in}

% \begin{abstract}
% This paper systematically investigates the mathematical structure of path measures, both from a measure-theoretical perspective and through stochastic differential equations. The realization of path measures as Langevin systems hinges on the pivotal role of second-order Hamilton--Jacobi equations, which form the foundation of stochastic geometric mechanics and applications in stochastic thermodynamics. We explore the emergence of the Onsager--Machlup functional in large deviation theory, the rates of entropy production in non-equilibrium thermodynamic processes, entropy minimization problems encoded in stochastic geometric mechanics, and the identification of Langevin systems from most probable paths.
%   \bigskip\\
%   \textbf{AMS 2010 Mathematics Subject Classification:} ... \\
%   \textbf{Keywords and Phrases:} ...
% \end{abstract}

\begin{abstract}
This paper provides a systematic investigation of the mathematical structure of path measures and their profound connections to stochastic differential equations (SDEs) through the framework of second-order Hamilton--Jacobi (HJ) equations. This approach establishes a unified methodology for analyzing large deviation principles (LDPs), entropy minimization, and entropy production in stochastic systems. Second-order HJ equations are shown to play a central role in bridging stochastic dynamics and measure theory while forming the foundation of stochastic geometric mechanics and their applications in stochastic thermodynamics.

The large deviation rate function is rigorously derived from the probabilistic structure of path measures and proved to be equivalent to the Onsager--Machlup functional of stochastic gradient systems coupled with second-order HJ equations. We revisit entropy minimization problems, including finite time horizon problems and Schr\"{o}dinger's problem, demonstrating the connections with stochastic geometric mechanics. Furthermore, we present a novel decomposition of entropy production for stochastic systems, revealing that thermodynamic irreversibility can be interpreted as the difference of the corresponding forward and backward second-order HJ equations. 
% we tackle the challenging problem of identifying stochastic gradient systems from observed most probable paths by reformulating the original nonlinear and non-convex problem into a linear and convex framework through a second-order HJ equation. 
% This reformulation facilitates the application of structure-preserving kernel methods, ensuring well-posedness and computational efficiency. 
Together, this work establishes a comprehensive mathematical study of the relations between path measures and stochastic dynamical systems, and their diverse applications in stochastic thermodynamics and beyond.
  %35F21, 70H20, 49Q22
\end{abstract}

\newpage

\tableofcontents

\section{Introduction}

Classical Hamilton--Jacobi (HJ) theory offers a powerful approach by reformulating Hamiltonian mechanics as a first-order partial differential equation for the action function~\cite{arnol2013mathematical}, while simultaneously providing profound geometric insights into the integration of motion~\cite{giaquinta1996calculus}. The modern stochastic extension of Hamilton--Jacobi formalism, formulated via \emph{second-order Hamilton--Jacobi equations}, enables a rigorous variational approach to stochastic processes with non-differentiable trajectories.
% ---a fundamental breakthrough unifying classical and stochastic dynamics. 
The history of second-order HJ equations originates from the field of stochastic control, where foundational contributions by Bismut \cite{bismut1973conjugate,bismut1976linear}, Peng \cite{peng1992stochastic}, %peng2005generalized
Pardoux \cite{pardoux1999bsdes}, and P.-L. Lions \cite{crandall1984some} laid the groundwork for their systematic development. In stochastic control problems, second-order HJ equations naturally emerge as tools for characterizing value functions, particularly through the dynamic programming principle and backward stochastic differential equations \cite{bismut1973conjugate,peng1992stochastic,pardoux1990adapted}. In Euclidean quantum mechanics, a probabilistic analogy with quantum mechanics inspired by Schr\"odinger \cite{Schrodinger1932}, these equations serve as an analytical bridge between stochastic processes and quantum dynamics 
% providing a probabilistic interpretation of quantum mechanics under a stochastic framework 
\cite{chung2003introduction}. %zambrini1986stochastic
Moreover, second-order Hamilton--Jacobi equations play a central role in stochastic optimal transport problems like Schr\"odinger's problem, where they govern the evolution of cost functions and probability measures in systems driven by stochastic flows \cite{mikami2004monge,mikami2021stochastic,leonard2012schrodinger,leonard2014survey}. More recently, the 2nd-order Hamilton--Jacobi theory has been developed for stochastic geometric mechanics, connecting to stochastic Lagrangian and Hamiltonian systems via canonical transformations of second-order symplectic structures. It derives stochastic Hamilton's equations and variational principles \cite{huang2023second,huang2022hamilton,huang2023gauge}, capturing the interplay between noise and geometry and offering geometric insights into stochastic optimal transport. Second-order HJ equations also act as a bridge connecting stochastic geometric mechanics and statistical mechanics \cite{huang2023stochastic}. 

% Beyond stochastic control, second-order HJ equations have found applications in stochastic geometric mechanics, where the classical geometric insights of symplectic and Hamiltonian systems are extended to stochastic settings, capturing the interplay between noise and geometry.

% For example, the Onsager--Machlup functional \cite{onsager1953fluctuations,machlup1953fluctuations}, a foundation of statistical mechanics, can be reformulated within the Hamilton--Jacobi framework to explore entropy production and fluctuation theorems in nonequilibrium thermodynamics \cite{seifert2005entropy,boffi2024deep}. This perspective unifies the probabilistic structure of stochastic differential equations with thermodynamic quantities while extending its reach to fields such as optimal transport and information theory. 

The perspective of stochastic thermodynamics engages in profound dialogue with probability theory, where stochastic processes are axiomatically constructed through the path space $C_0([0,T]; \R^d)$, the space of continuous paths starting at the origin, with Wiener measure $\mu_0$. In this context, the Wiener measure and its generalizations serve as fundamental dynamical primitives. At the core of this duality lies Girsanov's theorem, a cornerstone of measure theory. This theorem governs measure transformations under absolute continuity conditions, enabling a rigorous analysis of measure-theoretic operations. These include scaling and shift transformations on path spaces, conditional measure reconstruction, and the study of time-marginal densities and time-reversal symmetries in stochastic processes.

In this paper, we focus on a Gibbs measure $\nu_0$ on $C_0([0,T]; \R^d)$, which has a density with respect to the Wiener
measure $\mu_0$ of the form 
\begin{equation}\label{mesuare-nu}
  \frac{\d\nu_0}{\d\mu_0} \varpropto \exp \left( -\Phi \right),
\end{equation}
with some ``energy" functional $\Phi: C_0([0,T]; \R^d) \rightarrow \mathbb{R}$.  
The measure $\nu_0$ arises naturally in a variety of applications, including entropy production along stochastic trajectories \cite{seifert2005entropy}, the Kullback--Leibler (KL) divergence in information projection \cite{todorov2009efficient}, the theory of conditioned diffusions \cite{hairer2011signal}, and the Bayesian approach to inverse problems \cite{stuart2010inverse}.
From a Bayesian perspective, the study \cite{dashti2013map} analyzed the maximum a posteriori (MAP) estimator of the Onsager--Machlup (OM) action functional associated with the distribution $\nu_0$ defined in \eqref{mesuare-nu}. This MAP estimator corresponds to the most probable paths described by the OM functional for stochastic differential equations (SDEs) \cite{durr1978onsager}. Subsequent studies, such as \cite{pinski2015kullback,lu2017gaussian}, investigated Gaussian approximations for transition paths, utilizing the Kullback--Leibler (KL) divergence to quantify these approximations.  More recently, \cite{selk2021information} established a theoretical framework connecting the information projection with the OM action functional in the context of shifted measures, specifically focusing on the laws of SDEs with constant drifts. However, the precise relation between the measure $\nu_0$ in \eqref{mesuare-nu} and the underlying SDE remains unclear. 
From the perspective of stochastic thermodynamics \cite{peliti2021stochastic}, \cite{seifert2005entropy} introduced the concept of entropy production along a single trajectory in nonequilibrium systems, providing a comprehensive framework for understanding thermodynamic quantities in these systems. In \cite{boffi2024deep}, the authors rigorously derived the entropy production rate starting from the stationary state of Langevin systems. However, this derivation remains notably limited and technically intricate, primarily due to the involvement of time-reversed SDEs and the stationary nonequilibrium setting, as well as the lack of a clear geometric interpretation.

The main result of the present paper, in short, is given below; for a precise statement, of Theorem \ref{Cor-sde}. We establish an equivalence between the measure $\nu_0$ in \eqref{mesuare-nu} and the underlying stochastic gradient system 
\begin{align}\label{eqn:SDEintroduction}
\mathrm dX(t)= -\nabla S(t,X(t))\mathrm d t +  \mathrm dB(t), \quad X(0) =0,
\end{align} 
via a second-order Hamilton--Jacobi equation
\begin{equation}\label{HJ-intro}
  \begin{cases}
    \partial_t S(t,x) -\frac{1}{2} |\nabla S(t,x)|^2 + \frac{1}{2}\Delta S(t,x) = -V(t,x), & (t,x)\in [0,T)\times\R^d, \\
    S(T,x) =  g(x), & x\in\R^d,
  \end{cases}
\end{equation}
where we require the functional $\Phi$ to take the form of a cost function (using mass transportation terminology)
\begin{equation}\label{eqn:Phi=intV} 
\Phi(\omega) = \int_0^T V(t, \omega(t)) \, \mathrm{d}t + g(\omega(T)),
\end{equation}
$V: [0,T] \times \mathbb{R}^d \to \mathbb{R}$ and $g: \mathbb{R}^d \to \mathbb{R}$ stand for the running cost and terminal cost, respectively.

The second-order Hamilton--Jacobi equation \eqref{HJ-intro} (known as Hamilton--Jacobi--Bellman (HJB) in stochastic control theory \cite{fleming2006controlled}), establishes a significant connection between path measures \eqref{mesuare-nu} and stochastic dynamical systems \eqref{eqn:SDEintroduction}.

The proof, especially, to derive the 2nd-order HJ equation \eqref{HJ-intro} from $\nu_0$, strongly relies on a stochastic version of the fundamental theorem of calculus, i.e., Lemma \ref{Stoch-FTC}. This lemma highlights a profound rigidity: if the terminal value of one function of Brownian paths equals the time integral of another function along Brownian paths for almost every path, then any spatial dependence that could produce stochastic fluctuations must vanish — the only remaining possibility is that the integrand depends purely on time, which effectively reduces the statement to the classical fundamental theorem of calculus. In other words, echoing It\^o's formula, the presence or absence of the stochastic integral (martingale) term determines whether any spatial dependence can persist. 

Progress in this direction was made by C. L\'eonard in two related works, strongly motivated by Schr\"odinger's problem.
In \cite{leonard2011stochastic}, he analyzed what he called generalized $h$-transforms of a reference Markov process, with the same form as \eqref{mesuare-nu}-\eqref{eqn:Phi=intV}, under a finite relative entropy assumption, and derived the infinitesimal generator of the transformed law, without requiring smoothness of the potentials. This yields explicit formulas for the modified drift both in the diffusion and in the jump process setting.
Later, in \cite{leonard2022feynman}, he established trajectorial versions of the Feynman--Kac and Hamilton--Jacobi--Bellman identities in the same entropy-based framework, showing that the associated Feynman--Kac semigroup remains well-defined and that the transformed process is again a diffusion with the same diffusion matrix but with drift modified by an extended gradient term. These results focus on deriving the drift of the transformed process from the HJB equation, but do not provide the converse direction, namely, deducing HJB equation from the dynamics of the transformed process itself. The latter constitutes a main result of the present work, where we establish such a reverse implication by using the aforementioned stochastic counterpart of the fundamental theorem of calculus. 

This result offers new insights into the dynamical and statistical properties of thermodynamic processes, encompassing concepts such as the Onsager--Machlup functional, large deviations, entropy minimization, entropy production, and most probable paths.
On one hand, the measure-theoretic perspective on path measures allows for concise yet broadly accessible and applicable formulations. 
% For example, entropy production in stochastic thermodynamics can be more conveniently computed using a measure-theoretic approach, which also offers a geometric interpretation of stochastic entropy. Similarly, the large deviation theory of stochastic differential equations becomes more accessible when approached from the perspective of measure theory.
On the other hand, the viewpoint of stochastic dynamical systems provides richer and more detailed descriptions of the behavior of thermodynamic observables.
% For instance, although entropy production in stochastic thermodynamics appears more intuitive and concise from a measure-theoretic perspective, the stochastic dynamical systems approach enables the study of the corresponding time-reversed stochastic differential equations, revealing the precise contributions of different components of the system to the total entropy production. See Section \ref{Application II: Entropy production in stochastic thermodynamics} for more details.

Based on the equivalence established via second-order Hamilton--Jacobi equations, we summarize the following key applications:

\begin{enumerate}[(i).]
\item %Both the path measure and stochastic gradient system yield a large deviation principle when the noise is small. The rate function coincides up to a constant with the Onsager--Machlup action functional associated with the potential $\Phi$, with the knowledge of the 2nd-order HJ equation.

Both the path measure and the stochastic gradient system satisfy a large deviation principle (LDP) in the small noise regime, with the associated rate function coinciding, up to a constant, with Onsager--Machlup action functional derived from the potential $\Phi$. This coincidence is nontrivial, as Onsager--Machlup functional for SDEs typically differs from the standard Freidlin--Wentzell rate function by a divergence term involving the drift field (cf. \cite{PrivaultYangZambrini2016}). In our case, however, the two coincide because the SDE is coupled with a 2nd-order HJ equation that includes a small Laplacian term, ensuring equivalence with the path measure. As a result, the Onsager--Machlup functional offers a unified framework for describing the large deviation behavior of both path measures and gradient systems, with the analysis grounded in 2nd-order HJ equations.

% We rigorously establish the LDP for the family of stochastic processes governed by the SDE, leveraging the connection between large deviations and the 2nd-order HJ equation. By introducing a reformulated assumption that ensures the convergence of a terminal condition in the system, the analysis simplifies the representation of the action functional and its limiting behavior.

% The probabilistic representation of the HJ equation plays a central role here. Using Varadhan's lemma, the limiting value of the stochastic system is shown to satisfy the HJ equation, providing a variational characterization of the system's evolution in the limit of vanishing noise. A more explicit LDP is derived for an auxiliary gradient system governed by a modified SDE constructed from the limiting HJ equation. The rate function for this auxiliary system is explicitly computed as the Onsager--Machlup action functional, and its structure reflects the influence of the potential $\Phi$ and the underlying dynamics.

% To connect the original stochastic process with the auxiliary gradient system, the concept of exponential equivalence is employed. By analyzing the coupling between the two systems and using tools like Gronwall's lemma, it is shown that the difference between their trajectories becomes negligible as the noise diminishes. This equivalence allows the LDP for the auxiliary system, with its explicitly computed rate function, to be directly transferred to the original system.

\item We revisit the equivalence between entropy minimization problems with path measure constraints and stochastic optimal control problems with SDE constraints. Our main result serves as a natural bridge, allowing one problem to be solved via the other.
This equivalence can be viewed as a generalization of the portmanteau theorem presented in \cite{selk2021information}, extending it from Cameron--Martin (path-independent) shifts to path-dependent shifts.
Applying this equivalence, we recover the solutions to finite-horizon problems by imposing a fixed initial distribution on path measures. The 2nd-order HJ equation, then, determines the optimal drift field associated with the optimal measure.
% Moreover, we reinterpret optimal control theory through the lens of path measures arising from entropy minimization. Specifically, we employ a second-order Hamilton--Jacobi (HJ) equation to achieve this connection. This approach provides an alternative perspective to the classical proof, which relies primarily on the dynamic programming principle.
We further transform Schr\"odinger's problem into the framework of stochastic optimal control by assuming a terminal cost. We explicitly derive the corresponding Schr\"odinger's system and obtain the solution using the 2nd-order HJ equation.
Finally, we establish connections with stochastic geometric mechanics by deriving the associated stochastic Euler--Lagrange equation, thereby providing a comprehensive geometric interpretation of entropy minimization, stochastic optimal control, and Schr\"odinger's problem.

\item 

By composing path space measure $\nu$ (see \eqref{mesuare-nu}) with time-reversal, we obtain the associated time-reversed stochastic process. Our main theorem establishes its SDE representation and induced path measure. Using Girsanov's theorem, we derive the logarithmic path density ratio between forward and reversed measures, decomposing into boundary condition differences for forward/backward Schr\"odinger systems (coupled second-order Hamilton--Jacobi equations). In stochastic thermodynamics, this logarithmic ratio corresponds to total path entropy production $\Delta s_{\text{tot}}$ \cite{seifert2005entropy}, which satisfies $\mathbb{E}[\Delta s_{\text{tot}}] \geq 0$ by the Second Law. This quantifies directionality and provides an integral fluctuation theorem and a detailed fluctuation theorem for irreversibility. Our results are surprisingly universal since they are not only valid for stationary condition but also for arbitrary general nonequilibrium conditions. Within the SDE framework (time-independent drift), we rigorously prove these results using stochastic calculus, showing consistency with measure-theoretic formulations. Theorem \ref{theo:irreversibilityand2ndHJ} reveals that thermodynamic irreversibility originates from the boundary condition mismatch in Schr\"odinger systems.
\end{enumerate}

\section{The setting and the main result}\label{The setting and the main result}

In this section, we introduce the framework for path measures and present our main result, stated in Theorem \ref{Cor-sde}. Specifically, we establish a notable correspondence between path measures $\{\nu^\e,\e>0\}$ in \eqref{mu-tilde} with $\Phi^\epsilon$ in \eqref{phi} and stochastic gradient systems \eqref{SDE-2} using the second-order Hamilton--Jacobi (2nd-order HJ) equation \eqref{2nd-order HJ}. This correspondence is based on a stochastic analogue of the fundamental theorem of calculus, given in Lemma \ref{Stoch-FTC}. Furthermore, we use Feynman–Kac representation and the Cole–Hopf transformation to demonstrate the existence and uniqueness of classical solutions to this 2nd-order HJ equation.

\subsection{The path space}\label{subsec:path}

%Let $(\mathcal{X}, \|\cdot\|_{\mathcal{X}})$ be a separable Banach space, equipped with the Borel $\sigma$-algebra $\mathcal{F}$. Let $\mu_0$ be a centered Borel Gaussian measure on $(\mathcal{X}, \mathcal{F})$, which we refer to as a reference measure. 
%Denote by $(\mathcal{H}, \|\cdot\|_{\mathcal{H}})$ the Cameron--Martin subspace of $\mathcal{X}$ under $\mu_0$. The triple $(\mathcal{X}, \mathcal{H}, \mu_0)$ is referred to as an abstract Wiener space.

%Denote by $\mathcal{P}$ the set of all Borel probability measures that are absolutely continuous with respect to $\mu_0$. 

% We recall the notion of shift maps. For a path $\gamma \in\mathcal{X}$, the shift map associated with it is the map $T_\gamma: \mathcal X \to \mathcal X$ defined by $T_\gamma\omega=\omega+h$, $\omega\in\mathcal{X}$. The shift measure of $\mu_0$ by $h$ is the pushfoward measure $(T_\gamma)_*\mu_0$, i.e., $(T_\gamma)_*\mu_0(A)=\mu_0(T_\gamma^{-1}(A))$, $A\in\mathcal{F}$. Girsanov theorem tells that the shift measure $(T_\gamma)_*\mu_0$ is absolutely continuous with respect to $\mu_0$ if and only if $\gamma \in\mathcal H$. Let $\mathcal{P}$ be the set of Gaussian shift measures which are absolutely continuous with
% respect to $\mu_0$. 

Let $\mathcal{C}^{d,T} := C([0,T]; \R^d)$ be the \emph{path space} of all continuous functions $\omega: [0,T]\to \R^d$. We equip it with the supremum norm
\begin{equation*}
  \|\omega\|_T := \max_{t\in [0,T]} |\omega(t)|, \quad \omega \in \mathcal{C}^{d,T},
\end{equation*}
where $|\cdot|$ denotes the Euclidean norm in $\R^d$. $\mathcal{C}^{d,T}$ is a Banach space under this norm. Let $\mathcal B(\mathcal{C}^{d,T})$ be the Borel $\sigma$-field generated by the open sets in $\mathcal{C}^{d,T}$. For each $t\in [0,T]$, define a sub-$\sigma$-field of $\mathcal B(\mathcal{C}^{d,T})$ by $\mathcal B_t(\mathcal{C}^{d,T}) := \theta_t^{-1}(\mathcal B(\mathcal{C}^{d,T}))$, where $\theta_t: \mathcal{C}^{d,T} \to \mathcal{C}^{d,T}$ is the truncation map $(\theta_t \omega)(s) = \omega(t \wedge s), t\in [0,T]$. The set $\{ \mathcal B_t(\mathcal{C}^{d,T}) \}_{t\in[0,T]}$ forms a natural increasing filtration of $(\mathcal{C}^{d,T},\mathcal B(\mathcal{C}^{d,T}))$. We denote by $\mathcal{P}$ the set of all probability measures on $(\mathcal{C}^{d,T},\mathcal B(\mathcal{C}^{d,T}))$.

Whenever $X = \{X(t)\}_{t\in[0,T]}$ is a continuous stochastic process on a probability space $(\Omega, \mathcal F, \mathbf P)$, it can be regarded as a random variable on $(\Omega, \mathcal F, \mathbf P)$ with values in $(\mathcal{C}^{d,T},\mathcal B(\mathcal{C}^{d,T}))$, and the pushforward measure $\operatorname{Law} (X) := X_* \mathbf P = \mathbf P \circ X^{-1}$ on $(\mathcal{C}^{d,T},\mathcal B(\mathcal{C}^{d,T}))$ is called the law of $X$.

Denote by $\mathcal{C}^{d,T}_0$ the subspace of $\mathcal{C}^{d,T}$ consisting only of those functions that take the value $0\in\R^d$ at time 0.
The unique probability measure $\mu_0$ on $(\mathcal{C}^{d,T}_0,\mathcal B(\mathcal{C}^{d,T}_0))$, in which the coordinate mapping process 
$$W(t,\omega) := \omega(t), \quad t\in [0,T],$$
is a standard $d$-dimensional Brownian motion, is called Wiener measure. Note that $\mu_0$ can be pushforwarded to the whole path space $\mathcal{C}^{d,T}$ by the embedding $\mathcal{C}^{d,T}_0 \subset \mathcal{C}^{d,T}$.
Conversely, a standard, $d$-dimensional Brownian motion $B$ defined on any probability space can be thought of as a random variable in $(\mathcal{C}^{d,T}_0,\mathcal B(\mathcal{C}^{d,T}_0))$; regarded this way, the Brownian motion $B$ induces the Wiener measure by $\mu_0 = \text{Law}(B)$. Thus, $(\mathcal{C}^{d,T}_0,\mathcal B(\mathcal{C}^{d,T}_0), \mu_0)$ is called the canonical probability space for Brownian
motion.

In the sequel, we denote by $C^{k}_b(U)$ the space of all bounded continuous functions on a subdomain $U \subseteq\R^m$ with bounded and continuous derivatives up to order $k$. Similarly, $C^{k,l}_b([0,T]\times \R^d)$ denotes the space of all bounded continuous functions on $[0,T]\times \R^d$ whose derivatives up to order $k$ in $t$ and up to order $l$ in $x$ are also bounded and continuous.

We now introduce several basic operations and concepts in the path space $\mathcal{C}^{d,T}$ that will be used throughout the paper. More details are provided in Appendix \ref{appendixA}.
\begin{itemize}
    \item{\textbf{Scaling:}} For an $\e>0$, the scaling map $\delta_\e : \mathcal{C}^{d,T} \to \mathcal{C}^{d,T}$ is defined as $\delta_\e \omega = \sqrt\e \omega$. 
We denote by $\mu^\e_0 := (\delta_\e)_* \mu_0 = \mu_0 \circ \delta_\e^{-1}$ the $\e$-scaling of the Wiener measure $\mu_0$. 
The scaling measure $\mu^\e_0$ is the probability distribution of the $\e$-scaled Brownian motion $\sqrt\e W$.

It should be noted that this type of $\e$-scaling is not applied to a general measure $\nu$.

\item{\textbf{Shifts:}} For a path $\gamma \in \mathcal{C}^{d,T}$, the shift map associated with $\gamma$ is the map $T_\gamma: \mathcal{C}^{d,T} \to \mathcal{C}^{d,T}$ defined by $T_\gamma\omega = \omega + \gamma$. For a measure $\nu$ on $\mathcal{C}^{d,T}$, the pushforward $(T_\gamma)_*\nu = \nu \circ T_\gamma^{-1}$ is called the shift measure of $\nu$ by $\gamma$. We denote
\begin{equation}\label{Wiener-scaling-shift}
  \mu_x^\e := (T_x)_* \mu^\e_0 = (T_x)_* (\delta_\e)_* \mu_0.
\end{equation}
For an element $\omega \in \mathcal{C}^{d,T}$, we will denote
\begin{equation}\label{scaling-shift-abbr}
  \omega_x^\e := T_x \delta_\e \omega = x+\sqrt{\e} \omega.
\end{equation}
Clearly, the law of $\omega_x^\e$ under $\mu_0$ is $\mu_x^\e$.

\item{\textbf{Time-marginals:}} For any $t\in [0,T]$, let $\pi_t : \mathcal{C}^{d,T} \to \R^d$ be the projection map at time $t$, given by $\pi_t(\omega) = \omega(t)$. One can regard each $\pi_t$ as a random vector on $(\mathcal{C}^{d,T},\mathcal B(\mathcal{C}^{d,T}))$.
For a measure $\nu$ on $(\mathcal{C}^{d,T},\mathcal B(\mathcal{C}^{d,T}))$, we define its marginal at time $t$ by $\nu|_t := (\pi_t)_* \nu = \nu \circ \pi_t^{-1}$, as a measure on $\R^d$.
The time marginals of the Wiener measure $\mu_0$ have the following Lebesgue densities, known as heat kernels:
\begin{equation*}
  \rho_0(t, x) := \frac{\mathrm{d} \mu_0|_t (x)}{\mathrm{d} x} = \frac{1}{(2\pi t)^{d/2}}e^{-\frac{|x|^2}{2t}}, \quad (t,x) \in [0,T] \times \R^d.
\end{equation*}
% They give the heat semigroup $\{e^{\frac{1}{2} t\Delta}: t\in\R_+\}$, i.e., the semigroup generated by a standard Brownian motion, given by the convolution
% \begin{align*}
%   e^{\frac{1}{2} t\Delta} f(x) = \int_{\R^d} \rho_0(t, x-y) f(y) \d y, \quad f\in C_ c^\infty(\R^d).
% \end{align*}
The marginal of  $\mu_x^\e$ at time $t$ has the Lebesgue density $\rho_x^\e(t, \cdot) = \rho_0(\e t, \cdot - x)$.

\item{\textbf{Conditioning:}}
A Borel measurable map $f: \mathcal{C}^{d,T} \to \R^d$ can be regarded as a random element on $(\mathcal{C}^{d,T},\mathcal B(\mathcal{C}^{d,T}))$. One can define the conditional expectation of a $\sigma$-finite measure $\nu$ given $f$, denoted as $\E_\nu(\cdot | f) := \E_\nu(\cdot | \sigma(f))$, as well as the regular conditional measure of $\nu$ given $f=x\in\R^d$, denoted as $\nu(\cdot | f = x)$.

\item{\textbf{Kullback--Leibler divergence:}} Given two measures $\nu$ and $\eta$ on $(\mathcal{C}^{d,T}, \mathcal B(\mathcal{C}^{d,T}))$, the Kullback--Leibler (KL) divergence (or relative entropy) of $\nu$ with respect to $\eta$ is defined by the non-negative real number
\begin{equation*}
  D_{\mathrm{KL}}\left(\nu \| \eta\right) := 
  \begin{cases}
    \E_\nu\left[ \log \left(\frac{\mathrm d \nu}{\mathrm d \eta}\right) \right], & \nu \ll \eta, \\
    \infty, & \text{otherwise}.
  \end{cases}
\end{equation*}
It says how different $\nu$ is from $\eta$. Although $D_{\mathrm{KL}}$ is not a metric, it is useful, for instance, in information geometry.
\end{itemize}

\subsection{Path measures}

% \begin{Assu}\label{Assu-2}
% (a) We assume that $\mu_0(C < \infty) >0$ and $\E_{\mu_0}[\exp(-C) |C|]<\infty$.

% (b) We assume that $\E_{\mu_0}[C^2]<\infty $.
% \end{Assu}

% \begin{prop}[It\^o representation, {\cite[Theorem 5.55]{baudoin2014diffusion}}]\label{Ito-rep}
% Let $\left(\mathcal{C}_0[0, T], \mu_0\right)$ be classical Wiener space, and let $C: \mathcal{C}_0[0, T] \rightarrow \mathbb{R}$ be a functional satisfying Assumption \ref{Assu-2}(b). Then, there exists a progressively measurable process $f$ such that
% $$
% C=C_0-\int_0^T f(t)\mathrm{d} B(t),
% $$
% where $C_0 = \E_{\mu_0}[C]$ is a constant and $B$ is the standard Brownian motion on $(\mathcal C_0 [0,T], \mu_0)$.
% \end{prop}

Let $\{\mu^\e_x: x\in\R^d, \e>0\}$ be the family of shifted and rescaled Wiener measures on $(\mathcal{C}^{d,T}, \mathcal B(\mathcal{C}^{d,T}))$, defined in \eqref{Wiener-scaling-shift}, referred to as reference measures.
% From now on, we always fix an $\e>0$, not necessarily small. 
Let $\Phi^\e: \mathcal{C}^{d,T} \rightarrow \mathbb{R}$, $\e>0$, be a family of energy functionals that give a family of probability measures $\{\nu^\e_x: x\in\R^d, \e>0\}$ on $\mathcal{C}^{d,T}$, of which each is absolutely continuous w.r.t. $\mu_x^\e$, via the following Radon--Nikodym derivative
\begin{equation}\label{mu-tilde}
  \frac{\mathrm{d} \nu^\e_x}{\mathrm{d}\mu_x^\e}(\omega) = \frac{1}{Z^\e_{\Phi^\e}(x)} \exp \left(-\frac{1}{\e} \Phi^\e(\omega)\right),
\end{equation}
where 
\begin{equation}\label{norm-const}
  Z^\e_{\Phi^\e}(x) := \E_{\mu_x^\e} \left[ e^{-\frac{1}{\e}\Phi^\e} \right]
\end{equation}
is the normalizing constant. Note that each $\nu^\e_x$ is supported in $\mathcal{C}^{d,T}_x$. In the context of statistical mechanics, $\nu^\e_x$ is referred to as a Gibbs measure, and $Z^\e_{\Phi^\e}(x)$ is known as the partition function.

\begin{Assu}\label{asmp-Phi-1}
For each $\e>0$ and every $r > 0$, there exists an $M = M(\e,r)\in\R$, such that for all $\omega \in \mathcal{C}^{d,T}$,  
$$  
\Phi^\e(\omega) \ge M - r \|\omega\|_T^2.
$$
\end{Assu}

The specific form of the lower bound in Assumption \ref{asmp-Phi-1} is designed to ensure that the normalizing constant $Z^\e_{\Phi^\e}(x)$ is finite so that the r.h.s. of expression \eqref{mu-tilde} is normalizable to give the probability measure $\nu_x^\e$. Indeed, as $\mu_x^\e$ is a Gaussian measure, Fernique's theorem (see \cite[Corollary 2.8.6]{bogachev1998gaussian}) says that there exists $\alpha > 0$ such that,
\begin{equation*}%\label{Fernique}
  \E_{\mu_x^\e} \left[ e^{ \alpha \|\omega\|_T^2} \right] < \infty. 
\end{equation*}
Thus, putting $r = \e \alpha$ in Assumption \ref{asmp-Phi-1}, we have
\begin{equation*}
  Z^\e_{\Phi^\e}(x) \le \E_{\mu_x^\e} \left[ e^{\alpha \|\omega\|_T^2 - M / \e} \right] < \infty. 
\end{equation*}

\subsubsection*{The total measures}

It follows from \eqref{disint} that, given a measure $\mu^\e|_{t=0}$ on $\R^d$, one can construct a measure $\mu^\e$ on $(\mathcal{C}^{d,T},\mathcal B(\mathcal{C}^{d,T}))$ with initial measure $\mu^\e|_{t=0}$ and transition measure $\mu^\e_x$, as follows
\begin{equation}\label{mu-tilde-total}
  \mu^\e(\d\omega) = \int_{\R^d} \mu^\e_x(\d\omega) \mu^\e|_{t=0}(\d x).
\end{equation}
By Corollary \ref{cor-0}-(ii) and \eqref{mu-tilde},
\begin{equation*}
  \frac{\d\nu^\e}{\d\mu^\e}(\omega) %= \frac{\d\nu^\e|_{t=0}}{\d\mu^\e|_{t=0}} (\omega(0)) \frac{\d\nu^\e_x}{\d\mu^\e_x}(\omega) \bigg|_{x = \omega(0)} 
  = \frac{\d\nu^\e|_{t=0}}{\d\mu^\e|_{t=0}}(\omega(0)) \frac{1}{Z^\e_{\Phi^\e}(\omega(0))} \exp \left(-\frac{1}{\e} \Phi^\e(\omega)\right).
\end{equation*}
If we denote
\begin{equation}\label{f-def}
  f^\e(x) := \e\log Z^\e_{\Phi^\e}(x) - \e\log \frac{\d\nu^\e|_{t=0}}{\d\mu^\e|_{t=0}}(x),
\end{equation}
we then have
\begin{equation}\label{nu-mu}
  \frac{\d\nu^\e}{\d\mu^\e}(\omega) = \exp \left\{ -\frac{1}{\e} [f^\e(\omega(0)) + \Phi^\e(\omega)] \right\}.
\end{equation}
Thus, once the initial measure $\mu^\e|_{t=0}$ of $\mu^\e$ and the time zero marginal Radon--Nikodym density $\frac{\d\nu^\e|_{t=0}}{\d\mu^\e|_{t=0}}$ are known, the total measure $\nu^\e$ is fully determined.

\subsubsection*{Potential energy of the cost function form}

To link the somewhat abstract path probability measures $\nu^\e_x$ with concrete SDEs, we consider the following (potential) energy functionals: 
\begin{Assu}\label{asmp-Phi-cost}
  For each $\e>0$, $\Phi^\e: \mathcal{C}^{d,T} \rightarrow \mathbb{R}$ is of the cost function form
\begin{equation}\label{phi}
  \Phi^\e(\omega) = \int_0^T V(t, \omega(t)) \mathrm dt +g^\e(\omega(T)),
\end{equation}
with some bounded below functions $V: [0,T]\times \R^d \to \R$ and $g^\e: \R^d \to \R$.
\end{Assu}

The first term of $\Phi^\e$ is called running cost in optimal control, and the second one the terminal cost.
From \eqref{mu-tilde}, the family of transition measures $\{\nu^\e_x: \e>0 \}$ is given by
\begin{equation}\label{nu-mu-Phi}
  \frac{\d\nu^\e_x}{\d\mu^\e_x}(\omega) = \frac{1}{Z^\e_{\Phi^\e}(x)} \exp \left\{ -\frac{1}{\e} \left[ \int_0^T V(t, \omega(t)) \mathrm dt +g^\e(\omega(T)) \right] \right\}.
\end{equation}
Moreover, the family of total measures $\{\nu^\e: \e>0 \}$ admits also the following more symmetric form, by \eqref{nu-mu},
\begin{equation*}
  \frac{\d\nu^\e}{\d\mu^\e}(\omega) = \exp \left\{ -\frac{1}{\e} \left[ f^\e(\omega(0)) + \int_0^T V(t, \omega(t)) \mathrm dt +g^\e(\omega(T)) \right] \right\},
\end{equation*}
which is sometimes called a generalized $h$-transform \cite{leonard2011stochastic} or the $(e^{-\frac{1}{\e} f^\e}, e^{-\frac{1}{\e} g^\e})$-transform of the reference measure $\mu^\e$ \cite{leonard2014survey}.

\subsection{Correspondence between path measures and SDEs via second-order HJ equations}\label{SDE correspondence}

In this section, we consider the family of transition probability measures $\{\nu^\e_x: x\in\R^d, \e>0\}$ on $\mathcal{C}^{d,T}$ defined in~\eqref{mu-tilde}, where the reference measures $\{\mu^\e_x: x\in\R^d, \e>0\}$ are the shifted scaled Wiener measures in \eqref{Wiener-scaling-shift}, and the potential \(\Phi^\e\) is explicitly given by~\eqref{phi}.

Under this measure structure, we establish a clear correspondence between probability measures $\nu^\e_x$ and SDEs. Specifically, for overdamped Langevin equations (SDEs with additive noise), we show that if the drift term satisfies a nonlinear heat equation, then the distribution of solutions to the SDE in path space coincides exactly with the specified path measure \(\nu_x^\e\).
Furthermore, we demonstrate that for certain stochastic gradient systems, this correspondence between measures and SDE solutions holds if and only if the potential associated with the stochastic gradient system satisfies an appropriate second-order HJ equation.
Finally, we extend this correspondence result to the setting of time-reversed stochastic differential equations, thereby establishing a broader theoretical framework connecting path measures, stochastic systems, and nonlinear PDEs. Recall the shorthand notation $\omega_x^\e := x+\sqrt{\e} \omega$ in \eqref{scaling-shift-abbr}.

\begin{lem}\label{lem-1}
Let Assumptions \ref{asmp-Phi-1} and \ref{asmp-Phi-cost} hold. Fix $\e>0$ and $x\in\R^d$. Let $X^\e_x$, $B$, $(\Omega, \mathcal F, \mathbf P, \{\mathcal F_t\}_{t\in[0,T]})$ be a weak solution of the following functional SDE
\begin{align*}%\label{New-sde}
\mathrm dX^\e_x(t)= b^\e(t,X^\e_x)\mathrm d t + \sqrt\e \mathrm dB(t), \quad X^\e_x (0) =x,
\end{align*}
where $b^\e: [0,T] \times \mathcal{C}^{d,T} \to \R^d$ is an $\{ \mathcal B_t(\mathcal{C}^{d,T}) \}_{t\in[0,T]}$-adapted process satisfying 
\begin{equation}\label{Novikov-3}
   \E_{\mu_x^\e} \left[\exp \left(\frac{1}{2\e} \int_{0}^{T} |b^\e(t, \omega)|^{2} \mathrm{d} t\right)\right]<\infty. 
  % \E_\P \left[\exp \left(\frac{1}{2\e} \int_{0}^{T} |b^\e(t, X^\e_x)|^{2} \mathrm{d} t\right)\right]<\infty.
\end{equation}
Then the law of $X^\e_x$ is $\nu^\e_x$ if and only if $b^\e$ satisfies for $\mu_0$-a.s. $\omega\in \mathcal{C}^{d,T}_0$,
\begin{equation}\label{eqn-13}
  \e \log{Z^\e_{\Phi^\e}(x)} + \Phi^\e(\omega_x^\e) = -\sqrt\e \int_{0}^{T} b^\e(t, \omega_x^\e) \mathrm{d}\omega(t) + \frac{1}{2} \int_{0}^{T} |b^\e(t, \omega_x^\e)|^2 \mathrm{d} t.
\end{equation}
\end{lem}
The functional SDE in the above lemma means that the drift can depend on the past; that is, the SDE is non-Markovian \cite{ikeda2014stochastic}.

\begin{proof}
We first note that, as the law of $x+ \sqrt\e B$ is $\mu_x^\e$, condition \eqref{Novikov-3} amounts to
\begin{equation*}
  \E_\P \left[\exp \left(\frac{1}{2\e} \int_{0}^{T} |b^\e(\cdot, x+ \sqrt\e B)|^{2} \mathrm{d} t\right)\right]<\infty.
\end{equation*} 
Under this condition, we apply Lemma \ref{Gir-thm} (Girsanov theorem), by taking 
%$B$ as $W=\omega(\cdot)$, $\mu$ as $\mu_0$, and 
$\beta$ as $\frac{1}{\sqrt\e} b^\e(\cdot, x+ \sqrt\e B)$. We see that $B - \frac{1}{\sqrt\e} \int_0^\cdot b^\e(s, x+ \sqrt\e B) \mathrm ds$ is a standard Brownian motion under $\mathbf Q$ with density
\begin{equation}\label{Gis-tra}
\begin{split}
  \frac{\mathrm{d} \mathbf Q}{\mathrm{d} \P}(\omega) &= \exp \left(\frac{1}{\sqrt\e} \int_{0}^{T} b^\e(t, x+ \sqrt\e B(\omega)) \mathrm{d}B(t,\omega)-\frac{1}{2\e} \int_{0}^{T} |b^\e(t, x+ \sqrt\e B(\omega))|^{2} \mathrm{d} t\right).
\end{split}
\end{equation}
Thus, the law of $x+ \sqrt\e B$ under $\mathbf Q$ is the same as that of $X^\e_x$ under $\mathbf P$. It follows from Lemma \ref{lem-0}-(i) that, for $\mu_x^\e$-a.s. $\omega \in \mathcal{C}^{d,T}$,
\begin{equation*}
  \begin{split}
    \frac{\d (X^\e_x)_* \mathbf P}{\d \mu_x^\e} (\omega) &= \frac{\d (x+ \sqrt\e B)_* \mathbf Q}{\d (x+ \sqrt\e B)_* \mathbf P} (\omega) = \E_\P \left( \frac{\d \mathbf Q}{\d \mathbf P} \bigg| x+ \sqrt\e B = \omega \right) \\
    &= \E_\P \left[ \exp \left(\frac{1}{\sqrt\e} \int_{0}^{T} b^\e(t, x+ \sqrt\e B) \mathrm{d}B(t)-\frac{1}{2\e} \int_{0}^{T} |b^\e(t, x+ \sqrt\e B)|^{2} \mathrm{d} t\right) \bigg| x+ \sqrt\e B = \omega \right] \\
    &= \exp \left(\frac{1}{\e} \int_{0}^{T} b^\e(t, \omega) \mathrm{d} \omega(t)-\frac{1}{2\e} \int_{0}^{T} |b^\e(t, \omega)|^{2} \mathrm{d} t\right).
  \end{split}
\end{equation*}

This implies, compared with \eqref{mu-tilde}, that the law of $X^\e_x$ under $\P$ is $\nu^\e_x$ if and only if for $\mu_x^\e$-a.s. $\omega\in \mathcal{C}^{d,T}_x$,
\begin{equation*}%\label{eqn-27}
\e \log{Z^\e_{\Phi^\e}(x)} + \Phi^\e(\omega) = - \int_{0}^{T} b^\e(t, \omega) \mathrm{d}\omega(t) + \frac{1}{2} \int_{0}^{T} |b^\e(t, \omega)|^2 \mathrm{d} t,
\end{equation*}
or equivalently, \eqref{eqn-13} holds for $\mu_0$-a.s. $\omega\in \mathcal{C}^{d,T}_0$.
\end{proof}

We now present the main theorem of this paper, which establishes the equivalence between some path measures and stochastic gradient systems. A concise statement of this result is also provided in the introduction.

\begin{thm}\label{Cor-sde}
Let Assumptions \ref{asmp-Phi-1} and \ref{asmp-Phi-cost} hold. Fix $\e>0$ and $x\in\R^d$. Let $X^\e_x$, $B$, $(\Omega, \mathcal F, \mathbf P, \{\mathcal F_t\}_{t\in[0,T]})$ be a weak solution of the following SDE
\begin{align}\label{SDE-2}
  \mathrm dX^\e_x(t)= -\nabla S^\e(t,X^\e_x(t))\mathrm d t + \sqrt\e \mathrm dB(t), \quad X^\e_x (0) =x,
\end{align}
where the potential function $S^\e\in C_b^{1,3}([0,T]\times\mathbb{R}^d)$ satisfies 
\begin{equation}\label{Novikov-S}
   \E_{\mu_x^\e} \left[\exp \left(\frac{1}{2\e} \int_{0}^{T} |\nabla S^\e(t, \omega(t))|^{2} \mathrm{d} t\right)\right]<\infty. 
  %\E_\P \left[\exp \left(\frac{1}{2\e} \int_{0}^{T} |\nabla S^\e(t, X^\e_x(t))|^{2} \mathrm{d} t\right)\right]<\infty.
\end{equation}
Suppose $V\in C^{0,1}_b([0,T]\times\mathbb{R}^d)$ and $g^\e\in C^1_b(\R^d)$. 
Then the law of $X^\e_x$ is $\nu^\e_x$ if and only if $S^\e$ is determined (up to a function depending only on time) by the following second-order Hamilton--Jacobi (2nd-order HJ) equation
\begin{equation}\label{2nd-order HJ}
  \begin{cases}
    \partial_t S^\e(t,y) -\frac{1}{2} |\nabla S^\e(t,y)|^2 + \frac{\e}{2}\Delta S^\e(t,y) = -V(t,y), & (t,y)\in (0,T)\times\R^d, \\
    S^\e(T,y) =  g^\e(y), & y\in\R^d, \\
    S^\e(0,x) = -\e\log{Z^\e_{\Phi^\e}(x)}. &
  \end{cases}
\end{equation}
\end{thm}

In stochastic optimal control, Equation \eqref{2nd-order HJ} is known as Hamilton--Jacobi--Bellman \cite{fleming2006controlled}.

\begin{proof}
From Lemma \ref{lem-1}, we see that the law of $X^\e_x$ is $\nu^\e_x$ if and only if for $\mu_0$-a.s. $\omega\in \mathcal{C}^{d,T}$,
\begin{equation*}
\begin{split}
  &\e \log{Z^\e_{\Phi^\e}(x)} + g^\e(\omega_x^\e(T)) + \int_{0}^{T} V(t, \omega_x^\e(t)) \mathrm d t \\
  &=   \sqrt\e \int_{0}^{T} \nabla S^\e(t, \omega_x^\e(t)) \mathrm{d}\omega(t) + \frac{1}{2} \int_{0}^{T} |\nabla S^\e(t, \omega_x^\e(t))|^2 \mathrm{d} t.
\end{split}
\end{equation*}
Applying It\^o's formula to $S^\e(t, \omega_x^\e(t))$ in the above equality, as $S\in C^{1,2}$, we have
\begin{equation}\label{eqn-20}
  \begin{split}
    &\e \log{Z^\e_{\Phi^\e}(x)} + S^\e(0, x) + g^\e(\omega_x^\e(T)) - S^\e(T, \omega_x^\e(T)) \\
    &= \int_{0}^{T} \left( -\pt_t S^\e + \frac{1}{2} |\nabla S^\e|^2 - \frac{\e}{2} \Delta S^\e - V \right)(t, \omega_x^\e(t)) \mathrm{d} t.
  \end{split}
\end{equation} 
% and using the arbitrariness of $\omega\in \mathcal{C}^{d,T}$.
% For necessity, we plug equations \eqref{2nd-order HJ1} into \eqref{eqn-20}, we get $S^\e(0,x) = \e\log{Z^\e_{\Phi^\e}(x)}$. For sufficiency, we observe that \eqref{eqn-20} holds if we are provided with the 2nd-order HJ equation \eqref{2nd-order HJ} and $S^\e(0, x) - S^\e(T, \omega_x^\e(T)) = \e \log{Z^\e_{\Phi^\e}(x)} + g^\e(\omega_x^\e(T))$ for $\mu_0$-a.s. $\omega$. The latter is equivalent to $S^\e(T,x+y) = S^\e(0, x) - \e\log{Z^\e_{\Phi^\e}(x)} - g^\e(x+y)$ for all $y\in\R^d$, since $\omega(T)$ has full support on $\R^d$.
The sufficiency is now clear. 
The necessity follows from \eqref{eqn-20} and the following lemma.
\end{proof}

\begin{lem}[Stochastic version of fundamental theorem of calculus]\label{Stoch-FTC}
  Fix $\e>0$, $T>0$ and $x\in\R^d$. Let $f_1 \in C_b^1(\R^d)$ and $f_2 \in C_b^{0,1}([0,T] \times \R^d)$. If the following equality holds for $\mu_0$-a.s. $\omega\in \mathcal{C}^{d,T}_0$,
  \begin{equation*}
   f_1(\omega_x^\e(T)) = \int_0^T f_2(s, \omega_x^\e(s)) \d s,
  \end{equation*}
  Then there is a function $F\in C^1([0,T])$ such that
  \begin{equation*}
    f_2(t,x) = F'(t), \quad f_1(x) = F(T), \quad \forall (t,x)\in [0,T]\times\R^d.
  \end{equation*}
\end{lem}

\begin{proof}
Define
\begin{equation*}
  \begin{split}
    u(t,x) &:= \E_{\mu_0} \left[ f_1(\omega_x^\e(T-t)) - \int_t^T f_2(s, \omega_x^\e(s-t)) \d s \right] \\
    % &= \int_{\R^d} f_1(x+\sqrt\e y) \rho_0(T-t, y) \d y - \int_t^T \int_{\R^d} f_2(s, x+\sqrt\e y) \rho_0(s-t, y) \d y \d s \\
    &= \int_{\R^d} f_1(y) \rho_0^\e(T-t, y-x) \d y - \int_t^T \int_{\R^d} f_2(s, y) \rho_0^\e(s-t, y-x) \d y \d s,
  \end{split}
\end{equation*}
where $\rho_0^\e(t, \cdot)$ is the Lebesgue density of $\sqrt\e W(t)$.
From Lemma \ref{lemma-heat-kernel-itg}, we see that $u \in C^{1,2}([0,T] \times \R^d)$ and satisfies
\begin{equation}\label{eqn-31}
  \pt_t u(t,x) = - \frac{\e}{2} \Delta u(t,x) + f_2(t,x), \quad u(T,x) = f_1(x).
\end{equation}
We then apply It\^o's formula to $u(t, \omega_x^\e(t))$, and get
\begin{equation*}
  \begin{split}
    \d u(t, \omega_x^\e(t)) &= \left( \pt_t u + \frac{\e}{2} \Delta u \right) (t, \omega_x^\e(t)) \d t + \sqrt\e \nabla u(t, \omega_x^\e(t)) \d \omega(t) \\
    &= f_2(t, \omega_x^\e(t)) \d t + \sqrt\e \nabla u(t, \omega_x^\e(t)) \d \omega(t).
  \end{split}
\end{equation*}
Then, we use the assumption and the equality $u(T,\cdot) = f_1(\cdot)$ to derive that, a.s.,
\begin{equation*}
  \begin{split}
    0 &\equiv u(T, \omega_x^\e(T)) - \int_0^T f_2(s, \omega_x^\e(s)) \d s \\
    &= u(0,x) + \sqrt\e \int_0^T \nabla u(s, \omega_x^\e(s)) \d \omega(s).
  \end{split}
\end{equation*}
Since $\int_0^\cdot \nabla u(s, \omega_x^\e(s)) \d \omega(s)$ is a martingale, it follows that $\nabla u \equiv 0$ and thus
\begin{equation*}
  u(t,x) \equiv F(t), \quad \forall (t,x)\in [0,T]\times \R^d,
\end{equation*}
for some $F\in C^1([0,T])$. The result follows by plugging the above identity back in \eqref{eqn-31}.
\end{proof}

\begin{rem}
\normalfont
(i). When $S^\e$ and $V$ are not explicitly time-dependent, the law of $X^\e_x$ is $\nu^\e_x$ if and only if $S^\e = g^\e$, and $V$ is time-independent and satisfies (up to a constant for $g^\e$)
\begin{equation}\label{time-ind-HJ}
  \frac{1}{2} |\nabla g^\e|^2 - \frac{\e}{2}\Delta g^\e = V.
\end{equation}
This was the case discussed in \cite{du2021graph} or \cite[Section 6.1]{dashti2013map}.

(ii). From the viewpoint of a stochastic version of the fundamental theorem of calculus, embodied by It\^o's formula, this lemma highlights a profound rigidity imposed by pathwise equalities. The given equation holds for almost every Brownian path, yet it contains no stochastic integral (martingale part). For such a pathwise identity to be possible, the application of Itô's formula to the terms \( f_1(x+\sqrt{\epsilon} \omega(T)) \) and \( \int_0^T f_2(s, x+\sqrt{\epsilon} \omega(s)) \d s \) must yield a vanishing martingale term. This forces the functions \( f_1 \) and \( f_2 \) to be degenerate: they cannot genuinely depend on the spatial variable \( x \). Consequently, the stochastic setting reduces to a deterministic one. 

(iii). Theorem \ref{Cor-sde} follows directly as a corollary of Lemma \ref{Cor-sde-lemma}-(iii), if the same assumptions on $V$ and $g$ are imposed.
\end{rem}

Compared to the usual 2nd-order HJ equation, equation \eqref{2nd-order HJ} has an extra initial constraint $S^\e(0,x) = -\e\log{Z^\e_{\Phi^\e}}(x)$. We shall see that this condition is naturally satisfied if we use Feynman--Kac representation \eqref{FK-Z} of $Z^\e_{\Phi^\e}$.

\begin{lem}\label{well-posedness-2nd-order HJ}
Suppose that $V\in C_b^{0,1}([0,T]\times\mathbb{R}^d)$ and $g^\e\in C_b(\R^d)$.
Then for every $\e>0$, equation \eqref{2nd-order HJ} has a unique classical solution $S^\e \in C_b^{1,2}([0,T]\times\mathbb{R}^d)$, which admits the following probabilistic representation: for all $(t,x)\in [0,T]\times \mathbb{R}^d$, 
\begin{equation}\label{HJB-prob}
  S^\e(t,x) = -\e \log  \mathbf{E}_{\mu_0} \left[ \exp\left\{-\frac{1}{\e} \int_t^T V(s,x+ \sqrt{\e} W(s-t)) \mathrm{d}s-\frac{1}{\e} g^\e(x+ \sqrt{\e} W(T-t)) \right\} \right].
\end{equation}
Moreover, if $V\in C_b^{0,2}([0,T]\times\mathbb{R}^d)$, then $S^\e \in C_b^{1,3}([0,T]\times\mathbb{R}^d)$.
\end{lem}

\begin{proof}
First, by taking the following Cole--Hopf transformation \cite{Schrodinger1926} (a.k.a. log transformation in \cite{fleming2006controlled})
\begin{equation*}
  S^\e(t,x) = -\e \log \phi^\e(t,x),  
\end{equation*}
we observe that the existence and uniqueness of solutions of equation \eqref{2nd-order HJ} in the space $C_b^{1,2}([0,T]\times\mathbb{R}^d)$ and those of the backward heat equation \eqref{Hea-equ} in $C_b^{1,2}([0,T]\times\mathbb{R}^d, \R_+)$ are equivalent.
As we shall see in Lemma \ref{lem-Born}, under the assumptions on $V$ and $g^\e$, equation \eqref{Hea-equ} has a unique solution $\phi^\e\in C_b^{1,2}([0,T]\times \mathbb{R}^d, \R_+)$, yielding the Feynman--Kac representation \eqref{FK}.
% \begin{equation}
%   \phi^\e(t,x) = \mathbf{E}_{\mu_0} \left[ \exp\left\{-\frac{1}{\e} \int_t^T V(s,x+ \sqrt{\e} W(s-t)) \mathrm{d}s  \right\}\phi^\e(T,x+ \sqrt{\e} W(T-t)) \right].
% \end{equation}
Thus, the function $S^\e$ defined in \eqref{HJB-prob} is in $C_b^{1,2}([0,T]\times \mathbb{R}^d)$ and satisfies equation \eqref{2nd-order HJ1}.
Recalling the representation \eqref{FK-Z} of the normalizing constant $Z^\e_{\Phi^\e}(x)$, we see that the initial condition $S^\e(0,x) = -\e\log{Z^\e_{\Phi^\e}(x)}$ of \eqref{2nd-order HJ} is fulfilled automatically. As in Remark \ref{improve-reg}, one can ask for $S^\e \in C^{1,3}$ by assuming $V\in C_b^{0,2}$. The result follows.
\end{proof}

Log transformation is, in fact, fundamental in the history of quantum mechanics: Schr\"odinger, Dirac, Feynman...

The following corollary of Theorem \ref{Cor-sde}, which generalizes SDE \eqref{SDE-2} to general initial data, is clear.

\begin{cor}\label{Cor-sde-initial}
Let Assumptions \ref{asmp-Phi-1} and \ref{asmp-Phi-cost} hold. Fix $\e>0$. Suppose that $\nu^\e|_{t=0}$ has full support in $\R^d$. 
Let $X^\e$, $B$, $(\Omega, \mathcal F, \mathbf P, \{\mathcal F_t\}_{t\in[0,T]})$ be a weak solution of the following SDE
\begin{align}\label{SDE-S-initial}
  \mathrm dX^\e(t)= -\nabla S^\e(t,X^\e(t))\mathrm d t + \sqrt\e \mathrm dB(t), \quad \operatorname{Law}(X^\e (0)) = \nu^\e|_{t=0},
\end{align}
where the potential function $S^\e\in C_b^{1,3}([0,T]\times\mathbb{R}^d)$ satisfies \eqref{Novikov-S}.
Suppose $V\in C^{0,1}_b([0,T]\times\mathbb{R}^d)$ and $g^\e\in C^1_b(\R^d)$. 
Then the law of $X^\e$ is $\nu^\e$ if and only if $S^\e$ is determined (up to a function depending only on time) by the following second-order Hamilton--Jacobi equation:
\begin{equation}\label{2nd-order HJ-full}
  \begin{cases}
    \partial_t S^\e(t,y)-\frac{1}{2} |\nabla S^\e(t,y)|^2 + \frac{\e}{2}\Delta S^\e(t,y) = -V(t,y), & (t,y)\in (0,T)\times\R^d, \\
    S^\e(T,y) =  g^\e(y), & y\in\R^d, \\
    S^\e(0,y) = -\e\log{Z^\e_{\Phi^\e}(y)}, & y\in\R^d.
  \end{cases}
\end{equation}
\end{cor}

Comparing the expressions of $\nu^\e_x$ in \eqref{nu-mu-Phi} and its time-reversed measure $\rev \nu^\e_x$ in \eqref{nu-mu-Phi-rev}, we find also the following corollary.

\begin{cor}\label{Cor-sde-initial-reversed}
Let Assumptions \ref{asmp-Phi-1} and \ref{asmp-Phi-cost} hold. Fix $\e>0$ and $x\in\R^d$. Let $\rev X^\e_x$, $\widetilde B$, $(\Omega, \mathcal F, \mathbf P, \{\rev{\mathcal F}_t\}_{t\in[0,T]})$ be a weak solution of the following SDE
\begin{align*}
  \mathrm d\rev X^\e_x(t)= -\nabla \widetilde S^\e(T-t,\rev X^\e_x(t))\mathrm d t + \sqrt\e \mathrm d\widetilde B(t), \quad \rev X^\e_x (0) =x,
\end{align*}
where the potential function $\widetilde S^\e\in C_b^{1,3}([0,T]\times\mathbb{R}^d)$ satisfies 
\begin{equation*}
   \E_{\mu_x^\e} \left[\exp \left(\frac{1}{2\e} \int_{0}^{T} |\nabla \widetilde S^\e(T-t, \omega(t))|^{2} \mathrm{d} t\right)\right]<\infty. 
  % \E_\P \left[\exp \left(\frac{1}{2\e} \int_{0}^{T} |\nabla \widetilde S^\e(T-t, \rev X^\e_x(t))|^{2} \mathrm{d} t\right)\right]<\infty.
\end{equation*}
Suppose $V\in C^{0,1}_b([0,T]\times\mathbb{R}^d)$ and $f^\e\in C^1_b(\R^d)$. 
Then the law of $\rev X^\e_x$ is $\rev \nu^\e_x$ if and only if $\widetilde S^\e$ is determined (up to a function depending only on time) by the following second-order Hamilton--Jacobi (2nd-order HJ) equation
\begin{equation}\label{2nd-order HJ-rev}
  \begin{cases}
     \partial_t \widetilde S^\e(t,y)+\frac{1}{2} |\nabla \widetilde S^\e(t,y)|^2 - \frac{\e}{2}\Delta \widetilde S^\e(t,y) = V(t,y), & (t,y)\in (0,T)\times\R^d, \\
    \widetilde S^\e(0,y) =  f^\e(y), & y\in\R^d, \\
    \widetilde S^\e(T,x) = -\e\log{Z^\e_{\Psi^\e}(x)}. &
  \end{cases}
\end{equation}
Moreover, equation \eqref{2nd-order HJ-rev} has a unique classical solution $\widetilde S^\e \in C_b^{1,2}([0,T]\times\mathbb{R}^d)$, which admits the following probabilistic representation: for all $(t,x)\in [0,T]\times \mathbb{R}^d$, 
\begin{equation*}
  \widetilde S^\e(t,x) = -\e \log  \E_{\mu_0} \left[ \exp \left(-\frac{1}{\e} \int_0^t V(s, x+\sqrt\e W(t-s)) \mathrm dr - \frac{1}{\e} f^\e(x+\sqrt\e W(t)) \right) \right].
\end{equation*}
\end{cor}

The law of $\rev X^\e_x$ is $\rev \nu^\e_x$ amounts to saying that the process $\rev X^\e_x$ can be regarded as the conditional process of $\rev X^\e$, the time-reversed process of $X^\e$ in \eqref{SDE-S-initial}, conditioned on $\{ \rev X^\e (0) =x \}$. It should not be confused with the time-reversed process of $X^\e_x$ in \eqref{SDE-2}.

\begin{rem}\normalfont
Here, in contrast with what is required in quantum mechanics, one considers two forward (usual) SDEs and, therefore, two different processes, $X$ and $\tilde{X}$, in order to follow the framework of stochastic thermodynamics. This will require a comment on the difference with the approach in~\cite{huang2023second}, for instance in the interpretation of Feynman's commutation relations. We briefly discuss this point in Section \ref{Conclusion}.    
\end{rem}

\section{Measure-theoretical study of path measures}\label{sec:potentialenergy}

In this section, we investigate the Onsager--Machlup functional, large deviation principles, time-reversal, and the Kullback--Leibler divergence from the perspective of path measures. Building on the correspondence between path measures and stochastic differential equations established in Section \ref{The setting and the main result}, this framework provides the foundation for the subsequent applications formulated in terms of stochastic differential equations.

\subsection{Onsager--Machlup functional}

The Onsager--Machlup functional is a tool used to describe the dynamics of stochastic processes, particularly in the context of nonequilibrium systems \cite{onsager1953fluctuations,machlup1953fluctuations}. It provides a way to quantify the probability of a given path taken by a stochastic process, with the most probable paths corresponding to those that minimize the functional.

We recall a derivation of the Onsager--Machlup (OM) functional from the problem of maximum a posteriori estimators, following the exposition in \cite{dashti2013map}. The Onsager--Machlup functional $\mathrm{OM}_{\Phi^\e}: \mathcal{C}^{d,T} \to \R$ associated with the functional $\Phi^\e$ is defined by
\begin{equation}\label{OM}
  \mathrm{OM}_{\Phi^\e}[\omega] := 
  \begin{cases}
  \frac{1}{2} \|\omega\|_{H_0^1}^2 + \Phi^\e(\omega), & \omega \in \mathcal{H}^{d,T}, \\
  \infty, & \omega \in \mathcal{C}^{d,T} \setminus \mathcal{H}^{d,T}.
  \end{cases}
\end{equation}

\begin{Assu}\label{asmp-Phi-2}  
% The functional $\Phi^\e: \mathcal{C}^{d,T} \to \mathbb{R}$ satisfies the following conditions:  

% (i) For every $\e > 0$, there exists an $M\in\R$, such that for all $\omega \in \mathcal{C}^{d,T}$,
% $$  
% \Phi^\e(\omega) \ge M - \e \|\omega\|_T^2.  
% $$

% (ii) $\Phi^\e$ is locally bounded from above, i.e., for every $r > 0$, there exists $K = K_r>0$ such that, for all $\omega \in \mathcal{C}^{d,T}$ with $\|\omega\|_T < r$,
% $$  
% \Phi^\e(\omega) \le K.  
% $$  
%(iii) 
For each $\e>0$, $\Phi^\e$ is locally Lipschitz continuous, i.e., for every $r > 0$, there exists $M = M(\e,r)>0$ such that, for all $\omega_1, \omega_2 \in \mathcal{C}^{d,T}$ with $\|\omega_1\|_T, \|\omega_2\|_T < r$,  
$$  
\left|\Phi^\e(\omega_1) - \Phi^\e(\omega_2)\right| \le M \|\omega_1 - \omega_2\|_T.  
$$  
\end{Assu}

For $\omega\in \mathcal{C}^{d,T}$, denote by $B_r(\omega) \subset \mathcal{C}^{d,T}$ the open ball centered at $\omega$ with radius $r>0$. The following characterization of small tube probabilities of $\nu^\e_x$ is adapted from \cite[Corollary 3.3]{dashti2013map}.

\begin{prop}
Under Assumptions \ref{asmp-Phi-1} and \ref{asmp-Phi-2}, we have, for any $\gamma \in \mathcal{H}^{d,T}_x$,
\begin{equation*}
  \lim_{r\to 0} \frac{\nu^\e_x (B_r(\gamma))}{\mu^\e_0 (B_r(0))} = \frac{1}{Z^\e_{\Phi^\e}(x)} \exp \left(-\frac{1}{\e} \mathrm{OM}_{\Phi^\e}[\gamma]\right).
\end{equation*}
\end{prop}

% Indeed, by definition \eqref{mu-tilde} and Assumption \ref{asmp-Phi-2},
% \begin{equation*}
%   \begin{split}
%     \nu^\e_x (B_r(\gamma)) &= \frac{1}{Z^\e_{\Phi^\e}(x)} \int_{B_r(\gamma)} \exp \left(-\frac{1}{\e} \Phi^\e(\omega)\right) \mu^\e_x (\d\omega) \\
%     &= \frac{1}{Z^\e_{\Phi^\e}(x)} \exp \left(-\frac{1}{\e} \Phi^\e(\gamma) \right) \int_{B_r(\gamma)} \exp \left(-\frac{1}{\e} (\Phi^\e(\omega) - \Phi^\e(\gamma)) \right) \mu^\e_x (\d\omega) \\
%     &\sim \frac{1}{Z^\e_{\Phi^\e}(x)} \exp \left(-\frac{1}{\e} \Phi^\e(\gamma) \right) e^{Lr/\e} \mu^\e_x (B_r(\gamma)), \quad \text{as } r\to 0.
%   \end{split}
% \end{equation*}
% Now we apply \cite[Section 18, Proposition 3]{lifshits2013gaussian}, it yields
% \begin{equation*}
%   \begin{split}
%     \mu^\e_x (B_r(\gamma)) = \mu_0 \left( \frac{B_r(0) + \gamma-x}{\sqrt\e} \right) &\sim \exp \left( - \frac{1}{2} \left\| \frac{\gamma-x}{\sqrt\e} \right\|_{H_0^1}^2 \right) \mu_0 \left( \frac{B_r(0)}{\sqrt\e} \right) \quad \text{as } r\to 0 \\
%     &= \exp \left( - \frac{1}{2\e} \|\gamma\|_{H_0^1}^2 \right) \mu^\e_0 (B_r(0)).
%   \end{split}
% \end{equation*}
% % Now we apply \cite[Theorem 9.1]{Ikeda1980}, it yields
% % \begin{equation*}
% %   \begin{split}
% %     \mu^\e_x (B_r(\gamma)) \sim \exp \left( - \frac{1}{\e} \mathrm{OM}_X[\omega] \right) \mu^\e_0 (B_r(0)), \quad \text{as } r\to 0.
% %   \end{split}
% % \end{equation*}
% Combining the last two equations, we get the desired result \eqref{OM-original}.

\subsubsection*{The standard Lagrangian of OM functional}
% \label{The standard Lagrangian of the OM functional}

Recall from \eqref{OM} that, once the functional $\Phi^\e$ has the representation \eqref{phi}, its associated OM functional $\mathrm{OM}_{\Phi^\e}$ take values at $\gamma \in \mathcal{H}^{d,T}$ as 
\begin{equation}\label{OM-2}
  \mathrm{OM}_{\Phi^\e} [\gamma] = \Phi^\e(\gamma) + \frac{1}{2} \|\dot \gamma\|_{L^2[0, T]}^2 = \int_0^T \left( \frac{1}{2} |\dot \gamma(t)|^2 + V(t,\gamma(t)) \right) \mathrm d t + g^\e(\gamma(T)).
\end{equation}
This indicates that $\mathrm{OM}_{\Phi^\e}$ can be regarded as the action functional with terminal cost $g^\e$ and the following standard Euclidean Lagrangian
\begin{equation}\label{Lim-lag}
L_V(t,x,\dot{x}):= \frac{1}{2}|\dot{x}|^2+V(t,x),
\end{equation}
The corresponding Hamiltonian is
\begin{equation}\label{Hamiltonian}
  H_V(x,p,t) = \frac{1}{2} |p|^2 - V(t,x).
\end{equation}

The stationary-action principle for the functional \eqref{OM-2} on $\mathcal{H}^{d,T}_x$ is to cancel out its variation, i.e., $\delta \mathrm{OM}_{\Phi^\e} [\gamma] = 0$. This can be implemented formally, using the fact that $\delta \dot \gamma = \frac{\mathrm d}{\mathrm dt} \delta \gamma$ (as the time parameter $t$ is not varied), as follows:
\begin{equation*}
  \begin{split}
    \delta \mathrm{OM}_{\Phi^\e} [\gamma] &= \int_0^T \left( \dot \gamma(t) \delta \dot \gamma(t) + \nabla V(t,\gamma(t)) \delta \gamma(t) \right) \mathrm d t + \nabla g^\e(\gamma(T)) \delta \gamma(T) \\
    &= \int_0^T \left( - \ddot \gamma(t) \delta \gamma(t) + \nabla V(t,\gamma(t)) \delta \gamma(t) \right) \mathrm d t + \left[ \dot \gamma(T) + \nabla g^\e(\gamma(T)) \right] \delta \gamma(T).
  \end{split}
\end{equation*}
Thus, the associated Euler--Lagrange (EL) equation is
\begin{equation}\label{eur-lag}
\begin{cases}
  \ddot \gamma(t) = \nabla V(t,\gamma(t)), \quad t \in(0, T), \\
  \gamma(0)=x, \quad \dot \gamma(T) = - \nabla g^\e(\gamma(T)). 
\end{cases}
\end{equation} 

\begin{rem}
\normalfont
The Lagrangian \eqref{Lim-lag} and Hamiltonian \eqref{Hamiltonian} are Euclidean quantum ones, where the signs in front of the potential $V$ are opposite to the classical ones. More relations with quantum mechanics, the main objective of Schr\"odinger's original observation in \cite{Schrodinger1926}, were discussed in \cite{huang2023gauge}.
% and Lagrangian is $L(t,x,\dot{x})=\frac{1}{2}\dot{x}^2 + V(t,x)$ and the Euler--Lagrange equation is
% \begin{align*}
% \ddot{x} + \partial_t g^\e(t,x) +\frac{1}{2}\partial_{xx} g^\e(t,x) =0. 
% \end{align*}
\end{rem}

\subsection{Large deviations}

We recall from the classical large deviation theory that the family $\{\mu_x^\e: \e>0\}$ satisfies the large deviation principle in $(\mathcal{C}^{d,T}_x, \mathcal B(\mathcal{C}^{d,T}_x))$ with the following good rate function (see, e.g., \cite[Theorem 5.2.3]{DZ98})
\begin{equation*}
  I(\omega) := 
  \begin{cases}
  \frac{1}{2}\|\omega\|_{H_0^1}^2,\quad & \omega\in \mathcal{H}^{d,T}_x,\\
  \infty, & \omega\in \mathcal{C}^{d,T}_x \setminus \mathcal{H}^{d,T}_x.
  \end{cases}
\end{equation*}

\begin{Assu}\label{asmp-Phi-3}
There exists a continuous function $\Phi^0: \mathcal{C}^{d,T} \rightarrow \mathbb{R}$ such that for every $x\in\R^d$, \\
(i) $\Phi^0 + I - \inf_{\omega \in \mathcal{C}^{d,T}_x} [\Phi^0(\omega) + I(\omega)]$ is a good rate function, \\
(ii) the tail condition holds,
\begin{equation}\label{tail-cond}
  \lim _{M \rightarrow \infty} \limsup _{\epsilon \rightarrow 0} \epsilon \log \E_{\mu_x^\e} \left[e^{\Phi^0 / \epsilon} \ind_{\left\{\Phi^0 \geq M\right\}}\right]=-\infty,
\end{equation}
(iii)
\begin{equation}\label{exp-lim-Phi}
  \lim_{\epsilon \rightarrow 0} \epsilon \log \E_{\mu_x^\e} \left[e^{(\Phi^0 - \Phi^\e) / \epsilon} \right] = 0.
\end{equation}
\end{Assu}

This assumption is somewhat technical; relevant remarks and sufficient conditions are provided in Appendix \ref{app-rmk}.

It follows from \eqref{OM} that Onsager--Machlup functional $\mathrm{OM}_{\Phi^0}: \mathcal{C}^{d,T} \to \R$ associated with the functional $\Phi^0$ is $\mathrm{OM}_{\Phi^0} = \Phi^0 + I$. We denote the good rate function of Assumption \ref{asmp-Phi-3}-(i) by
\begin{equation}\label{rate-func}
  I_{\Phi^0}^x := \mathrm{OM}_{\Phi^0} - \inf_{\omega \in \mathcal{C}^{d,T}_x} \mathrm{OM}_{\Phi^0}[\omega].
\end{equation}

The following result is a generalization of the tilted large deviation principle in \cite[Theorem III.17]{hollander2000large}. We also refer to \cite[Lemma 3.2]{selk2024smallnoise} for a set of assumptions concerning the Taylor expansion of $\Phi^\e$ with respect to $\e$.

\begin{prop}\label{tilted-large-deviation}
Let Assumptions \ref{asmp-Phi-1} and \ref{asmp-Phi-3} hold. 
%Assume that $\Phi^\e$ is bounded from below. 
For each $x\in \R^d$, the family $\{\nu^\e_x: \e>0\}$ satisfies the large deviation principle in $(\mathcal{C}^{d,T}_x, \mathcal B(\mathcal{C}^{d,T}_x))$, with the good rate function $I_{\Phi^0}^x$.
\end{prop}

\begin{proof}
We apply Varadhan's integral lemma \cite[Theorem 4.3.1]{DZ98} in view of the large deviation principle of $\{\mu_x^\e: \e>0\}$, and get for any bounded continuous function $F: \mathcal{C}^{d,T}_x \to \R$,
\begin{equation}\label{app-Varadhan}
  \lim_{\e \rightarrow 0} \e \log \E_{\mu_x^\e} \left[ \exp \left( -\frac{F + \Phi^0}{\e} \right) \right] = - \inf_{\omega \in \mathcal{C}^{d,T}_x}[F(\omega) + \Phi^0(\omega) + I(\omega)] = - \inf_{\omega \in \mathcal{C}^{d,T}_x} [F(\omega) + \mathrm{OM}_{\Phi^0}[\omega]].
\end{equation}
The applicability of Varadhan's integral lemma is ensured by the tail condition \eqref{tail-cond}.
Moreover, by applying H\"older's inequality and its reverse, we have for any $p>1$,
\begin{equation*}
  \begin{split}
    p\e \log \E_{\mu_x^\e} \left[ \exp \left( -\frac{F + \Phi^0}{p\e} \right) \right] &+ (1-p)\e \log \E_{\mu_x^\e} \left[ \exp \left( \frac{\Phi^0 - \Phi^\e}{(1-p)\e} \right) \right] \\
    &\le \e \log \E_{\mu_x^\e} \left[ \exp \left( -\frac{F + \Phi^\e}{\e} \right) \right] \\
    &\le \frac{\e}{p} \log \E_{\mu_x^\e} \left[ \exp \left( -\frac{p(F + \Phi^0)}{\e} \right) \right] + \frac{\e}{p'} \log \E_{\mu_x^\e} \left[ \exp \left( \frac{p'(\Phi^0 - \Phi^\e)}{\e} \right) \right],
  \end{split}
\end{equation*}
where $p'$ is the H\"older conjugate of $p$, i.e., $\frac{1}{p} + \frac{1}{p'} = 1$.
Now, we take the limit $\e\to 0$ of the above inequalities and use \eqref{exp-lim-Phi} and \eqref{app-Varadhan}, and obtain
\begin{equation}\label{ext-Varadhan}
  \lim_{\e \rightarrow 0} \e \log \E_{\mu_x^\e} \left[ \exp \left( -\frac{F + \Phi^\e}{\e} \right) \right] = \lim_{\e \rightarrow 0} \e \log \E_{\mu_x^\e} \left[ \exp \left( -\frac{F + \Phi^0}{\e} \right) \right] = - \inf_{\omega \in \mathcal{C}^{d,T}_x} [F(\omega) + \mathrm{OM}_{\Phi^0}[\omega]].
\end{equation}
Therefore,
\begin{equation*}
  \begin{split}
    \lim_{\e \rightarrow 0} \e \log \E_{\nu^\e_x} \left[ \exp \left( -\frac{F}{\e} \right) \right] &= \lim_{\e \rightarrow 0} \e \log \E_{\mu_x^\e} \left[ \exp \left( -\frac{F + \Phi^\e}{\e} \right) \right] - \lim_{\e \rightarrow 0} \e \log \E_{\mu_x^\e} \left[ \exp \left( -\frac{\Phi^\e}{\e} \right) \right] \\
    %&= - \inf_{\omega \in \mathcal{C}^{d,T}_x}[F(\omega) + \Phi^\e(\omega) + I(\omega)] + \inf_{\omega \in \mathcal{C}^{d,T}_x}[\Phi^\e(\omega) + I(\omega)] \\
    &= - \inf_{\omega \in \mathcal{C}^{d,T}_x}[F(\omega) + \mathrm{OM}_{\Phi^0}[\omega]] + \inf_{\omega \in \mathcal{C}^{d,T}_x} \mathrm{OM}_{\Phi^0}[\omega] \\
    &= - \inf_{\omega \in \mathcal{C}^{d,T}_x}[F(\omega) + I_{\Phi^0}^x(\omega)].
  \end{split}
\end{equation*}
It follows from Bryc's inverse Varadhan lemma \cite[Theorem 4.4.13]{DZ98} that the family $\{\nu^\e_x: \e>0\}$ satisfies the large deviation principle with the good rate function $I_{\Phi^0}^x$. 
\end{proof}

% \subsection{Normalizing constants}

As a byproduct of the above proof, we obtain the following asymptotics of the normalizing constants $Z^\e_{\Phi^\e}(x)$ as $\e \to 0$ by putting $F\equiv0$ in \eqref{ext-Varadhan}:
\begin{equation*}
  \lim_{\e \rightarrow 0} \e \log Z^\e_{\Phi^\e}(x) = \lim_{\e \rightarrow 0} \e \log \E_{\mu_x^\e} \left[ \exp \left( -\frac{\Phi^\e}{\e} \right) \right] %= - \inf_{\omega \in \mathcal{C}^{d,T}_x}[\Phi^\e(\omega) + I(\omega)] 
  = - \inf_{\omega \in \mathcal{C}^{d,T}_x} \mathrm{OM}_{\Phi^0}[\omega].
\end{equation*}

% As we shall see in \eqref{app-Varadhan} in the next subsection that, a direct application of Varadhan's integral lemma yields the following asymptotics of $Z^\e_{\Phi^\e}(x)$ as $\e \to 0$,
% \begin{equation*}
%   \lim_{\e \rightarrow 0} \e \log Z^\e_{\Phi^\e}(x) = - \inf_{\omega \in \mathcal{C}^{d,T}_x} \mathrm{OM}_{\Phi^\e}[\omega].
% \end{equation*}

\subsection{Kullback--Leibler divergence}

\begin{lem}\label{KL-nu-tilde}
Let Assumption \ref{asmp-Phi-1} hold. Fix an $\e>0$. Let $\tilde\nu^\e$ be a probability measure on $(\mathcal{C}^{d,T},\mathcal B(\mathcal{C}^{d,T}))$ that is absolutely continuous with respect to $\nu^\e$. Then 
\begin{equation*}
  D_{\mathrm{KL}} \left(\tilde\nu^\e \| \nu^\e\right) = D_{\mathrm{KL}}\left(\tilde\nu^\e|_{t=0} \| \mu^\e|_{t=0}\right) + \frac{1}{\e} \E_{\tilde\nu^\e} \left[ f^\e(\omega(0)) + \frac{1}{2} \int_{0}^{T} |b^\e(t, \omega)|^{2} \mathrm{d} t + \Phi^\e(\omega) \right],
\end{equation*}
where $b^\e$ is a progressively measurable process such that the triple $\omega(\cdot)$, $\widetilde B$, $(\mathcal{C}^{d,T}, \mathcal B(\mathcal{C}^{d,T}), \tilde\nu^\e, \{\mathcal B_t(\mathcal{C}^{d,T})\}_{t\in [0,T]})$ is a weak solution of the functional SDE
\begin{equation}\label{control-sde}
  \mathrm{d} \omega(t) = b^\e(t, \omega) \mathrm{d}t + \sqrt\e \mathrm{d} \widetilde B(t).
\end{equation}
\end{lem}

\begin{proof}
Recall from \eqref{mu0-scaling} that, the process $(W-x)/\sqrt\e$ is a standard Brownian motion under $\mu_x^\e$.
It follows from Lemma \ref{Gir-inv} that, by putting $(\Omega, \mathcal F, \{\mathcal F_t\}_{t\in[0,T]}) = (\mathcal{C}^{d,T},\mathcal B(\mathcal{C}^{d,T}), \{\mathcal B_t(\mathcal{C}^{d,T})\}_{t\in[0,T]})$, $\mathbf P$ as $\mu_x^\e$ and $B$ as $(W-x)/\sqrt\e$, if $\tilde\nu^\e_x \sim \mu_x^\e$, then there exists a progressively measurable process $\frac{1}{\sqrt\e} b^\e$ satisfying $\frac{1}{\e} \int_{0}^{T} |b^\e(t)|^{2} \mathrm{d} t < \infty$, $\mu_x^\e$-a.s., such that the process $\frac{1}{\sqrt\e} (W-x) - \frac{1}{\sqrt\e} \int_0^{\cdot} b^\e(s)\mathrm ds$ is a standard Brownian motion under $\tilde\nu^\e_x$, which we denote as $\widetilde B$. In other words, the triple $W$, $\widetilde B$, $(\mathcal{C}^{d,T}, \mathcal B(\mathcal{C}^{d,T}), \tilde\nu^\e_x, \{\mathcal B_t(\mathcal{C}^{d,T})\}_{t\in [0,T]})$ is a weak solution of the functional SDE:
\begin{equation*}
  \mathrm{d} W(t) = b^\e(t, W) \mathrm{d}t + \sqrt\e \mathrm{d} \widetilde B(t), \quad W(0) = x.
\end{equation*}
%that is, the canonical process $\omega(\cdot)$ has its generator $\frac{\e}{2} \Delta + b^\e(t) \cdot \nabla$ under $\tilde\nu^\e_x$ (cf. \cite[Section 1.2]{leonard2014survey}).
Moreover,
\begin{equation*}
  \begin{split}
    \frac{\mathrm{d} \tilde\nu^\e_x}{\mathrm{d} \mu^\e_x} (\omega) &= \exp \left( \frac{1}{\e} \int_0^T b^\e(t, \omega) \mathrm d\omega(t) - \frac{1}{2\e} \int_0^T |b^\e(t, \omega)|^2\mathrm  dt \right) \\
    &= \exp \left( \frac{1}{\sqrt\e} \int_0^T b^\e(t, \omega)\mathrm  d\widetilde B(t,\omega) + \frac{1}{2\e} \int_0^T |b^\e(t, \omega)|^2\mathrm  dt \right).
  \end{split}
\end{equation*}
Thus,
\begin{equation*}
  \begin{split}
    D_{\mathrm{KL}} \left(\tilde\nu^\e_x \| \nu^\e_x\right) &= \E_{\tilde\nu^\e_x} \left[ \log \left( \frac{\mathrm{d} \tilde\nu^\e_x}{\mathrm{d} \nu^\e_x} \right) \right] = \E_{\tilde\nu^\e_x} \left[ \log \left( \frac{\mathrm{d} \tilde\nu^\e_x}{\mathrm{d} \mu^\e_x} \right) - \log \left( \frac{\mathrm{d} \nu^\e_x}{\mathrm{d} \mu^\e_x} \right) \right] \\
    &= \E_{\tilde\nu^\e_x} \left[ \frac{1}{\sqrt\e} \int_{0}^{T} b^\e(t, \omega) \mathrm{d} \widetilde B(t,\omega) + \frac{1}{2\e} \int_{0}^{T} |b^\e(t, \omega)|^{2} \mathrm{d} t + \frac{1}{\e} \Phi^\e(\omega) \right] + \log Z^\e_{\Phi^\e}(x) \\
    &= \frac{1}{\e} \E_{\tilde\nu^\e_x} \left[ \frac{1}{2} \int_{0}^{T} |b^\e(t, \omega)|^{2} \mathrm{d} t + \Phi^\e(\omega) \right] + \log Z^\e_{\Phi^\e}(x).
  \end{split}
\end{equation*}
Now applying Lemma \ref{lemma-KL} and recalling the definition \eqref{f-def} of $f^\e$, we obtain
\begin{equation*}
  \begin{split}
    D_{\mathrm{KL}} \left(\tilde\nu^\e \| \nu^\e\right) &= D_{\mathrm{KL}}\left(\tilde\nu^\e|_{t=0} \| \nu^\e|_{t=0}\right) + \int_{\R^d} D_{\mathrm{KL}} \left(\tilde\nu^\e_x \| \nu^\e_x\right) \tilde\nu^\e|_{t=0}(\d x) \\
    &= \E_{\tilde\nu^\e|_{t=0}} \left[ \log \left( \frac{\mathrm{d} \tilde\nu^\e|_{t=0}}{\mathrm{d} \mu^\e|_{t=0}} \right) - \log \left( \frac{\mathrm{d} \nu^\e|_{t=0}}{\mathrm{d} \mu^\e|_{t=0}} \right) \right] + \frac{1}{\e} \E_{\tilde\nu^\e} \left[ \frac{1}{2} \int_{0}^{T} |b^\e(t, \omega)|^{2} \mathrm{d} t + \Phi^\e(\omega) \right] \\
    &\quad\ + \E_{\tilde\nu^\e|_{t=0}} \left( \log Z^\e_{\Phi^\e} \right) \\
    &= \E_{\tilde\nu^\e|_{t=0}} \left[ \log \left( \frac{\mathrm{d} \tilde\nu^\e|_{t=0}}{\mathrm{d} \mu^\e|_{t=0}} \right) \right] + \frac{1}{\e} \E_{\tilde\nu^\e} \left[ f^\e(\omega(0)) + \frac{1}{2} \int_{0}^{T} |b^\e(t, \omega)|^{2} \mathrm{d} t + \Phi^\e(\omega) \right].
  \end{split}
\end{equation*}
The result follows.
\end{proof}

% OM and FW action functionals on manifold
% Let $\{x(t):t\in[0,T]\}$ be a nondegenerate conservative diffusion process on a Riemannian manifold $(M,g)$. Assume that the diffusion $\{x(t):t\in[0,T]\}$ is associated with the drifted heat equation:
% \begin{align*}
% \frac{\partial u}{\partial t}=\frac{1}{2} \Delta u+f u,
% \end{align*}
% where $\Delta$ is the Laplace-Beltrami operator of $(M, g)$ and $f$ is a (smooth) vector field on M . By \cite{hara1996wiener,hara1996lagrangian}, the OM action functional associated to the diffusion process $\{x(t):t\in[0,T]\}$ is 
% \begin{align*}
% S(\gamma)=\int_0^T L(\gamma(t), \dot{\gamma}(t)) \mathrm d t,
% \end{align*}
% and
% \begin{align*}
% L(x, v)=\frac{1}{2}|v-f|_x^2+\frac{1}{2} \operatorname{div} f(x)+\frac{1}{12} R(x).    
% \end{align*}
% Here, $|\cdot|$ is the Riemannian norm on the tangent space $T_x M$ at $x$, div stands for the divergence and $R(x)$ is the scalar curvature at $x$.

\subsection{Time-reversals}\label{subsec:reversal}

We consider the simplest case where the initial time marginal of $\mu^\e$ is the scaled Lebesgue measure on $\R^d$, namely,
\begin{Assu}\label{mu-initial}
  For each $\e>0$, $\mu^\e|_{t=0}(\d x) = (\delta_\e)_* \d x = \d x/ \sqrt\e$.
\end{Assu}

This can lead to a reversible $\mu^\e$, since the time marginals of $\mu^\e$ are stationary:
\begin{equation}\label{mu-marginal}
  \mu^\e|_t(\d y) = \int_{\R^d} \mu^\e_x|_t(\d y) \mu^\e|_{t=0}(\d x) %= \frac{1}{\sqrt\e} \int_{\R^d} \rho_0(\e t, y - x) \d x \d y 
  = \frac{\d y}{\sqrt\e},
\end{equation}
and $\mu^\e$ satisfies the stationary detailed balance condition, i.e., the two-time marginals are symmetric: $\mu^\e|_{0,t}(\d x, \d y) = \mu^\e_x|_t(\d y) \mu^\e|_{t=0}(\d x) = \frac{1}{\sqrt\e} \rho_0(\e t, y - x) \d x \d y$.
In fact, when $\e=1$, $\mu^1 = \int_{\R^d} \mu_x(\cdot) \d x$ is the law of the so-called reversible Brownian motion, which is sometimes used as reference measure in the study of Schr\"odinger's problem \cite{Leo14b}. We show that $\mu^\e$ is the $\e$-scaling of $\mu^1$, as from \eqref{mu0-scaling},
\begin{equation}\label{mu-scaling}
  \begin{split}
    \mu^\e(\d\omega) &= \int_{\R^d} \mu^\e_x(\d\omega) \frac{\d x}{\sqrt\e} = \int_{\R^d} \mu_0\left( \frac{\d\omega - x}{\sqrt\e} \right) \frac{\d x}{\sqrt\e} \\
    &= \int_{\R^d} \mu_0\left( \frac{\d\omega}{\sqrt\e} -y \right) \d y = \int_{\R^d} (\delta_\e)_* \mu_y(\d\omega) \d y =  (\delta_\e)_* \mu^1(\d\omega).
  \end{split}
\end{equation}
% Finally, in this case, we have
% \begin{equation*}
%   f^\e(x) = \e\log Z^\e_{\Phi^\e}(x) - \e\log \sqrt\e \frac{\d\nu^\e|_{t=0}(x)}{\d x}.
% \end{equation*}
Therefore, for the time-reversed measures $\rev\mu^\e := R_* \mu^\e$, it is clear that (see also \cite{anderson1982reverse})
\begin{lem}\label{mu-rev-equal}
  Under Assumption \ref{mu-initial}, we have $\rev\mu^\e = \mu^\e$ for all $\e>0$.
\end{lem}

% This, combined with \eqref{mu-scaling} and the fact that the operators $R$ and $\delta_\e$ commute, leads to
% \begin{equation*}
%   \rev\mu^\e = R_* \mu^\e = R_* (\delta_\e)_* \mu^1 = (\delta_\e)_* R_* \mu^1 = (\delta_\e)_* \mu^1 = \mu^\e.
% \end{equation*}

We now reverse the time direction of the target measures $\nu^\e$, i.e., we consider their time-reversals $\rev\nu^\e := R_* \nu^\e$.
\begin{lem}
Under Assumptions \ref{asmp-Phi-1} and \ref{mu-initial}, we have for each $\e>0$,
\begin{equation}\label{RN-rev}
  \frac{\mathrm{d} \nu^\e}{\mathrm{d} \rev\nu^\e}(\omega) %= \frac{\d\nu^\e}{\d\mu^\e}(\omega) \Big/ \frac{\mathrm{d} \rev\nu^\e}{\mathrm{d} \mu^\e}(\omega) 
  = \exp \left\{ -\frac{1}{\e} [f^\e(\omega(0)) + \Phi^\e(\omega)] + \frac{1}{\e} [f^\e(\omega(T)) + \Phi^\e(R(\omega))] \right\}.
\end{equation}
\end{lem}

\begin{proof}
It can be derived from Lemma \ref{lem-0}-(i) and \eqref{nu-mu} that
\begin{equation}\label{nu-rev-mu}
  \begin{split}
    \frac{\mathrm{d} \rev\nu^\e}{\mathrm{d} \mu^\e}(\omega') &= \frac{\mathrm{d} \rev\nu^\e}{\mathrm{d} \rev\mu^\e}(\omega') = \frac{\mathrm{d} (R_*\nu^\e)}{\mathrm{d} (R_*\mu^\e)}(\omega') \\
    &= \E_{\mu^\e} \left[ \frac{\mathrm{d} \nu^\e}{\mathrm{d} \mu^\e}(\omega) \bigg| R(\omega)= \omega' \right] \\
    &= \E_{\mu^\e} \left[ \exp \left\{ -\frac{1}{\e} [f^\e(\omega(0)) + \Phi^\e(\omega)] \right\} \bigg| R(\omega)=\omega' \right] \\
    &= \E_{\mu^\e} \left[ \exp \left\{ -\frac{1}{\e} [f^\e(R \circ R \circ \omega(0)) + \Phi^\e(R \circ R \circ \omega)] \right\} \bigg| R(\omega)=\omega' \right] \\
    &= \exp \left\{-\frac{1}{\e} [f^\e(\omega'(T)) + \Phi^\e(R(\omega'))] \right\}.
  \end{split}
\end{equation}
The result follows from
\begin{equation*}
  \frac{\mathrm{d} \nu^\e}{\mathrm{d} \rev\nu^\e}(\omega) = \frac{\d\nu^\e}{\d\mu^\e}(\omega) \Big/ \frac{\mathrm{d} \rev\nu^\e}{\mathrm{d} \mu^\e}(\omega). 
  %= \exp \left\{ -\frac{1}{\e} [f^\e(\omega(0)) + \Phi^\e(\omega)] + \frac{1}{\e} [f^\e(\omega(T)) + \Phi^\e(R(\omega))] \right\}.
\end{equation*}
\end{proof}

When $\Phi^\e$, $\e>0$, are of the cost function form, the following formula holds for the time-reversed measures $\rev\nu^\e := R_* \nu^\e$.

\begin{cor}
Let Assumptions \ref{asmp-Phi-1}, \ref{asmp-Phi-cost} and \ref{mu-initial} hold. For each $\e>0$, the time-reversed measure $\rev\nu^\e$ is given by
\begin{equation}\label{nu-rev-mu-Phi}
  \frac{\mathrm{d} \rev\nu^\e}{\mathrm{d} \mu^\e}(\omega) = \exp \left\{-\frac{1}{\e} \left[ g^\e(\omega(0)) + \int_0^T V(T-t, \omega(t)) \mathrm dt +f^\e(\omega(T)) \right] \right\},
\end{equation}
and for $\mu^\e|_{t=0}$-a.s. $x\in\R^d$, the transition measure $\rev\nu^\e_x := \rev\nu^\e(\cdot | \omega(0)=x)$ is given by
\begin{equation}\label{nu-mu-Phi-rev}
  \frac{\mathrm{d} \rev\nu^\e_x}{\mathrm{d} \mu^\e_x}(\omega) = \frac{1}{Z^\e_{\Psi^\e}(x)} \exp \left(-\frac{1}{\e} \Psi^\e(\omega)\right),
\end{equation}
where the potential functional $\Psi^\e: \mathcal{C}^{d,T} \rightarrow \mathbb{R}$ is defined by
\begin{equation*}
  \Psi^\e(\omega) := \int_0^T V(T-t, \omega(t)) \mathrm dt +f^\e(\omega(T)).
\end{equation*}
and $Z^\e_{\Psi^\e}(x)$ is its normalizing constant
\begin{equation}\label{norm-const-2}
Z^\e_{\Psi^\e}(x) := \E_{\mu_x^\e} \left[ e^{-\frac{1}{\e}\Psi^\e} \right].
\end{equation}
\end{cor}

Here we indulge an abuse of notation: $\rev\nu^\e_x$ denotes the conditional measure $\rev\nu^\e(\cdot | \omega(0)=x)$, instead of the time-reversal of the conditional measure $\nu^\e_x$.

\begin{proof}
When $\Phi^\e$ takes the form \eqref{phi}, we have
\begin{equation*}
  \Phi^\e(R (\omega)) = \int_0^T V(t, \omega(T-t)) \mathrm dt +g^\e(\omega(0)) = \int_0^T V(T-t, \omega(t)) \mathrm dt +g^\e(\omega(0)).
\end{equation*}
Equation \eqref{nu-rev-mu-Phi} follows from \eqref{nu-rev-mu}. 
Applying Corollary \ref{cor-0}-(i) with $t=0$ and using \eqref{nu-rev-mu-Phi}, we get for $\mu^\e|_0$-a.s. $x \in \R^d$,
\begin{equation*}
  \begin{split}
    \frac{\d \rev\nu^\e|_{t=0}}{\d \mu^\e|_{t=0}} (x) &= \E_{\mu^\e(\cdot | \omega(0) = x)} \left( \frac{\mathrm{d} \rev\nu^\e}{\mathrm{d} \mu^\e} \right) \\
    &= \E_{\mu^\e} \left[ \exp \left(-\frac{1}{\e} \left[ g^\e(\omega(0)) + \int_0^T V(T-t, \omega(t)) \mathrm dt +f^\e(\omega(T)) \right] \right) \bigg| \omega(0) = x \right] \\
    &= e^{-\frac{1}{\e} g^\e(x)} \E_{\mu^\e_x} \left[ \exp \left(-\frac{1}{\e} \left[ \int_0^T V(T-t, \omega(t)) \mathrm dt +f^\e(\omega(T)) \right] \right) \right] \\
    &= e^{-\frac{1}{\e} g^\e(x)} Z^\e_{\Psi^\e}(x),
  \end{split}
\end{equation*}

Then, Corollary \ref{cor-0}-(ii) implies that for $\mu^\e$-a.s. $\omega \in \mathcal{C}^{d,T}$,
\begin{equation*}
  \frac{\mathrm{d} \rev\nu^\e}{\mathrm{d} \mu^\e}(\omega) = \frac{\d \rev\nu^\e|_{t=0}}{\d \mu^\e|_{t=0}} (\omega(0)) \frac{\mathrm{d} \rev\nu^\e_x}{\mathrm{d} \mu^\e_x}(\omega) \bigg|_{x = \omega(0)} = e^{-\frac{1}{\e} g^\e(\omega(0))} Z^\e_{\Psi^\e}(\omega(0)) \frac{\mathrm{d} \rev\nu^\e_x}{\mathrm{d} \mu^\e_x}(\omega) \bigg|_{x = \omega(0)},
\end{equation*}
which, compared with \eqref{nu-rev-mu-Phi}, yields \eqref{nu-mu-Phi-rev}.
\end{proof}

It also follows from \eqref{RN-rev} that
\begin{equation*}
  \frac{\mathrm{d} \nu^\e}{\mathrm{d} \rev\nu^\e}(\omega) = \exp \left\{ \frac{1}{\e} [f^\e(\omega(T)) - g^\e(\omega(T))] - \frac{1}{\e} [f^\e(\omega(0)) - g^\e(\omega(0))] + \frac{1}{\e} \int_0^T (V(T-t, \omega(t)) - V(t, \omega(t))) \mathrm dt \right\}.
\end{equation*}
Combining \eqref{nu-mu-Phi-rev} and \eqref{nu-mu-Phi}, we get
\begin{equation}\label{eqn:RNnu}
  \frac{\mathrm{d} \nu^\e_x}{\mathrm{d} \rev\nu^\e_x} (\omega) = \frac{Z^\e_{\Psi^\e}(x)}{Z^\e_{\Phi^\e}(x)} \exp \left\{ \frac{1}{\e} \left[ \int_0^T (V(T-t, \omega(t)) - V(t, \omega(t))) \mathrm dt + f^\e(\omega(T)) - g^\e(\omega(T)) \right] \right\}.
\end{equation}

An interesting special case is when $V$ is not explicitly time-dependent, then we have
\begin{align*}
  \frac{\mathrm{d} \nu^\e}{\mathrm{d} \rev\nu^\e}(\omega) = \exp \left\{ \frac{1}{\e} [f^\e(\omega(T)) - g^\e(\omega(T))] - \frac{1}{\e} [f^\e(\omega(0)) - g^\e(\omega(0))] \right\},
\end{align*}
and
\begin{equation}\label{RNnu1} 
  \frac{\mathrm{d} \nu^\e_x}{\mathrm{d} \rev\nu^\e_x} (\omega) = \frac{Z^\e_{\Psi^\e}(x)}{Z^\e_{\Phi^\e}(x)} \exp \left\{ \frac{1}{\e} \left[ f^\e(\omega(T)) - g^\e(\omega(T)) \right] \right\}.
\end{equation}
We observe that the r.h.s.'s of the above two equations only depend on the initial and terminal states of the path $\omega \in \mathcal{C}^{d,T}$, but not the whole trajectory. Such property is referred to as `path-independence'. It is
at the heart of Schr\"odinger's original (1931-32) observation \cite{Schrodinger1932}.

\subsection{Born-type formula for time marginals}\label{subsec-Born}

We now prove the following Born-type formula for the time marginals of $\nu^\e$.

\begin{lem}\label{lem-Born}
Let Assumptions \ref{asmp-Phi-1}, \ref{asmp-Phi-cost} and \ref{mu-initial} hold. Suppose that $V\in C_b^{0,1}([0,T]\times\mathbb{R}^d)$ and $f^\e, g^\e\in C_b(\R^d)$. Then for $\e>0$ and $t\in [0,T]$,
\begin{equation}\label{eqn:nu-Born}
  \nu^\e|_t(\mathrm{d} x) = \frac{1}{\sqrt\e} \phi^\e(t,x) \psi^\e(t,x) \d x,
\end{equation}
where $\phi^\e \in C^{1,2}([0,T]\times \R^d, \R_+)$ is the unique solution of the following backward heat equation
\begin{equation}\label{Hea-equ}
  \begin{cases}
    \e \partial_t\phi^\e(t,x) + \frac{\e^2}{2} \Delta\phi^\e(t,x) - V(t,x)\phi^\e(t,x)=0, & (t,x)\in [0,T)\times\R^d, \\
    \phi^\e(T,x) = e^{-\frac{1}{\e} g^\e(x)}, & x\in\R^d,
  \end{cases}
\end{equation}
and $\phi^\e \in C^{1,2}([0,T]\times \R^d, \R_+)$ is the unique solution of the following forward heat equation 
\begin{equation}\label{Hea-equ-2}
  \begin{cases}
    \e \partial_t\psi^\e(t,x) - \frac{\e^2}{2} \Delta\psi^\e(t,x) + V(t,x)\psi^\e(t,x) = 0, & (t,x)\in (0,T]\times\R^d, \\
    \psi^\e(0,x) = e^{-\frac{1}{\e} f^\e(x)}, & x\in\R^d.
  \end{cases}
\end{equation}
\end{lem}

\begin{proof}
By Corollary \ref{cor-0}-(i), \eqref{nu-mu}, \eqref{mu-scaling} and \eqref{reg-cond-push},
\begin{equation}\label{RN-marginal}
  \begin{split}
    \frac{\mathrm{d} \nu^\e|_t}{\mathrm{d} \mu^\e|_t}(x) &= \E_{\mu^\e(\cdot | \omega(t)=x)} \left( \frac{\mathrm{d} \nu^\e}{\mathrm{d} \mu^\e} \right) \\
    &= \E_{(\delta_\e)_* \mu^1(\cdot | \omega(t)=x)} \left[ \exp \left\{ -\frac{1}{\e} [f^\e(\omega(0)) + \Phi^\e(\omega)] \right\} \right] \\
    &= \E_{\mu^1(\cdot | \sqrt\e \omega(t)=x)} \left[ \exp \left\{ -\frac{1}{\e} \left[ f^\e(\sqrt\e \omega(0)) + \Phi^\e(\sqrt\e \omega) \right] \right\} \right].
  \end{split}
\end{equation}
When $\Phi^\e$ is of the form \eqref{phi}, we can use the Markov property of $\mu^1$ (e.g. \cite[Lemma 11.1]{Kallenberg2021FoundationsOM}) to derive
\begin{equation}\label{RN-marginal-2}
  \begin{split}
    \frac{\mathrm{d} \nu^\e|_t}{\mathrm{d} \mu^\e|_t}(x) &= \E_{\mu^1} \left[ \exp \left(- \frac{1}{\e} \int_t^T V(s, \sqrt\e \omega(s)) \mathrm ds - \frac{1}{\e} g^\e(\sqrt\e \omega(T)) \right) \bigg| \sqrt\e \omega(t) = x \right] \\
    &\quad \times \E_{\mu^1} \left[ \exp \left(-\frac{1}{\e} \int_0^t V(s, \sqrt\e \omega(s)) \mathrm ds - \frac{1}{\e} f^\e(\sqrt\e \omega(0)) \right) \bigg| \sqrt\e \omega(t) = x \right] \\
    &=: \phi^\e(t,x) \psi^\e(t,x).
  \end{split}
\end{equation}
As $\omega(\cdot)$ is a reversible Brownian motion under $\mu^1$, we use the properties of independence and stationary increments and obtain
\begin{equation}\label{FK}
  \begin{split}
    \phi^\e(t,x) &= \E_{\mu^1} \left[ \exp \left(- \frac{1}{\e} \int_t^T V(s, x+ \sqrt\e \omega(s-t)) \mathrm ds - \frac{1}{\e} g^\e(x+ \sqrt\e \omega(T-t)) \right) \bigg| \omega(0) = 0 \right] \\
    &= \E_{\mu_0} \left[ \exp \left(- \frac{1}{\e} \int_t^T V(s, x+ \sqrt\e \omega(s-t)) \mathrm ds - \frac{1}{\e} g^\e(x+ \sqrt\e \omega(T-t)) \right) \right].
  \end{split}
\end{equation}
By Feynman--Kac theory \cite[Chapter 1, Theorems 12 and 16]{friedman1964partial}, under the regularity assumptions on $V$ and $g^\e$, $\phi^\e$ is the unique solution of the backward heat equation \eqref{Hea-equ}.
For the function $\psi^\e: [0,T]\times \R^d$, we transform it using Lemma \ref{mu-rev-equal} and \eqref{reg-cond-push},
\begin{equation*}
  \begin{split}
    \psi^\e(T-t,x) &= \E_{\rev\mu^1} \left[ \exp \left(-\frac{1}{\e} \int_0^{T-t} V(s, \sqrt\e \omega(s)) \mathrm ds - \frac{1}{\e} f^\e(\sqrt\e \omega(0)) \right) \bigg| \sqrt\e \omega(T-t) = x \right] \\
    &= \E_{\mu^1} \left[ \exp \left(-\frac{1}{\e} \int_0^{T-t} V(s, \sqrt\e \omega(T-s)) \mathrm ds - \frac{1}{\e} f^\e(\sqrt\e \omega(T)) \right) \bigg| \sqrt\e \omega(t) = x \right] \\
    &= \E_{\mu^1} \left[ \exp \left(-\frac{1}{\e} \int_0^{T-t} V(s, \omega_x^\e(T-t-s)) \mathrm ds - \frac{1}{\e} f^\e(\omega_x^\e(T-t)) \right) \bigg| \omega(0) = 0 \right] \\
    &= \E_{\mu_0} \left[ \exp \left(-\frac{1}{\e} \int_t^T V(T-r, \omega_x^\e(r-t)) \mathrm dr - \frac{1}{\e} f^\e(\omega_x^\e(T-t)) \right) \right].
  \end{split}
\end{equation*}
For the same reason as $\phi^\e$, under the assumption on $f^\e$, we infer that the function $\psi^\e$ is the unique solution of 
\begin{equation*}
  \e \partial_t [\psi^\e(T-t,x)] + \frac{\e^2}{2} \Delta\psi^\e(T-t,x) - V(T-t,x)\psi^\e(T-t,x) = 0, \quad \psi^\e(0,x) = e^{-\frac{1}{\e} f^\e(x)},
\end{equation*}
which is the forward heat equation \eqref{Hea-equ-2}.
Combining \eqref{RN-marginal-2} with \eqref{mu-marginal}, we get the desired result.
\end{proof}

\begin{rem}\label{improve-reg}
\normalfont
One can improve the regularity of $\phi^\e$ and $\psi^\e$ to $C^{1,3}$ by imposing the strong condition $V\in C_b^{0,2}$. Cf. \cite[Chapter 1, Sections 4--6]{friedman1964partial}.
\end{rem}

One can extract from \eqref{RN-marginal-2} a system of equations for $g^\e$ and $f^\e$, by taking $t=0$ and $t=T$, as follows:
\begin{equation}\label{Sch-sys}
  \left\{
  \begin{aligned}
    e^{-\frac{1}{\e} f^\e(x)} \E_{\mu^\e} \left[ \exp \left(-\frac{1}{\e} \int_0^T V(s, \omega(s)) \mathrm ds - \frac{1}{\e} g^\e(\omega(T)) \right) \bigg| \omega(0) = x \right] &= \frac{\nu^\e|_{t=0}(\mathrm{d} x)}{\d x}, \\
    e^{-\frac{1}{\e} g^\e(x)} \E_{\mu^\e} \left[ \exp \left(-\frac{1}{\e} \int_0^T V(s, \omega(s)) \mathrm ds - \frac{1}{\e} f^\e(\omega(0)) \right) \bigg| \omega(T) = x \right] &= \frac{\nu^\e|_{t=T}(\mathrm{d} x)}{\d x}.
  \end{aligned}
  \right.
\end{equation}
This system has been referred to as Schr\"odinger's system, see \cite[Theorem 2.4]{leonard2014survey} or \cite[Eqs.~(3.17), (3.18)]{jamison1974reciprocal}. Note that the Schr\"{o}dinger system can admit non-uniqueness up to a constant. Indeed, the system still holds after adding a constant to $f^\e$ and subtracting $g^\e$ by the same constant. 

The most general proof of existence and uniqueness of its solution $\{e^{-\frac{1}{\e}f^\e}, e^{-\frac{1}{\e}g^\e}\}$, not necessarily integrable is due to Beurling when the right hand side of equation \eqref{Sch-sys} is strictly positive \cite{Beurling1960}. A general class of potentials $V$ in equation \eqref{Sch-sys} is the one of Kato. For more about that, cf. \cite{CruzeiroZambrini1991} and, in a more entropic perspective \cite{CruzeiroFollmerZambrini2006}.

The project to construct diffusion processes from the data of two probability densities $\nu^\e|_{t=0}(dx)/dx$ and $\nu^\e|_{t=T}(dx)/dx$ is due to Schr\"odinger \cite{Schrodinger1932}. Only when the Born-type form \eqref{eqn:nu-Born} is required, the processes are Markovian. In this Markovian case, the project was (informally) realized in \cite{Zambrini1986}, using insights of S. Bernstein, B. Jamison and A. Beurling. The diffusions have also been called Bernstein reciprocal processes, because this author gave an early \cite{Bernstein1932} informal account of their properties, including the fact that, in general, they form a class larger than the Markovian one. Indeed, since Euclidean quantum field theory of the seventies, they are called one-dimensional Markov random fields on $[0,T]$.

Schr\"odinger wanted to find a statistical mechanical analogy with Born's interpretation of the quantum wave function: $\varphi\overline{\varphi}dx=\|\varphi\|_{2}^2dx$ in Hilbert space, but involving well-defined probability measures. Notice that $\phi^\e$ and $\psi^\e$ in \eqref{Hea-equ}-\eqref{Hea-equ-2} are Euclidean counterparts of two, generally unrelated quantum states in $L^2(\mathbb{R}^d)$ associated with the same Hamiltonian operator of equations \eqref{Hea-equ}-\eqref{Hea-equ-2} since the r.h.s. probabilities of \eqref{Sch-sys} were arbitrarily given in Schr\"odinger's problem. So the above-mentioned Born product form is in fact a counterpart of any $L^2$ scalar product, which is given here a probabilistic meaning, in strong contrast with quantum theory.

On the other hand, the two boundary ``states'' $\{e^{-\frac{1}{\e}f^\e}, e^{-\frac{1}{\e}g^\e}\}$ of \eqref{Sch-sys} and action functional $S^\e$ \eqref{HJB-prob} correspond to a nonlinear transformation of the form state~$=e^{-\frac{1}{\e}S^\e(x)}$ which, up to a factor $i=\sqrt{-1}$ in front of the action, has been fundamental in all historical approaches to the theory (Schr\"odinger \cite{Schrodinger1926}, Dirac, Feynman...).
The same nonlinear transformation has also been found useful in stochastic optimal control \cite{fleming2006controlled} and lies as well behind the stochastic version of classical action functional of \cite{CruzeiroZambrini1991,huang2023second}, less dependent of Feynman--Kac formula than the version given here. We shall come back to this in Subsection \ref{subsec:SEL}.

\subsubsection*{Normalizing constants}

Recall that we obtained a pair of heat equations, \eqref{Hea-equ} and \eqref{Hea-equ-2}, both of which yield Feynman--Kac representations, as shown in the proof of Lemma \ref{lem-Born}. It is now easy to show that the normalizing constants $Z^\e_{\Phi^\e}(x)$ in \eqref{norm-const} and $Z^\e_{\Psi^\e}(x)$ in \eqref{norm-const-2} can also have Feynman--Kac representations, when $\Phi^\e$ is of the cost function form \eqref{phi}.

\begin{cor}
Under the assumptions of Lemma \ref{lem-Born}, we have for each $\e>0$ and $x\in\R^d$,
\begin{align}
  Z^\e_{\Phi^\e}(x) = \phi^\e(0,x) &= \mathbf{E}_{\mu_0} \left[ \exp\left\{-\frac{1}{\e} \int_0^T V(s,\omega_x^\e(s) )\mathrm{d}s - \frac{1}{\e} g^\e(\omega_x^\e(T))  \right\} \right], \label{FK-Z} \\
  Z^\e_{\Psi^\e}(x) = \psi^\e(T,x) &= \mathbf{E}_{\mu_0} \left[ \exp\left\{-\frac{1}{\e} \int_0^T V(T-s,\omega_x^\e(s) )\mathrm{d}s - \frac{1}{\e} f^\e(\omega_x^\e(T))  \right\} \right]. \label{FK-Z-star}
\end{align}
\end{cor}

\begin{proof}
Since $\mu_x^\e$ is the transition measure of $\mu^\e$ as in \eqref{mu-tilde-total}, we derive in the same way as \eqref{RN-marginal} that
\begin{equation*}
  \begin{split}
    Z^\e_{\Phi^\e}(x) &= \E_{\mu^\e(\cdot | \omega(0)=x)} \left[ \exp \left( -\frac{1}{\e}\Phi^\e (\omega) \right) \right] = \E_{\mu^1(\cdot | \sqrt\e \omega(0)=x)} \left[ \exp \left( -\frac{1}{\e}\Phi^\e (\sqrt\e \omega) \right) \right] \\
    &= \E_{\mu_0} \left[ \exp \left( -\frac{1}{\e}\Phi^\e (x+ \sqrt\e \omega) \right) \right].
  \end{split}
\end{equation*}
By plugging the expression \eqref{phi} of $\Phi^\e$ into the above equation and using the Feynman--Kac representation \eqref{FK} of $\phi^\e$, we obtain \eqref{FK-Z}. Equation \eqref{FK-Z-star} follows in a similar fashion.
\end{proof}

Combining equations \eqref{RNnu1}, \eqref{FK-Z} and \eqref{FK-Z-star}, when both $V$ and $\phi^\e$ are time-independent, we have
\begin{equation}\label{eqn-19}
  \begin{split}
    \frac{\mathrm{d} \nu^\e_x}{\mathrm{d} \rev\nu^\e_x} (\omega) &= \frac{\psi^\e(T,x)}{\phi^\e(x)} \frac{\phi^\e(\omega(T))}{\psi^\e(0,\omega(T))} = \frac{\rho^\e(T,x)}{\phi^\e(x)^2} \frac{\phi^\e(\omega(T))^2}{\rho^\e(0,\omega(T))} \\
    &= \frac{\rho^\e(T,x)}{\rho^\e(0,\omega(T))} \exp \left\{ \frac{2}{\e} \left[ g^\e(x) - g^\e(\omega(T)) \right] \right\},
  \end{split}
\end{equation}
where $\rho^\e(t,x) = \frac{\d \nu^\e|_t(x)}{\d x}$ is the time marginal density of $\nu^\e$.

% \subsection{KL divergence}

% In Lemma \ref{KL-nu-tilde},
% \begin{equation*}
%   D_{\mathrm{KL}} \left(\tilde\nu^\e_x \| \nu^\e_x\right) = \frac{1}{\e} \E_{\tilde\nu^\e_x} \left[ \int_{0}^{T} \left( \frac{1}{2} |b^\e(t, \omega)|^{2} + V(t, \omega(t)) \right) \mathrm dt + g^\e(\omega(T)) \right] + \log Z^\e_{\Phi^\e}(x).
% \end{equation*}

\section{Application I: Onsager--Machlup functional and large deviations}

With the representation \eqref{phi} of $\Phi^\e$, Assumption \ref{asmp-Phi-3} reduces to

\begin{Assu}\label{asmp-Phi-3-1}
There exists a continuous function $g^0: \R^d \to \R$, such that the following functional 
\begin{equation}\label{phi0}
  \Phi^0(\omega) = \int_0^T V(t, \omega(t)) \mathrm dt +g^0(\omega(T)),
\end{equation}
satisfies Assumption \ref{asmp-Phi-3}-(i) and (ii), and for every $x\in\R^d$,
\begin{equation*}%\label{exp-lim-g}
  \lim_{\epsilon \rightarrow 0} \epsilon \log \E_{\mu_x^\e} \left[e^{(g^0(\omega(T)) - g^\e(\omega(T))) / \epsilon} \right] = 0.
\end{equation*}
\end{Assu}

The following corollary is a straightforward consequence of Proposition \ref{tilted-large-deviation} and Theorem \ref{Cor-sde}. 

\begin{cor}
Let Assumptions \ref{asmp-Phi-1}, \ref{asmp-Phi-cost} and \ref{asmp-Phi-3-1} hold. Let $X^\e_x$, $B$, $(\Omega, \mathcal F, \mathbf P, \{\mathcal F_t\}_{t\in[0,T]})$ be a weak solution of the SDE \eqref{SDE-2}, where $S^\e\in C^{1,2} ([0,T]\times\mathbb{R}^d)$ satisfies condition \eqref{Novikov-S} and the 2nd-order HJ equation \eqref{2nd-order HJ}. Then the family $\{X^\e_x: \e>0\}$ satisfies the large deviation principle in $\mathcal{C}^{d,T}_x$, with the rate function $I_{\Phi^0}^x$ of \eqref{rate-func} where $\Phi^0$ admits the representation \eqref{phi0}.
\end{cor}

We will give a more specific large deviation result for the solutions of SDE \eqref{SDE-2}, via the classical Freidlin--Wentzell theory. We first recall the probabilistic representation \eqref{HJB-prob}.
Define $\Phi^\e_{t,x}(\omega)=\int_t^T V(s,x+\omega(s-t)) \mathrm{d}s+ g^\e(x+\omega(T-t))$, $\omega \in \mathcal{C}^{d,T-t}_0$, $\e\ge 0$. Then
\begin{equation*}
  S^\e(t,x) = -\e \log  \mathbf{E}_{\mu^\e_x} \left[ \exp\left(-\frac{1}{\e} \Phi^\e_{t,x} \right) \right], \quad \e>0.
\end{equation*}
As in \eqref{ext-Varadhan}, we apply Varadhan's lemma and obtain the following limit
\begin{align*}
  S^0(t,x) := \lim_{\e\to0}S^\e(t,x) =   \inf_{\gamma \in \mathcal{H}^{d,T-t}_0} \left\{ \Phi^0_{t,x}(\gamma) + \frac{1}{2} \|\gamma\|_{H_0^1}^2 \right\}. 
\end{align*}
Note that $\Phi^\e_{0,x}= \Phi^\e \circ T_x$ and 
\begin{equation*}
  S^0(0,x) =  \inf_{\gamma \in \mathcal{H}^{d,T}_0} \left\{ \Phi^0(x+\gamma) + \frac{1}{2} \|\gamma\|_{H_0^1}^2 \right\} =  \inf_{\gamma\in \mathcal{H}^{d,T}_x} \left\{ \Phi^0(\gamma) + \frac{1}{2} \|\gamma\|_{H_0^1}^2 \right\} =  \inf_{\omega \in \mathcal{C}^{d,T}_x} \operatorname{OM}_{\Phi^0}[\omega].
\end{equation*}
Furthermore, by taking the limit $\e\to 0$ in the 2nd-order HJ equation \eqref{2nd-order HJ}, we see that $S^0$ formally satisfies the following classical Hamilton--Jacobi equation 
\begin{equation}\label{HJ}
  \begin{cases}
    \partial_t S^0(t,y)-\frac{1}{2} |\nabla S^0(t,y)|^2 = -V(t,y), & (t,y)\in (0,T)\times\R^d, \\
    S^0(T,y) =  g^0(y), & y\in\R^d, \\
    S^0(0,x) =  \inf_{\omega \in \mathcal{C}^{d,T}_x}\operatorname{OM}_{\Phi^0}[\omega]. &
  \end{cases}
\end{equation}

Next, we consider the following family of SDEs
\begin{align}\label{SDE-3}
\mathrm dX^{\e,0}_x(t)= -\nabla S^{0}(t,X^{\e,0}_x(t))\mathrm d t + \sqrt\e \mathrm dB(t), \quad X^{\e,0}_x (0) =x,
\end{align}
where $S^0$ satisfies HJ equation \eqref{HJ}. The Freidelin--Wentzell large deviation theory asserts that $\{X^{\e,0}_x: \e>0\}$ satisfy the large deviation principle with the good rate function 
\begin{align*}
  I_0^x(\omega) = 
  \begin{cases}
    \frac{1}{2}\int_0^T|\dot \omega(t) + \nabla S^{0}(t,\omega(t))|^2 \mathrm{d}t, & \omega\in \mathcal{H}^{d,T}_x, \\
    \infty, & \omega\in \mathcal{C}^{d,T}_x \setminus \mathcal{H}^{d,T}_x.
  \end{cases}
\end{align*}
Using Hamilton--Jacobi equation \eqref{HJ}, we obtain that for $\gamma\in \mathcal{H}^{d,T}_x$,
\begin{align*}
I_0^x(\gamma) &  =\int_0^T\left(\frac{1}{2}|\dot \gamma(t)|^2 + \nabla S^{0}(t,\gamma(t))\cdot \dot \gamma(t) + \frac{1}{2} |\nabla S^{0}(t,\gamma(t))|^2\right)\mathrm{d}t \\
&=\int_0^T\left(\frac{1}{2}|\dot \gamma(t)|^2 - \frac{\pt}{\pt t} S^{0}(t,\gamma(t)) + \frac{1}{2} |\nabla S^{0}(t,\gamma(t))|^2\right)\mathrm{d}t + S^{0}(T,\gamma(T)) - S^{0}(0,x) \\
&=\int_0^T\left(\frac{1}{2}|\dot \gamma(t)|^2 + V(t,\gamma(t)) \right) \mathrm{d}t + g^0(\gamma(T)) - \inf_{\omega\in \mathcal{C}^{d,T}_x} \operatorname{OM}_{\Phi^0}[\omega] \\
&=\operatorname{OM}_{\Phi^0}[\gamma] - \inf_{\omega\in \mathcal{C}^{d,T}_x} \operatorname{OM}_{\Phi^0}[\omega].  
\end{align*}
This means that the rate function $I_0^x$ coincides with $I_{\Phi^0}^x$ in \eqref{rate-func}.

Now, classical large deviation theory \cite[Theorem 4.2.13]{DZ98} tells us that, if the family $\{X^{\e,0}_x: \e>0\}$ is exponentially equivalent to $\{X^\e_x: \e>0\}$ given by \eqref{SDE-2}, as will be shown in the next lemma, then the LDP with the same rate function $I_0^x = I_{\Phi^0}^x$ holds for $\{X^\e_x:\e>0\}$. 

\begin{lem}
Let $X^\e_x$ and $X^{\e,0}_x$ be the unique solution of SDEs \eqref{SDE-2} and \eqref{SDE-3} respectively, where $S^\e$ and $S^0$ satisfy equations \eqref{2nd-order HJ} and \eqref{HJ} respectively. Suppose that 
the family $\{S^\e: 0<\e \ll 1\}$ is uniformly bounded and 
$S^\e$ converges to $S^0$ in $C^1$ norm on any compact set of $\mathbb{R}^d$, as $\e\to0$. Then the two families $\{X^\e_x: \e>0\}$ and $\{X^{\e,0}_x: \e>0\}$ are exponentially equivalent, that is,
\begin{align*}
\limsup_{\e\to0} \e\log \mathbf P \left( \left\| X^\e_x-X^{\e,0}_x \right\|_T > \delta \right) = -\infty.
\end{align*}
\end{lem}

\begin{proof}
Fix $t\in[0,T]$ and let $e(t):= |X^\e_x(t)-X^{\e,0}_x(t)|$. Since the family $\{X^{\e,0}_x\}$ satisfies an LDP, it is exponentially tight, i.e., for any $\alpha<\infty$, there exists a compact set $K_{\alpha}\subset\mathbb{R}^d$, such that
\begin{align}\label{exp-tight}
\limsup_{\e\to0} \e\log \mathbf{P}(X^{\e,0}_x\in K_\alpha^c) < -\alpha.    
\end{align} 
On the event $\{ X^{\e,0}_x\in K_\alpha \}$, we have
\begin{align*}
e(t)&= \left|\int_0^t\nabla S^{\e}(s,X^\e_x(s))\mathrm{d}s-\int_0^t\nabla S^0(s,X^{\e,0}_x(s))\mathrm{d}s\right| \\
&\leq \int_0^t\left|\nabla S^{\e}(s,X^\e_x(s))-\nabla S^\e(s,X^{\e,0}_x(s))\right|\mathrm{d}s + \int_0^t\left|\nabla S^\e(s,X^{\e,0}_x(s))-\nabla S^{0}(s,X^{\e,0}_x(s))\right|\mathrm{d}s \\
&\leq \|S^\e\|_{C^2 }\int_0^te(s)\mathrm{d}s + \|S^{\e}-S^0\|_{C^1(K_\alpha)} t.
\end{align*}
Then, by Gronwall's lemma, we obtain $e(t)\leq \|S^{\e}-S^0\|_{C^1(K_\alpha)} t\exp\{\|S^\e\|_{C^2}t\} $. Consequently, as $\{S^\e: 0<\e \ll 1\}$ is uniformly bounded,
\begin{align*}
\|X^\e_x-X^{\e,0}_x\|_T \leq \|S^{\e}-S^0\|_{C^1(K_\alpha)} Te^{MT}, \quad \text{on } \{ X^{\e,0}_x\in K_\alpha \},
\end{align*}
for some $M>0$ and all $0<\e \ll 1$. Since $S^\e$ converges to $S^0$ in $C^2$ norm on $K_{\alpha}$, for any $\delta>0$, there exists $\e_0>0$ such that for all $\e\leq\e_0$, $\|S^{\epsilon}-S^0\|_{C^1(K_\alpha)} <\frac{\delta}{Te^{MT}}$. 
Thus, for all $\e\leq\e_0$,
\begin{equation*}
\begin{aligned}
  \mathbf P \left( \left\|X^\e_x-X^{\e,0}_x\right\|_T > \delta \right) &= \mathbf P \left( \left\|X^\e_x-X^{\e,0}_x\right\|_T > \delta; X^{\e,0}_x\in K_\alpha\right) +  \mathbf P \left( \left\|X^\e_x-X^{\e,0}_x\right\|_T > \delta;X^{\e,0}_x\in K_\alpha^c \right) \\
  &\le \mathbf P \left( \|S^{\e}-S^0\|_{C^1(K_\alpha)} Te^{MT} > \delta \right) +  \mathbf P \left( X^{\e,0}_x\in K_\alpha^c \right) \\
  &= \mathbf P \left( X^{\e,0}_x\in K_\alpha^c \right).
\end{aligned}
\end{equation*}
Taking $\limsup_{\e\to0} \e\log$ to both sides and using \eqref{exp-tight}, we get
\begin{equation*}
  \limsup_{\e\to0} \e\log \mathbf P \left( \left\|X^\e_x-X^{\e,0}_x\right\|_T > \delta \right) < -\alpha.
\end{equation*}
The result follows from the arbitrariness of $\alpha$.
\end{proof}

\section{Application II: Entropy minimization problems}

The inference principle of minimizing the KL divergence $D_{\mathrm{KL}} \left(\tilde\nu^\e \| \nu^\e\right)$, due to Kullback, is known as the principle of minimum discrimination information. A closely related quantity, the relative entropy, is usually defined as the negative the Kullback--Leibler divergence. The principle of maximum entropy states that the probability distribution which best represents our current state of knowledge about a system is the one with largest entropy.

\subsection{Equivalence with stochastic optimal control problems}

The following is a straightforward corollary of Lemma \ref{KL-nu-tilde}.

\begin{cor}\label{equi-opm}
Let Assumption \ref{asmp-Phi-1} hold. Let $\widetilde{\mathcal P}$ be a subset of $\mathcal P$. Let $\nu^\e$ be the measure defined in \eqref{nu-mu}. Then the entropy minimization problem
\begin{equation}\label{entropy-min-prob}
  \inf_{\tilde\nu^\e \in \widetilde{\mathcal P}} \e D_{\mathrm{KL}} \left(\tilde\nu^\e \| \nu^\e \right)
\end{equation}
is equivalent to the following stochastic optimal control problem:
\begin{equation}\label{sto-opt-control}
  \inf_{\tilde\nu^\e \in \widetilde{\mathcal P}} \E_{\tilde\nu^\e} \left[ \e \left( \log Z^\e_{\Phi^\e} + \log \frac{\mathrm{d} \tilde\nu^\e|_{t=0}}{\mathrm{d} \nu^\e|_{t=0}} \right) (\omega(0)) + \frac{1}{2} \int_{0}^{T} |b^\e(t, \omega)|^{2} \mathrm dt + \Phi^\e(\omega) \right],
\end{equation}
where $b^\e$ is a progressively measurable process such that the triple $\omega(\cdot)$, $\widetilde B$, $(\mathcal{C}^{d,T}, \mathcal B(\mathcal{C}^{d,T}), \tilde\nu^\e, \{\mathcal B_t(\mathcal{C}^{d,T})\}_{t\in [0,T]})$ is a weak solution of the functional SDE \eqref{control-sde}.
\end{cor}

The set $\widetilde{\mathcal P}$ plays the role of constraints in the entropy minimization problem \eqref{entropy-min-prob}. Since the KL divergence $D_{\mathrm{KL}} \left(\tilde\nu^\e \| \nu^\e \right)$ is disintegrable as in Lemma \ref{KL-nu-tilde}, if the constraints implied by set $\widetilde{\mathcal P}$ is also disintegrable into initial distributions and transition probabilities, we can first separate \eqref{entropy-min-prob} into two minimization problems: one is to minimize the KL divergence of initial distributions $D_{\mathrm{KL}} \left(\tilde\nu^\e|_{t=0} \| \nu^\e|_{t=0} \right)$, the other is to minimize that of transition probabilities $D_{\mathrm{KL}} \left(\tilde\nu^\e_x \| \nu^\e_x \right)$.

The stochastic optimal control problem \eqref{sto-opt-control} can be reformulated into a more familiar form. That the triple $W$, $\widetilde B$, $(\mathcal{C}^{d,T}, \mathcal B(\mathcal{C}^{d,T}), \tilde\nu^\e, \{\mathcal B_t(\mathcal{C}^{d,T})\}_{t\in [0,T]})$ is a weak solution of the functional SDE \eqref{control-sde} is equivalent to saying that $\tilde\nu^\e$ can be realized as the law of $X^\e$ where the triple $X^\e$, $B$, $(\Omega, \mathcal F, \mathbf P, \{\mathcal F_t\}_{t\in[0,T]})$ is a weak solution of the functional SDE
\begin{equation}\label{SDE-opt-cont}
  \mathrm dX^\e(t)= b^\e(t,X^\e)\mathrm d t + \sqrt\e \mathrm dB(t), \quad X^\e (0) \sim \tilde\nu^\e|_{t=0}.
\end{equation}
The stochastic optimal control problem \eqref{sto-opt-control} now turns into
\begin{equation}\label{SDE-opt-cont-min}
  \inf_{\operatorname{Law} X^\e \in \widetilde{\mathcal P}} \E_{\P} \left[ \e \left( \log Z^\e_{\Phi^\e} + \log \frac{\mathrm{d} \tilde\nu^\e|_{t=0}}{\mathrm{d} \nu^\e|_{t=0}} \right) (X^\e(0)) + \frac{1}{2} \int_{0}^{T} |b^\e(t, X^\e)|^{2} \mathrm dt + \Phi^\e(X^\e) \right].
\end{equation}
Such reformulation offers a door to make use of our SDE correspondence results of Section \ref{SDE correspondence}.

\begin{rem}[Information projection: Cameron--Martin constraints]
\normalfont
Denote by $\mathcal{P}_x^\e \subset \mathcal{P}$ the set of all shift measures that are absolutely continuous with respect to $\mu_x^\e$, i.e.,
\begin{equation*}
  \mathcal{P}_x^\e := \left\{ (T_\gamma)_*\mu_x^\e: \gamma \in \mathcal{H}^{d,T}_0 \right\}.
\end{equation*}
Take $\widetilde{\mathcal P} = \mathcal{P}_x^\e$ in Corollary \ref{equi-opm}. For a measure $\tilde\nu^\e_x \in \mathcal{P}_x^\e$, its realization SDE \eqref{control-sde} must have a deterministic drift, more precisely, $b^\e(t, \omega) = \dot \gamma(t)$ for some $\gamma \in \mathcal{H}^{d,T}_0$.
We then see that the entropy minimization problem
\begin{equation*}
  \inf_{\tilde\nu^\e_x \in \mathcal{P}_x^\e} \e D_{\mathrm{KL}} \left(\tilde\nu^\e_x \| \nu^\e_x \right)
\end{equation*}
is equivalent to the following stochastic optimal control problem:
\begin{equation*}
  \inf_{\tilde\nu^\e_x \in \mathcal{P}_x^\e} \E_{\tilde\nu^\e_x} \left[ \frac{1}{2} \int_{0}^{T} |\dot \gamma(t)|^{2} \mathrm dt + \Phi^\e(\omega) \right] = \inf_{\gamma \in \mathcal{H}^{d,T}_0} \frac{1}{2} \int_{0}^{T} |\dot \gamma(t)|^{2} \mathrm dt + \E_{\mu_x^\e} \left[ \Phi^\e(\omega+\gamma) \right] = \inf_{\gamma \in \mathcal{H}^{d,T}_0} \mathrm{OM}_{\widetilde\Phi_x^\e}[\gamma],
\end{equation*}
where $\widetilde\Phi_x^\e(z) := \E_{\mu_x^\e} \left[ \Phi^\e(\omega+z) \right], z\in \mathcal{C}^{d,T}$.
This recovers the portmanteau theorem in \cite{selk2021information}. In other words, Corollary \ref{equi-opm} generalizes the result of \cite{selk2021information} from path-independent shifts to path-dependent shifts.
\end{rem}

\subsection{Finite time horizon problems: Fixed initial distributions}

Let $\rho_0$ be a probability measure on $\R^d$ and assume that the initial time-marginal measure of $\nu^\e$ is given by $\nu^\e|_{t=0} = \rho_0$. One can construct the measure $\nu^\e$ from the initial measure $\rho_0$ and the transition probabilities $\{\nu^\e_x: x\in\R^d\}$ in \eqref{mu-tilde}, as in \eqref{nu-mu}.

We denote by $\mathcal P_{\rho_0}$ the set of all probability measures on $(\mathcal{C}^{d,T},\mathcal B(\mathcal{C}^{d,T}))$, with initial distribution $\rho_0$. We consider the simplest and trivial case of the entropy minimization problem \eqref{entropy-min-prob}, that is, to take $\widetilde{\mathcal P}$ as $\mathcal P_{\rho_0}$ in Corollary \ref{equi-opm}. It is trivial because the only condition $\tilde\nu^\e = \nu^\e$ that vanishes the entropy $D_{\mathrm{KL}} \left(\tilde\nu^\e \| \nu^\e \right) = 0$ can be achieved. Thus, $\nu^\e$ is the unique minimizer of the stochastic optimal control problem \eqref{sto-opt-control}. Since the initial distribution is fixed, the $\omega(0)$ term in \eqref{sto-opt-control} does not affect the minimization. Thus, 
\begin{equation}\label{opt-cont-argmin}
  \nu^\e = \arg\min_{\tilde\nu^\e \in \mathcal P_{\rho_0}} \E_{\tilde\nu^\e} \left( \frac{1}{2} \int_{0}^{T} |b^\e(t, \omega)|^{2} \mathrm dt + \Phi^\e(\omega) \right).
\end{equation}
That is, the probability measure $\nu^\e$ constructed from $\rho_0$ and $\{\nu^\e_x: x\in\R^d\}$ in \eqref{mu-tilde} minimizes the stochastic functional in \eqref{opt-cont-argmin}, over all probability measures on $(\mathcal{C}^{d,T},\mathcal B(\mathcal{C}^{d,T}))$ with initial distribution $\rho_0$.

Using the reformulation \eqref{SDE-opt-cont-min}--\eqref{SDE-opt-cont} and Corollary \ref{Cor-sde-initial}, we recover the following result of optimal control theory \cite[Theorem 3.1 and Corollary 3.1 in Chapter IV]{fleming2006controlled}.

\begin{prop}\label{sol-control}
Suppose $V\in C^{0,2}_b([0,T]\times\mathbb{R}^d)$ and $g^\e\in C^1_b(\R^d)$. 
%Suppose that $\nu^\e|_{t=0}$ has full support in $\R^d$. 
Given a probability measure $\tilde\nu^\e|_{t=0} = \rho_0$ on $\R^d$, consider the stochastic optimal control problem of minimizing
\begin{equation*}
  J(b^\e) := \E_\P \left[ \int_{0}^{T} \left( \frac{1}{2} |b^\e(t, X^\e)|^{2} + V(t, X^\e(t)) \right) \mathrm dt + g^\e(X^\e(T)) \right],
\end{equation*}
where the triple $X^\e$, $B$, $(\Omega, \mathcal F, \mathbf P, \{\mathcal F_t\}_{t\in[0,T]})$ is a weak solution of SDE \eqref{SDE-opt-cont}.
Let $S^\e\in C^{1,3}_b([0,T]\times \mathbb{R}^d)$ be the unique classical solution of the second-order HJ equation \eqref{2nd-order HJ-full} given in Lemma \ref{well-posedness-2nd-order HJ}.
Then 
%$S^\e(t,x)=J^\e(t,x)=J(t,x;\nabla S^\e)$, where
\begin{equation*}
\min J(b^\e) = J(-\nabla S^\e) = \E_\P \left[ \int_{0}^{T} \left( \frac{1}{2} |\nabla S^\e(t, X^\e(t))|^{2} + V(t, X^\e(t)) \right) \mathrm dt + g^\e(X^\e(T)) \right].
\end{equation*}
\end{prop}

% \begin{prop}
%   Fix an $\e>0$. Suppose that ??? Then the stochastic optimal control problem \eqref{opt-cont}-\eqref{control-sde-1} has a unique solution $x^\e_*$, satisfying \eqref{control-sde-1} with $b^\e(t) = \nabla S^\e(t, x^\e(t))$, where $S^\e$ satisfies the (backward) second-order HJ equation:
% \begin{equation}\label{2nd-order HJ-optimal control}
% \begin{cases}
% \partial_t S^\e(t,x)+\frac{1}{2} |\nabla S^\e(t,x)|^2 + \frac{\e}{2}\Delta S^\e(t,x) = V(t,x), & (t,x)\in [0,T)\times\R^d, \\
% S^\e(T,x) = -\tilde g^\e(x), & x\in\R^d.
% \end{cases}
% \end{equation}
% \end{prop}

% \subsubsection{A 2nd-order HJ equation}

% We now consider the following stochastic optimal control problem:
% \begin{equation}\label{opt-cont}
% J^\e(t,x):= \min \E_{\mu_0} \left[ \int_{t}^{T} \left( \frac{1}{2} |b^\e(s)|^{2} + V(s, X^\e_x(s)) \right) \mathrm ds + g^\e(X^\e_x(T)) \right],
% \end{equation}
% where $X^\e_x$ and $b^\e$ are progressively measurable processes satisfying
% \begin{equation}\label{control-sde-1}
%   \mathrm dX^\e_x(s) = b^\e(s) \mathrm ds + \sqrt\e \mathrm d \omega(s), \quad X^\e_x(t)= x.
% \end{equation}
% Note that different from SDE \eqref{control-sde} in weak form, here we need SDE \eqref{control-sde-1} to hold in the strong sense. The conditions for the existence of classical solutions to the (backward) second-order HJ equation can be found in \cite[Theorems IV.4.1, IV.4.2, and IV.4.3]{fleming2006controlled}.

\subsection{Schr\"odinger's problem: Fixed initial and terminal distributions}

Let $\nu^\e$ be constructed with some initial measure.
Given two probability measures $\tilde\rho_0(\d x) = \tilde\rho_0(x) \d x$ and $\tilde\rho_T(\d x) = \tilde\rho_T(x) \d x$ on $\R^d$, we denote by $\mathcal P_{\tilde\rho_0, \tilde\rho_T}$ the set of all probability measures on $(\mathcal{C}^{d,T},\mathcal B(\mathcal{C}^{d,T}))$, with initial time-marginal distributions $\tilde\rho_0$ and terminal $\tilde\rho_T$.
Take $\widetilde{\mathcal P}$ as $\mathcal P_{\tilde\rho_0, \tilde\rho_T}$ in Corollary \ref{equi-opm}. The entropy minimization problem
\begin{equation}\label{Sch-prob}
  \min_{\tilde\nu^\e \in \mathcal P_{\tilde\rho_0, \tilde\rho_T}} \e D_{\mathrm{KL}} \left(\tilde\nu^\e \| \nu^\e \right)
\end{equation}
has been called Schr\"odinger's problem \cite{leonard2012schrodinger,leonard2014survey}.
As the $\omega(0)$ and $\omega(T)$ terms do not affect the minimization, \eqref{sto-opt-control} reduces to
\begin{equation}\label{Sch-prob-OT}
  \min_{\tilde\nu^\e \in \mathcal P_{\tilde\rho_0, \tilde\rho_T}} \E_{\tilde\nu^\e} \left[ \int_{0}^{T} \left( \frac{1}{2} |b^\e(t, \omega)|^{2} + V(t, \omega(t)) \right) \mathrm dt \right].
\end{equation}
This is the stochastic optimal transport formulation of Schr\"odinger's problem \eqref{Sch-prob}, cf. \cite{Schrodinger1932,mikami2004monge,mikami2021stochastic,leonard2014survey}.

To solve the stochastic optimal transport problem \eqref{Sch-prob-OT}, we compare it with \eqref{opt-cont-argmin} and notice that \eqref{Sch-prob-OT} does not rely on the terminal cost function $g^\e$ but has one more constraint that $\tilde\nu^\e|_{t=T} = \tilde\rho_T$. We can choose a terminal cost such that the minimizer of \eqref{opt-cont-argmin} fulfills the additional constraint of terminal distribution $\tilde\rho_T$. In this way, we can transform the stochastic optimal transport problem \eqref{Sch-prob-OT} into an optimal control problem.

More precisely, let $\nu^\e_*$ be the probability measure constructed from the initial distribution $\tilde\rho_0$ and the transition probabilities $\{\nu^\e_{*x}: x\in\R^d\}$ given by
\begin{equation}\label{transition-Sch-prob}
  \frac {\d\nu^\e_{*x}}{\d\mu^\e_x}(\omega) = \frac{1}{Z^\e_{\Phi^\e_*}(x)} \exp \left\{ -\frac{1}{\e} \Phi^\e_*(\omega) \right\}.
\end{equation}
where
\begin{equation*}
  \Phi^\e_*(\omega) = \int_0^T V(t, \omega(t)) \mathrm dt +g^\e_*(\omega(T)).
\end{equation*}
Then by \eqref{opt-cont-argmin},
\begin{equation*}
  \nu^\e_* = \arg\min_{\tilde\nu^\e \in \mathcal P_{\tilde\rho_0}} \E_{\tilde\nu^\e} \left[ \int_{0}^{T} \left( \frac{1}{2} |b^\e(t, \omega)|^{2} + V(t, \omega(t)) \right) \mathrm dt + g^\e_*(\omega(T)) \right].
\end{equation*}
We thus infer that this $\nu^\e_*$ solves the stochastic optimal transport problem \eqref{Sch-prob-OT} if and only if $\nu^\e_*|_{t=T} = \tilde\rho_T$. 
It follows from \eqref{Sch-sys} that the function $g^\e_*$, together with another function $f^\e_*$, satisfies the following Schr\"odinger's system:
\begin{equation}\label{schrodinger system}
  \left\{
  \begin{aligned}
    e^{-\frac{1}{\e} f^\e_*(x)} \E_{\mu^\e} \left[ \exp \left(-\frac{1}{\e} \int_0^T V(s, \omega(s)) \mathrm ds - \frac{1}{\e} g^\e_*(\omega(0)) \right) \bigg| \omega(0) = x \right] &= \tilde\rho_0(x), \\
    e^{-\frac{1}{\e} g^\e_*(x)} \E_{\mu^\e} \left[ \exp \left(-\frac{1}{\e} \int_0^T V(s, \omega(s)) \mathrm ds - \frac{1}{\e} f^\e_*(\omega(0)) \right) \bigg| \omega(T) = x \right] &= \tilde\rho_T(x).
  \end{aligned}
  \right.
\end{equation}

Therefore, combining with Corollary \ref{Cor-sde-initial}, we obtain

\begin{prop}
Suppose $V\in C^{0,2}_b([0,T]\times\mathbb{R}^d)$. Let $g^\e\in C^1_b(\R^d)$ be a solution of Schr\"{o}dinger's system \eqref{schrodinger system}. 
Let $S^\e_* \in C^{1,3}_b([0,T]\times \mathbb{R}^d)$ be the unique classical solution of the following second-order HJ equation 
\begin{equation}\label{2nd-order HJ-optimal transport}
\begin{cases}
\partial_t S^\e_*(t,x)-\frac{1}{2} |\nabla S^\e_*(t,x)|^2 + \frac{\e}{2}\Delta S^\e_*(t,x) = -V(t,x), & (t,x)\in [0,T)\times\R^d, \\
S^\e_*(T,x) =  g^\e_*(x), & x\in\R^d,
\end{cases}
\end{equation}
given by the following probabilistic representation,
\begin{align*}
  S^\e_*(t,x) = -\e \log  \mathbf{E}_{\mu_0} \left[ \exp\left\{-\frac{1}{\e} \int_t^T V(s,x+ \sqrt{\e} W(s-t)) \mathrm{d}s-\frac{1}{\e} g^\e_*(x+ \sqrt{\e} W(T-t)) \right\} \right].
\end{align*}
Then
\begin{equation*}
  \begin{split}
    \min_{\tilde\nu^\e \in \mathcal P_{\tilde\rho_0, \tilde\rho_T}} \e D_{\mathrm{KL}} \left(\tilde\nu^\e \| \nu^\e\right) &= \E_{\nu^\e_*} \left[ \int_{0}^{T} \left( \frac{1}{2} |\nabla S^\e_*(t, \omega(t))|^{2} + V(t, \omega(t)) \right) \mathrm dt \right] \\
    &\quad + \e \E_{\tilde\rho_0} \left( \log Z^\e_{\Phi^\e} + \log \frac{\mathrm{d} \tilde\rho_0}{\mathrm{d} \nu^\e|_{t=0}} \right) + \E_{\tilde\rho_T} (g^\e_*),
  \end{split}
\end{equation*}
where $\nu^\e_*$ denotes the probability measure with initial distribution $\tilde\rho_0$ and transition probabilities $\{\nu^\e_{*x}: x\in\R^d\}$ in \eqref{transition-Sch-prob}.
\end{prop}

\subsection{Stochastic Euler--Lagrange equation}\label{subsec:SEL}

The underlying geometry and mechanics of Schrodinger's problem, or more generally, stochastic optimal transport, have been developed in the recent work \cite{huang2023second}. 

Observe the analogy between the OM functional \eqref{OM} and the stochastic action functional in the previous subsections, especially the one in \eqref{opt-cont-argmin}: 
\begin{equation*}
  \mathcal S_{\tilde\nu^\e}[\omega] := \E_{\tilde\nu^\e} \left( \frac{1}{2} \int_{0}^{T} |b^\e(t, \omega)|^{2} \mathrm dt + \Phi^\e(\omega) \right).
\end{equation*}
We now introduce a ``stochastic derivative'' for a path $\omega \in \mathcal{C}^{d,T}$, such that when the triple $\omega(\cdot)$, $\widetilde B$, $(\mathcal{C}^{d,T}, \mathcal B(\mathcal{C}^{d,T}), \tilde\nu^\e, \{\mathcal B_t(\mathcal{C}^{d,T})\}_{t\in [0,T]})$ is a weak solution of the functional SDE \eqref{control-sde}, the stochastic derivative of the path $\omega(\cdot)$ at time $t$ equals to $b^\e(t, \omega)$. Thus, the stochastic action functional in \eqref{opt-cont-argmin} can be regarded as a stochastic counterpart of the OM functional \eqref{OM}.

More precisely, we define for $\omega \in \mathcal{C}^{d,T}$,
\begin{equation*}
  D_{\tilde\nu^\e} \omega(t) = \lim_{\Delta t \to 0^+} \E_{\tilde\nu^\e} \left[ \frac{\omega(t+\Delta t) - \omega(t)}{\Delta t} \bigg| \mathcal B_t(\mathcal{C}^{d,T}) \right].
\end{equation*}
This is the so-called Nelson's mean derivative when $\tilde\nu^\e$ is realized as the law of a semimartingale $X$ on a filtered probability space $(\Omega, \mathcal F, \mathbf P, \{\mathcal F_t\}_{t\in[0,T]})$, as follows (see e.g. \cite{Nel01}):
\begin{equation*}
  DX(t) = \lim_{\e \to 0^+} \E_\P \left[ \frac{X(t+\Delta t) - X(t)}{\Delta t} \bigg| \mathcal F_t \right].
\end{equation*}

The stochastic action functional becomes
\begin{equation*}
  \mathcal S_{\tilde\nu^\e}[\omega] := \E_{\tilde\nu^\e} \left( \frac{1}{2} \int_{0}^{T} |D_{\tilde\nu^\e} \omega(t)|^{2} \mathrm dt + \Phi^\e(\omega) \right).
\end{equation*}
When $\Phi^\e$ is of the form \eqref{phi}, with running cost $V$ and terminal cost $g^\e$, it becomes
\begin{equation}\label{stoch-action}
  \begin{split}
    \mathcal S_{\tilde\nu^\e}[\omega] &= \E_{\tilde\nu^\e} \left[ \int_{0}^{T} \left( \frac{1}{2} |D_{\tilde\nu^\e} \omega(t)|^{2} + V(t, \omega(t)) \right) \mathrm dt + g^\e(\omega(T)) \right] \\
    &= \E_{\tilde\nu^\e} \left[ \int_0^T L_V\left( t, \omega(t), D_{\tilde\nu^\e} \omega(t) \right) dt + g^\e(\omega(T)) \right],
  \end{split}
\end{equation}
where $L_V$ is the standard Euclidean Lagrangian of \eqref{Lim-lag}.

Using a Cameron--Martin variation $\gamma \in \mathcal{H}^{d,T}_0$ for a path $\omega \in \mathcal{C}^{d,T}$, i.e., $\delta \omega = h$, it was shown in \cite[Section 7.2]{huang2023second} (cf. also \cite{CruzeiroZambrini1991}) that 
\begin{equation*}
  \delta D_{\tilde\nu^\e} \omega = \dot \gamma.
\end{equation*}
The stationary-action principle of the stochastic action functional \eqref{stoch-action} for $\tilde\nu^\e \in \mathcal P_{\rho_0}$ follows, applying It\^o's formula,
\begin{equation*}
  \begin{split}
    0 = \delta \mathcal S_{\tilde\nu^\e}[\omega] &= \E_{\tilde\nu^\e} \left[ \int_{0}^{T} \left( D_{\tilde\nu^\e} \omega(t) \delta D_{\tilde\nu^\e} \omega(t) + \nabla V(t, \omega(t)) \delta \omega(t) \right) \mathrm dt + \nabla g^\e(\omega(T)) \delta \omega(T) \right] \\
    &= \E_{\tilde\nu^\e} \left[ \int_{0}^{T} \left( D_{\tilde\nu^\e} \omega(t) \dot \gamma(t) + \nabla V(t, \omega(t)) \gamma(t) \right) \mathrm dt + \nabla g^\e(\omega(T)) \gamma(T) \right] \\
    &= \E_{\tilde\nu^\e} \left[ \int_{0}^{T} \left( - D_{\tilde\nu^\e} D_{\tilde\nu^\e} \omega(t) + \nabla V(t, \omega(t)) \right) \gamma(t) \mathrm dt + \left( D_{\tilde\nu^\e} \omega(T) + \nabla g^\e(\omega(T)) \right) \gamma(T) \right],
  \end{split}
\end{equation*}
which leads to the following stochastic Euler--Lagrange equation
% \begin{equation*}
%   D \left( \frac{\pt L_V}{\pt \dot x} \right) = \frac{\pt L_V}{\pt x}
% \end{equation*}
% reads
\begin{equation}\label{s-EL}
  \begin{cases}
    D_{\tilde\nu^\e} D_{\tilde\nu^\e} \omega(t) = \nabla V(t, \omega(t)), \quad t\in (0,T), \\
    \tilde\nu^\e|_{t=0} = \rho_0, \quad D_{\tilde\nu^\e} \omega(T) = - \nabla g^\e(\omega(T)).
  \end{cases}
\end{equation}
Such an equation is called a `mean differential equation' in \cite{huang2023second}, compared with the Euler--Lagrange equation \eqref{eur-lag} of the OM functional.

Since $\nu^\e$ is a minimizer of $\mathcal S_{\tilde\nu^\e}[\omega]$ as shown in \eqref{opt-cont-argmin}, it solves the stochastic EL equation \eqref{s-EL}. If $\nu^\e$ can be realized as the law of SDE \eqref{SDE-opt-cont} with a Markovian drift $b^\e(t, \omega) = b^\e(t, \omega(t))$ for some function $b^\e: [0,T] \times \R^d \to \R^d$, or equivalently, the triple $\omega(\cdot)$, $\widetilde B$, $(\mathcal{C}^{d,T}, \mathcal B(\mathcal{C}^{d,T}), \nu^\e, \{\mathcal B_t(\mathcal{C}^{d,T})\}_{t\in [0,T]})$ is a weak solution of the functional SDE \eqref{control-sde}, then
\begin{equation*}
  D_{\tilde\nu^\e} \omega(t) = b^\e(t, \omega(t)),
\end{equation*}
and by It\^o's formula,
\begin{equation*}
  D_{\tilde\nu^\e} D_{\tilde\nu^\e} \omega(t) = D_{\tilde\nu^\e} [b^\e(t, \omega(t))] = \left( \frac{\pt}{\pt t} + b^\e(t, \omega(t)) \cdot \nabla + \frac{\e}{2} \Delta \right) b^\e(t, \omega(t)).
\end{equation*}
Thus, the stochastic Euler--Lagrange equation \eqref{s-EL} amounts to the viscous Burgers' equation \eqref{Navier-Stokes}. Comparing this equation with the nonlinear heat equation \eqref{NH}, we see that 
\begin{equation*}
  \pt_i b^\e_j(t,x) = \pt_j b^\e_i(t,x)
\end{equation*}
if $b^\e(t,x) \ne 0$. The above identity means that the vector field $b^\e$ is closed, and thus, locally exact, i.e., $b^\e$ is locally a gradient field. Suppose $b^\e(t,x) = \nabla S^\e(t,x)$ for some function $S^\e: [0,T] \times U \to \R$ with a domain $U\subset \R^d$, then $S$ satisfies the 2nd-order HJ equation \eqref{2nd-order HJ} on $[0,T] \times U$. This gives a partial converse of Proposition \ref{sol-control}.

The second-order/stochastic geometric interpretation of the 2nd-order HJ equation \eqref{2nd-order HJ-optimal transport} and its canonical relations with stochastic Hamiltonian mechanics have been established in \cite{huang2023second}.

\section{Application III: Entropy production in stochastic thermodynamics}\label{app-3}

% The emergence of stochastic dynamics, as a fundamental framework for describing temporal evolution, reflects profound methodological shifts in modern applied mathematics. This development has been stimulated by advancements in the quantification of biological systems \cite{qian2011nonlinear} and persistent theoretical challenges in statistical mechanics \cite{qian2001mathematical,ao2008emerging}, where randomness fundamentally reshapes our understanding of dynamical processes. Building on connections between stochastic motion and analytical mechanics, the Hamilton--Jacobi approach serves as a powerful tool to analyze large deviation principles and the most probable paths in stochastic dynamics \cite{ge2012analytical}.
% In \cite{miao2024emergence}, a Hamilton--Jacobi equation emerges in a description of the evolution of entropy of a stochastic dynamical system under observation and in the limit of large information extent in homogeneous space and time.

In this section, we employ the framework of stochastic thermodynamics to illustrate the applicability of our results in this context. We begin by revisiting several foundational concepts in thermodynamics, such as the probability current and entropy production. Building on these concepts, we proceed to analyze the irreversibility of thermodynamic systems and establish how the second law of thermodynamics emerges as a natural consequence of our findings.

In particular, we provide a rigorous derivation and justification of the fluctuation theorem, which serves as a statistical foundation for understanding how macroscopic thermodynamic laws arise from the stochastic dynamics of microscopic systems. Furthermore, we present a novel decomposition formula for entropy production that applies to more general thermodynamic systems, utilizing the potential energy representation \eqref{phi}.

% We begin by recalling several fundamental notions, such as the probability current and entropy production, which characterize the nonequilibrium features of stochastic systems. Based on these quantities, we analyze the emergence of irreversibility and show how the second law of thermodynamics arises as a direct and rigorous consequence of the underlying stochastic dynamics. Finally, we establish a fluctuation theorem that provides the statistical foundation for the second law, clarifying how macroscopic thermodynamic regularities emerge from microscopic stochasticity. Moreover, by employing the potential energy representation~\eqref{phi}, we derive a novel decomposition formula for entropy production that extends to a broad class of thermodynamic systems.

\subsection{Entropy production}

Stochastic dynamics has become a cornerstone in the modern description of temporal evolution, marking a profound methodological transition in applied mathematics. This paradigm has been anticipated in the quantitative modeling of biological systems~\cite{qian2011nonlinear} and in ongoing theoretical challenges of stochastic thermodynamics~\cite{qian2001mathematical,ao2008emerging,peliti2021stochastic}, where randomness fundamentally alters the traditional deterministic picture of dynamical processes. Within this perspective, the connection between stochastic motion and analytical mechanics has proven particularly fruitful: the Hamilton--Jacobi framework provides a natural analytical formulation for studying large deviation principles and the most probable paths in stochastic systems~\cite{ge2012analytical}. Notably, Miao \emph{et al.}~\cite{miao2024emergence} demonstrated that a Hamilton--Jacobi equation can emerge in the description of entropy evolution for stochastic dynamical systems under observation, in the limit of large information extent and homogeneous space--time. Motivated by these developments, we turn in this section to the framework of \emph{stochastic thermodynamics}, which provides a natural setting to connect stochastic dynamics with macroscopic thermodynamic behavior.

Consider the Markovian stochastic differential equation \eqref{SDE-Markovian} for a fixed $\epsilon > 0$, with a time-independent drift term, given by  
\begin{align}\label{eqn:SDEentropy}
\mathrm{d}X^\epsilon_x(t) = b^\epsilon(X^\epsilon_x(t))\,\mathrm{d}t + \sqrt{\epsilon}\,\mathrm{d}B(t), \quad X^\epsilon_x(0) = x,
\end{align}
where $b^\epsilon : \mathbb{R}^d \to \mathbb{R}^d$ represents the drift. Assume that the solution process $X^\epsilon_x$ admits a family of probability densities $\rho^\epsilon(t, \cdot)$, for $t \in [0, T]$, which satisfy the associated Fokker--Planck equation:  
\begin{align}\label{eqn:FPE}
\partial_t \rho^\epsilon(t, x) = -\nabla \cdot \left[\left(b^\epsilon(x) - \frac{\epsilon}{2} \nabla \log \rho^\epsilon(t, x)\right) \rho^\epsilon(t, x)\right].
\end{align}

To simplify further analysis, we give the following concepts.
\begin{Defi}[Probability current and current velocity]\label{def:currentandv}
    Define the probability current $j^\epsilon: [0,T]\times \R^d \to \R^d$ and the associated velocity field $v^\epsilon: [0,T]\times \R^d \to \R^d$ of system \eqref{eqn:SDEentropy} as follows:
\begin{align*}
j^\epsilon(t, x) &:= b^\epsilon(x)\rho^\epsilon(t, x) - \frac{\epsilon}{2} \nabla \rho^\epsilon(t, x), \\
v^\epsilon(t, x) &:= b^\epsilon(x) - \frac{\epsilon}{2} \nabla \log \rho^\epsilon(t, x).
\end{align*}
\end{Defi}
Here, $j^\epsilon$ represents the flux of the probability density, and $v^\epsilon$ describes the effective velocity field of the probability flow. 

In the limit $T \to \infty$, the system reaches its statistically steady state, characterized by the invariant distribution $\rho^\epsilon_\infty$, where the condition $-\nabla \cdot j^\epsilon = 0$ holds. However, if the drift term $b^\epsilon$ is not expressible as the gradient of some potential function, the probability current $j^\epsilon(x)$ does not vanish, i.e., $j^\epsilon(x) \neq 0$. This implies that the detailed balance condition is not satisfied in the steady state, keeping the system out of thermodynamic equilibrium. Consequently, the entropy production of the system remains nonzero in the steady state, as will be discussed in the next subsection.

Inspired by \cite{seifert2005entropy,huang2026entropy, peliti2021stochastic}, we give the following definition.
\begin{Defi}[Stochastic entropy along trajectories]
    The \emph{stochastic entropy} of the system \eqref{eqn:SDEentropy} along a trajectory $\{\omega(t)\}_{t \in [0, T]}$ is defined as
\begin{align}\label{def:sysEP}
    s_\mathrm{sys}(t, \omega) = -\log \rho^\epsilon(T-t, \omega(t)),
\end{align}
where $\rho^\epsilon(T-t, \omega(t))$ is the probability density associated with the solution process of \eqref{eqn:SDEentropy} at time $T-t$ and at position $\omega(t)$. 
\end{Defi}
\begin{rem}
\normalfont
The definition of stochastic entropy proposed here differs from the conventional formulation, \(-\log\rho^{\epsilon}(t,\omega(t))\), which is commonly used in the physical literature (see, e.g., \cite{seifert2005entropy,peliti2021stochastic}). Our definition is adopted because it uniquely provides a mathematically rigorous characterization of process irreversibility \textit{prior} to the system relaxing to its stationary state. This distinction is crucial for analyzing the thermodynamics of non-stationary regimes. We emphasize, however, that both definitions coincide once the system reaches its steady state.
\end{rem}
The corresponding stochastic differential of $s_\mathrm{sys}$ can then be written, using definition \ref{def:currentandv},
\begin{equation}\label{eqn:entropyrate2}
\begin{aligned}
    \mathrm{d}s_\mathrm{sys}(t, \omega) 
    &= -\mathrm{d} \log \rho^\epsilon(T-t, \omega(t)) \\
    &= \frac{(\partial_t \rho^\epsilon)(T-t, \omega(t))}{\rho^\epsilon(T-t, \omega(t))} \,\mathrm{d}t 
       - \frac{\nabla \rho^\epsilon(T-t, \omega(t))}{\rho^\epsilon(T-t, \omega(t))} \circ \mathrm{d} \omega(t) \\
    &= \frac{(\partial_t \rho^\epsilon)(T-t, \omega(t))}{\rho^\epsilon(T-t, \omega(t))} \,\mathrm{d}t 
       - \frac{2}{\epsilon} \left[b^\epsilon(\omega(t)) - \frac{j^\epsilon(T-t, \omega(t))}{\rho^\epsilon(T-t, \omega(t))}\right] \circ \mathrm{d} \omega(t),
\end{aligned}
\end{equation}
where $\circ$ denotes the Stratonovich differential, which is used to preserve the chain rule in stochastic calculus. The term $\frac{2}{\epsilon} b^\epsilon(\omega(t)) \circ \mathrm{d}\omega(t)$ can be interpreted as the rate of heat dissipation in the medium (see \cite{seifert2005entropy}), which motivates the definition of the \emph{medium entropy production}:
\begin{align}\label{def:mEP}
    \mathrm{d}s_\mathrm{m}(t, \omega) := \frac{2}{\epsilon} b^\epsilon(\omega(t)) \circ \mathrm{d}\omega(t).
\end{align}
The remaining terms of the last equality in \eqref{eqn:entropyrate2} contribute to the \emph{total entropy production}, which is defined as
\begin{align*}
    \mathrm{d}s_\mathrm{tot}(t, \omega) := \frac{(\partial_t \rho^\epsilon)(T-t, \omega(t))}{\rho^\epsilon(T-t, \omega(t))} \,\mathrm{d}t 
    + \frac{2}{\epsilon} \frac{j^\epsilon(T-t, \omega(t))}{\rho^\epsilon(T-t, \omega(t))} \circ \mathrm{d} \omega(t).
\end{align*}

Combining these expressions, we obtain the following decomposition of the entropy production.
\begin{lem}[Entropy production decomposition formula]\label{eqn:entropyrate}
    The entropy production decomposition formula for the stochastic differential equation \eqref{eqn:SDEentropy} is given by:  
\begin{align}\label{eqn:entropydecomposition}
    \underbrace{\frac{(\partial_t \rho^\epsilon)(T-t, \omega(t))}{\rho^\epsilon(T-t, \omega(t))} \,\mathrm{d}t 
    + \frac{2}{\epsilon} \frac{j^\epsilon(T-t, \omega(t))}{\rho^\epsilon(T-t, \omega(t))} \circ \mathrm{d}\omega(t)}_{\mathrm{d}s_\mathrm{tot}(t, \omega)}
    = \mathrm{d}s_\mathrm{sys}(t, \omega) 
    + \underbrace{\frac{2}{\epsilon} b^\epsilon(\omega(t)) \circ \mathrm{d}\omega(t)}_{\mathrm{d}s_\mathrm{m}(t, \omega)}.
\end{align}  
Here, the total entropy production, \(\mathrm{d}s_\mathrm{tot}(t, \omega)\), consists of two components: the system entropy production, \(\mathrm{d}s_\mathrm{sys}(t, \omega)\), and the medium entropy production, \(\mathrm{d}s_\mathrm{m}(t, \omega)\).
\end{lem}

\subsection{Thermodynamic irreversibility and fluctuation theorems}
The Second Law of Thermodynamics asserts that the total entropy of an isolated system can never decrease over time; it either increases or, in the case of a reversible process, remains constant. This law is intrinsically tied to the irreversibility of natural processes. In real-world systems, energy transformations are inherently inefficient, leading to an increase in entropy  and turning the processes irreversible. For example, heat spontaneously flows from a hotter region to a colder one but never in the reverse direction without external work, exemplifying the natural tendency toward higher entropy and irreversibility. This is in sharp contrast with Schr\"odinger's problem of Sections \ref{subsec:reversal} and \ref{subsec-Born}. 

In what follows, we analyze this phenomenon within the framework of stochastic differential equations, employing the tools of stochastic analysis to provide a rigorous mathematical perspective.

The Fokker--Planck equation \eqref{eqn:FPE} satisfied by the probability densities $\rho^\epsilon(t,x)$ can be rewritten in the following form:  
\begin{align*}
    \partial_t \rho^\epsilon(t,x)
    = -\nabla \cdot \left(\left( b^\epsilon(x) - \epsilon \nabla \log \rho^\epsilon(t,x) \right) \rho^\epsilon(t,x)\right) 
    - \frac{\epsilon}{2} \Delta \rho^\epsilon(t,x),
\end{align*}
where we use the identity  
\[
\epsilon \nabla \cdot (\nabla \rho^\epsilon(t,x)) = -\epsilon \Delta \rho^\epsilon(t,x) + 2\epsilon \nabla \cdot (\nabla \log \rho^\epsilon(t,x) \rho^\epsilon(t,x)).
\]

The time-reversed process $\rev X^\epsilon_x$, defined in Section \ref{subsec:path}, has the probability density  
\begin{align*}
    \rev\rho^\epsilon(t,x) = \rho^\epsilon(T-t,x),
\end{align*}
and this density satisfies the following Fokker--Planck equation:  
\begin{align}\label{eqn:newFPE}
    \partial_t \rev \rho^\epsilon(t,x) 
    = -\nabla \cdot \left(\left( -b^\epsilon(x) + \epsilon \nabla \log \rho^\epsilon(T-t,x) \right) \rev \rho^\epsilon(t,x)\right) 
    + \frac{\epsilon}{2} \Delta \rev \rho^\epsilon(t,x),
\end{align}
where the drift term \(-b^\epsilon(x) + \epsilon \nabla \log \rho^\epsilon(T-t,x)\) is recognized as the new effective drift. Consequently, the process $\rev X^\epsilon_x$ is a weak solution to the following SDE:  
\begin{align}\label{eqn:reversedSDE}
    \mathrm{d} \rev X^\epsilon_x(t) 
    = \left(-b^\epsilon(\rev X^\epsilon_x(t)) + \epsilon \nabla \log \rho^\epsilon(T-t, \rev X^\epsilon_x(t)) \right) \mathrm{d}t 
    + \sqrt{\epsilon} \, \mathrm{d} \widetilde B(t), \quad \rev X^\epsilon_x(0) \sim \rho^\epsilon(T,\cdot),
\end{align}
where $\widetilde B$ is a standard Brownian motion defined on the probability space \((\Omega^R, \mathcal{F}^R, \mathbf{P}^R)\). The above derivation is informal, as the PDE \eqref{eqn:newFPE} only establishes that \eqref{eqn:reversedSDE} shares the same probability flow as the time-reversal of the process \(X^\epsilon_x\). However, this result can be made rigorous, as shown in \cite{anderson1982reverse}.  

It is straightforward to verify that the conditional pushforward measures \(B_*\mathbf{P}(\cdot \mid x_0)\) and \(\widetilde B_*\mathbf{P}^R(\cdot \mid x_0)\) are identical, as both correspond to the Wiener measure \(\mu_{x}\) conditioned on starting at \(x\) in the path space.

The laws of $X^\e_x$ and $\rev{X}^\e_x$ are denoted as $\nu_x^\e$ and $\rev{\nu}^\e_x$ respectively, and the Radon--Nikodym derivative between the pushforward measures $\nu_x^\e$ and $\rev{\nu}^\e_x$ can be derived directly using Girsanov's theorem:
\begin{equation}\label{ean:entropyderivation}
\begin{aligned}
    \frac{\mathrm{d} \nu_x^\e}{\mathrm{d} \rev \nu_x^\e}(\omega)
    =&\ \frac{\mathrm{d} \nu_x^\e}{\mathrm{d} \mu_x} (\omega) 
    \frac{\mathrm{d} \mu_x}{\mathrm{d}\rev{\nu}^\e_x} (\omega) \\
    =&\ \exp\left\{ \frac{1}{\epsilon} \int_0^T b^\epsilon(\omega(t))  \mathrm{d} \omega(t) 
    - \frac{1}{2\epsilon} \int_0^T \left| b^\epsilon(\omega(t)) \right|^2 \, \mathrm{d}t \right. \\
    &\ - \frac{1}{\epsilon} \int_0^T \left( b^\epsilon(\omega(t)) - 2v^\epsilon(T-t, \omega(t)) \right)  \mathrm{d} \omega(t)  \left. + \frac{1}{2\epsilon} \int_0^T \left| b^\epsilon(\omega(t)) - 2v^\epsilon(T-t, \omega(t)) \right|^2 \, \mathrm{d}t \right\} \\
    =&\ \exp\left\{ \frac{2}{\epsilon} \int_0^T v^\epsilon(T-t, \omega(t))  \mathrm{d} \omega(t) - \frac{2}{\epsilon} \int_0^T \left( b^\epsilon(\omega(t)) - v^\epsilon(T-t, \omega(t)) \right)  v^\epsilon(T-t, \omega(t)) \, \mathrm{d}t \right\} \\
  \underset{(*)}{=}&\ \exp\left\{
\frac{2}{\epsilon} \int_0^T v^\epsilon(T-t, \omega(t)) \circ \mathrm{d} \omega(t)
- \frac{1}{\epsilon}\int_0^T \nabla \cdot v^\epsilon(T-t, \omega(t)) \, \mathrm{d}t
\right\}\\
    &\ \left. - \frac{2}{\epsilon} \int_0^T \left( b^\epsilon(\omega(t)) - v^\epsilon(T-t, \omega(t)) \right) \cdot v^\epsilon(T-t, \omega(t)) \, \mathrm{d}t \right\} \\
    =&\ \exp\left\{ \frac{2}{\epsilon} \int_0^T v^\epsilon(T-t, \omega(t)) \circ \mathrm{d} \omega(t) 
    - \int_0^T \frac{\nabla \cdot \left( v^\epsilon(T-t, \omega(t)) \rho^\epsilon(T-t, \omega(t)) \right)}{\rho^\epsilon(T-t, \omega(t))} \, \mathrm{d}t \right\}.
\end{aligned}
\end{equation}
Here, in the equality $(*)$, we convert the Itô integral into the Stratonovich integral. This transformation is necessary because the entropy production terms defined in \eqref{eqn:entropyrate} are formulated within the Stratonovich framework. Note that \eqref{ean:entropyderivation} is a special case of \eqref{eqn:RNnu} in Subsection \ref{subsec:reversal}.

It is easy to see that $\rho^\epsilon(T-t, \omega(t))$ satisfies the following PDE:
\begin{align*}
    -(\partial_t \rho^\epsilon)(T-t, \omega(t)) 
    &= \nabla \cdot \left[ \left( b^\epsilon(\omega(t)) - \frac{\epsilon}{2} \nabla \log \rho^\epsilon(T-t, \omega(t)) \right) \rho^\epsilon(T-t, \omega(t)) \right] \\
    &= \nabla \cdot \left( v^\epsilon(T-t, \omega(t)) \rho^\epsilon(T-t, \omega(t)) \right) \\
    &= \nabla \cdot j^\epsilon(T-t, \omega(t)),
\end{align*}
and
\begin{align*}
    v^\epsilon(T-t, \omega(t)) \circ \mathrm{d} \omega(t) 
    = b^\epsilon(\omega(t)) \circ \mathrm{d} \omega(t)     - \frac{\epsilon}{2} \nabla \log \rho^\epsilon(T-t, \omega(t)) \circ \mathrm{d} \omega(t).
\end{align*}

Recalling the entropy production formula of Lemma \ref{eqn:entropyrate} and combining it with \eqref{ean:entropyderivation}, we deduce that  
\begin{equation}\label{eqn:totalderivation}
\begin{aligned}
    \log \frac{\mathrm{d} \nu^\epsilon_x}{\mathrm{d} \rev{\nu}^\epsilon_x}(\omega)
    =&\ \int_0^T \mathrm{d}s_\mathrm{sys}(t, \omega) + \int_0^T \mathrm{d}s_\mathrm{m}(t, \omega) \\
    =&\ \int_0^T \mathrm{d}s_\mathrm{tot}(t, \omega) \\
    =&\ s_\mathrm{tot}(T, \omega) - s_\mathrm{tot}(0, \omega).
\end{aligned}
\end{equation}
We now establish the following mathematical formulation of the second law of thermodynamics.

\begin{thm}[Second law of thermodynamics]\label{theo:secondlaw}
    For the system \eqref{eqn:SDEentropy} with an arbitrary initial distribution, the ensemble average of the total entropy production over any time interval $[0, T]$ is non-decreasing, i.e.,  
    \begin{equation*}
       \E_{\nu^\epsilon} \left[ s_\mathrm{tot}(T, \omega) - s_\mathrm{tot}(0, \omega) \right] \geq 0.
    \end{equation*}
\end{thm}

\begin{proof}
    From the previously derived equalities \eqref{eqn:totalderivation} and the nonnegativity of the Kullback--Leibler divergence \eqref{Gibbs-ineq}, we have  
    \begin{equation*}
        \begin{aligned}
            \E_{\nu^\epsilon} \left[ \log \frac{\mathrm{d} \nu^\epsilon_x}{\mathrm{d} \rev{\nu}^\epsilon_x}(\omega) \right]
            = \E_{\nu^\epsilon} \left[ s_\mathrm{tot}(T, \omega) - s_\mathrm{tot}(0, \omega) \right] \geq 0,
        \end{aligned}
    \end{equation*}
    which completes the proof.
\end{proof}

This theorem provides a rigorous mathematical formulation of the \textbf{Second Law of Thermodynamics}. Specifically, it implies that any natural process evolves in suach a way that the total entropy of all systems involved in the process does not decrease on average.

Next, we examine the mathematical structure of entropy production at the microscopic level. While Theorem \ref{theo:secondlaw} describes the macroscopic physical law that entropy production is always nondecreasing on average, at the microscopic level, certain individual trajectories may exhibit decreasing stochastic entropy production. This phenomenon is known as \emph{fluctuation relation} in physics.

To illustrate this, we consider a general setting where the initial distribution of the process \(X^\epsilon\) in \eqref{eqn:SDEentropy} is an arbitrary probability distribution  \(\nu^\epsilon|_{t=0}\). The following holds:  
\begin{equation*}
    \begin{aligned}
        \E_{\nu^\epsilon} \left[ e^{-\log\frac{\mathrm{d} \nu^\epsilon_x}{\mathrm{d} \rev{\nu}^\epsilon_x}(\omega)} \right] 
        &:= \int_{\mathbb{R}^d} \int_\Omega e^{-\log\frac{\mathrm{d} \nu^\epsilon_x}{\mathrm{d} \rev{\nu}^\epsilon_x}(\omega)} \nu^\epsilon_x(\mathrm{d}\omega) \nu^\epsilon|_{t=0}(\mathrm{d}x) \\
        &= \int_{\mathbb{R}^d} \int_\Omega \frac{\mathrm{d} \rev{\nu}^\epsilon_x}{\mathrm{d} \nu^\epsilon_x}(\omega) \nu^\epsilon_x(\mathrm{d}\omega) \nu^\epsilon|_{t=0}(\mathrm{d}x) \\
        &= 1.
    \end{aligned}
\end{equation*}
This result recovers the fluctuation theorem \cite{seifert2005entropy}, which can be formulated as the following statement.
\begin{thm}[Integral fluctuation theorem]\label{theo:IFT}
    For the system \eqref{eqn:SDEentropy} with an arbitrary initial distribution, the stochastic total entropy production over any time interval $[0,T]$ satisfies the following relation:
    \begin{equation*}
    \begin{aligned}
         \E_{\nu^\epsilon} \left[ e^{-\Delta s_\mathrm{tot} } \right] =\E_{\nu^\epsilon} \left[ e^{-(s_\mathrm{tot}(T, \omega) - s_\mathrm{tot}(0, \omega))} \right]  = 1.
    \end{aligned}
\end{equation*}
\end{thm}

{{Beside the integral fluctuation theorem \ref{theo:IFT}, we are able to provide a stronger relation on irreversibility. Recall Corollary~\ref{cor-0}(ii). For the total measure $\nu^\e$, we have
\begin{equation*}
    \begin{aligned}
         \log \frac{\d\nu^\e}{\d\rev\nu^\e}(\omega) =&  \log\frac{\d\nu|_{t=0}}{\d\rev\nu^\e|_{t=0}} (\omega(0)) +\log\frac{\d\nu_x}{\d\rev\nu^\e_x}(\omega) \bigg|_{x = \omega(0)}\\
         =& \log\frac{\rho^\e(0,\omega(0))}{\rho^\e(T,\omega(0))} + \Delta s_\mathrm{tot}(\omega).
    \end{aligned}
\end{equation*}
We introduce the \emph{irreversibility index}:
\begin{equation}\label{IIF}
    \mathcal{R}(\omega):=\log\frac{\rho^\e(0,\omega(0))}{\rho^\e(T,\omega(0))} + \Delta s_\mathrm{tot}(\omega).
\end{equation}
The functional $\mathcal{R}$ satisfies the following detailed fluctuation relation.

\begin{thm}[Detailed fluctuation theorem]\label{theo:DFT}
    For the system \eqref{eqn:SDEentropy} with an arbitrary initial distribution, the irreversibility index $\mathcal{R}$ defined in \eqref{IIF} over any time interval $[0,T]$ satisfies
\begin{equation*}
    \begin{aligned}
         \frac{\rho_\mathcal{R}(\sigma)}{\rho_\mathcal{R}(-\sigma)}=e^{\sigma},\quad \text{for all}\quad \sigma \in \mathbb{R},
    \end{aligned}
\end{equation*}
where $\rho_{\mathcal{R}}$ denotes the probability distribution of the random variable $\mathcal{R}$.
\end{thm}
\begin{proof}
Recall the time-reversal operator $R$ introduced in Appendix~\ref{subsec:path}. We note that
\begin{equation*}
        \begin{aligned}
            \frac{\d\nu^\e}{\d\rev\nu^\e}(\omega) =&  \frac{\d\nu^\e(\omega)}{\d \nu^\e\circ R^{-1}(\omega)}  
            = \frac{\d\nu^\e(\omega)}{\d \nu^\e(\rev\omega)}.
        \end{aligned}
    \end{equation*}
    Hence
    \begin{equation*}
        \mathcal{R}(\omega)=-\mathcal{R}(\rev\omega).
    \end{equation*}
    For any $\sigma\in\mathbb{R}$, we compute
    \begin{equation*}
        \begin{aligned}
            \rho_\mathcal{R}(\sigma)
        =& \int_\Omega \delta(\mathcal{R}(\omega)-\sigma)\d\nu^\e(\omega) \\
        =& \int_\Omega \delta(\mathcal{R}(\omega)-\sigma)e^{\mathcal{R}(\omega)}\d\nu^\e(\rev\omega) \\
        =& e^{\sigma}\int_\Omega \delta(\mathcal{R}(\omega)-\sigma) \d\nu^\e(\rev\omega) \\
        =& e^{\sigma}\int_\Omega \delta(-\mathcal{R}(\rev\omega)-\sigma) \d\nu^\e(\rev\omega) \\
        =& e^{\sigma}\rho_\mathcal{R}(-\sigma).
        \end{aligned}
    \end{equation*}
    This completes the proof.
\end{proof}
}}

\begin{rem}
\normalfont
Note that similar versions of Theorems~\ref{theo:secondlaw},~\ref{theo:IFT} and \ref{theo:DFT} were derived in \cite{seifert2005entropy} via physical arguments. 
We observe that the approaches in \cite{seifert2005entropy} (see equation~(14) therein) and ours constitute two distinct perspectives in stochastic thermodynamics. 
Specifically, in \cite{seifert2005entropy}, time-reversal is applied to trajectories in path space~$C[0,T]$, rather than to the process~$X_{\cdot}$ itself. 
In contrast, our arguments rely entirely on the comparison between forward and reversed processes.
\end{rem}

Now we are ready to calculate the total entropy production in a more computational perspective. 
\begin{thm}[Total entropy production]\label{theo:EPRformula}
    The total entropy production on average of the system \eqref{eqn:SDEentropy} is given as:
    \begin{equation*}
\begin{aligned}
     \E_{\nu^\epsilon} \left[ s_\mathrm{tot}(T, \omega) - s_\mathrm{tot}(0, \omega) \right] 
    = \frac{2}{\epsilon} \int_0^T\int_{\mathbb{R}^d} \left|v^\epsilon(T-t, x) \right|^2 \rho^\e(t,x)\d x\mathrm{d}t,
\end{aligned}
\end{equation*}
where $v^\e(\cdot,\cdot)$ is the current velocity filed given in Definition \ref{def:currentandv} and $\rho(\cdot,\cdot)$ is the probability density of the solution process of \eqref{eqn:SDEentropy}.
\end{thm}
\begin{proof}
Recall in \eqref{ean:entropyderivation}, we have 
\begin{equation*}
\begin{aligned}
    \frac{\mathrm{d} \nu_x^\e}{\mathrm{d} \rev \nu_x^\e}(\omega)
    =&\ \exp\left\{ \frac{2}{\epsilon} \int_0^T v^\epsilon(T-t, \omega(t)) \mathrm{d} \omega(t) - \frac{2}{\epsilon} \int_0^T \left( b^\epsilon(\omega(t)) - v^\epsilon(T-t, \omega(t)) \right)   v^\epsilon(T-t, \omega(t)) \, \mathrm{d}t \right\} \\
    =&\ \exp\left\{ \frac{2}{\epsilon} \int_0^T \left|v^\epsilon(T-t, \omega(t)) \right|^2 \mathrm{d}t + \frac{2}{\epsilon} \int_0^T  v^\epsilon(T-t, \omega(t))   (\mathrm{d} \omega(t)- b^\epsilon(\omega(t))\mathrm{d}t) \right\} .
\end{aligned}
\end{equation*}
Under measure $\nu^\e$, $\omega$ is a trajectory of the the solution process for \eqref{eqn:SDEentropy}, thus we have,
\begin{equation*}
\begin{aligned}
     &\ \E_{\nu^\epsilon} \left[ s_\mathrm{tot}(T, \omega) - s_\mathrm{tot}(0, \omega) \right]\\
     =&\ \E_{\nu^\epsilon} \left[\log\frac{\mathrm{d} \nu_x^\e}{\mathrm{d} \rev \nu_x^\e}(\omega)\right]\\
    =&\ \E_{\nu^\epsilon} \left[ \frac{2}{\epsilon} \int_0^T \left|v^\epsilon(T-t, \omega(t)) \right|^2 \mathrm{d}t + \frac{2}{\epsilon} \int_0^T  v^\epsilon(T-t, \omega(t))   (\mathrm{d} \omega(t)- b^\epsilon(\omega(t))\mathrm{d}t) \right] \\
    =&\ \E_{\nu^\epsilon} \left[ \frac{2}{\epsilon} \int_0^T \left|v^\epsilon(T-t, \omega(t)) \right|^2 \mathrm{d}t + \frac{2}{\epsilon} \int_0^T  v^\epsilon(T-t, \omega(t))   \sqrt{\e}\mathrm{d}B(t) \right]\\
    =&\ \frac{2}{\epsilon} \int_0^T\int_{\mathbb{R}^d} \left|v^\epsilon(T-t, x) \right|^2 \rho^\e(t,x)\d x\mathrm{d}t.
\end{aligned}
\end{equation*}
The proof is complete.
\end{proof}

It is readily observed that when $\nu^\epsilon|_{t=0}$ corresponds to the invariant distribution of the system \eqref{eqn:SDEentropy}, the results, Theorems \ref{theo:secondlaw}, \ref{theo:IFT}, \ref{theo:DFT} and \ref{theo:EPRformula}, naturally reduce to those established in \cite{seifert2005entropy,boffi2024deep}. Furthermore, we note that a discrete analogue of the aforementioned derivation has been utilized to analyze the entropy production formula for jump-diffusion processes \cite{huang2026entropy}.

\subsection{Potential energy representation of total entropy production}
In this subsection, we establish a connection between stochastic thermodynamics and the results presented in Section \ref{sec:potentialenergy}. Recall that the irreversibility of a non-equilibrium process is quantified through the time-forward and time-reversed measures, $\nu^\epsilon_x$ and $\rev{\nu}^\epsilon_x$, as shown in \eqref{ean:entropyderivation}. Consequently, the entropy decomposition in \eqref{eqn:entropyrate} is derived by leveraging the Fokker--Planck equation \eqref{eqn:FPE}.

In Section \ref{sec:potentialenergy}, we have discussed time-reversals of general path measures. In particular, when $\Phi^\e$ takes the form \eqref{phi} and $V$ is not explicitly time-dependent, it follows from \eqref{eqn-19} that
\begin{equation}\label{ean:8-v/ve}
  \log \frac{\mathrm{d} \nu^\e_x}{\mathrm{d} \rev\nu^\e_x} (\omega) =  \log \rho^\epsilon(T, \omega(0)) - \log \rho^\epsilon(0, \omega(T)) + \frac{2}{\epsilon} \big[g^\e(\omega(0)) - g^\e(\omega(T))\big].
\end{equation}
As our system \eqref{eqn:SDEentropy} has a time-independent drift, we see from Remark \ref{remark:5.5}-(iii) that to identify $\nu^\epsilon_x$ with the law of \eqref{eqn:SDEentropy} we need the drift $b^\e$ to be the gradient field $b^\e = -\nabla g^\e$, $V$ to be time-independent and $g$ solving the stationary 2nd-order HJ equation \eqref{time-ind-HJ}.
In this case, the medium entropy production \eqref{def:mEP} is given by
\begin{equation*}
  \mathrm{d}s_\mathrm{m}(t, \omega) = \frac{2}{\epsilon} b^\epsilon(\omega(t)) \circ \mathrm{d}\omega(t) = - \frac{2}{\epsilon} \mathrm{d} [g^\e(\omega(t))].
\end{equation*}
Then we have, using the definition \eqref{def:sysEP} of stochastic entropy,
\begin{equation*}
  \begin{aligned}
    \log \frac{\mathrm{d} \nu^\epsilon_x}{\mathrm{d} \rev{\nu}^\epsilon_x} (\omega) 
    % =& \frac{\rho^\epsilon(T, x)}{\rho^\epsilon(0, \omega(T))} \exp \left\{ \frac{2}{\epsilon} \big[g^\e(x) - g^\e(\omega(T))\big] \right\}\\
    &= \log \rho^\epsilon(T, \omega(0)) - \log \rho^\epsilon(0, \omega(T)) + \frac{2}{\epsilon} \big[g^\e(\omega(0)) - g^\e(\omega(T))\big] \\
    &= s_\mathrm{sys}(T, \omega) - s_\mathrm{sys}(0, \omega) + s_\mathrm{m}(T, \omega) - s_\mathrm{m}(0, \omega) \\
    &= s_\mathrm{tot}(T, \omega) - s_\mathrm{tot}(0, \omega),
  \end{aligned}
\end{equation*}
which once again recovers \eqref{eqn:totalderivation}.

On the other hand, consider the stochastic gradient system \eqref{SDE-2} with time-independent gradient field:
\begin{align}\label{SDE-2-time-idpdt}
  \mathrm dX^\e_x(t)= -\nabla S^\e(X^\e_x(t))\mathrm d t + \sqrt\e \mathrm dB(t), \quad X^\e_x (0) =x,
\end{align}
where $S^\e$ satisfies the stationary version of the 2nd-order HJ equation \eqref{2nd-order HJ}.
We can provide a new decomposition of the irreversibility in terms of the pair of 2nd-order HJ equations \eqref{2nd-order HJ} and \eqref{2nd-order HJ-rev}, as follows: recalling \eqref{ean:8-v/ve},
\begin{equation*}
  \log \frac{\mathrm{d} \nu^\e_x}{\mathrm{d} \rev\nu^\e_x} (\omega) = \frac{1}{\e} \left[ \widetilde S^\e(0,\omega(T)) - S^\e(\omega(T)) \right]  - \frac{1}{\e} \left[ \widetilde S^\e(T,x) - S^\e(x) \right].
\end{equation*}
Comparing the last equality with \eqref{eqn:totalderivation}, we can suggest the following definition of the stochastic total entropy:
\begin{equation}\label{total-entropy}
  s_\mathrm{tot}(t, \omega) = \frac{1}{\e} \left[ \widetilde S^\e(T-t,\omega(t)) - S^\e(\omega(t)) \right].
\end{equation}

As the right-hand side of \eqref{ean:8-v/ve} depends solely on the initial and terminal states of the path $\omega \in \mathcal{C}^{d,T}$, but not on the entire trajectory, this `path-independence' property is characteristic of entropy production and medium entropy production for systems with time-independent drifts. Consequently, this provides an alternative interpretation of the irreversibility of thermodynamic systems from the perspective of the Born-type formula for time marginals in Section \ref{subsec-Born}. Specifically, the non-equilibrium property (i.e., the irreversibility) of a stochastic thermodynamic process $\nu^\e_x$ given in \eqref{mu-tilde} is precisely captured by the probabilistic decomposition \eqref{eqn:nu-Born} of its total measure. This perspective establishes a novel bridge between Euclidean quantum mechanics (the analogy suggested long ago by Schr\"odinger \cite{Schrodinger1932,CruzeiroZambrini1991}) and stochastic thermodynamics.
Let us summarize these observations as the following theorem.
\begin{thm}\label{theo:irreversibilityand2ndHJ}
  Under the assumptions of Theorem \ref{Cor-sde}, suppose further that $V\in C^1_b(\mathbb{R}^d)$ and $f^\e, g^\e \in C_b(\mathbb{R}^d)$. Then, the stochastic total entropy of the stochastic gradient system \eqref{SDE-2-time-idpdt} is determined by the difference \eqref{total-entropy} of the associated backward and stationary forward second-order Hamilton-Jacobi equations, \eqref{2nd-order HJ} and \eqref{2nd-order HJ-rev}.
\end{thm}

Furthermore, our measure-theoretical framework is capable of addressing more general processes, far beyond \eqref{SDE-2} and \eqref{SDE-Markovian} where the drift fields explicitly depend on the time variable. It can provide a novel perspective on the decomposition of total entropy production (cf. \eqref{eqn:RNnu}) and elucidates how the path-independent decomposition explicitly relates to the time-dependent nature of drift terms, offering an alternative to the earlier decomposition in \eqref{eqn:entropydecomposition}. From a thermodynamic standpoint, this formulation highlights the intrinsic connection between entropy production and the interplay of forward and backward Schr\"odinger bridges. This connection has the potential to yield new physical insights, which we aim to explore further in future research.

\section{Conclusion}\label{Conclusion}
Although Schr\"odinger observation \cite{Schrodinger1932} was initially designed as a classical statistical physics (or ``Euclidean'') analogy with quantum theory, and allowed, indeed, to turn into rigorous statements some fundamental intuitions of Feynman (for instance his path integral version of the canonical commutation relation between position and momentum observables (cf.~(4.36), (5.25) in the second part of \cite{chung2003introduction}), Schr\"odinger's idea can and should also be considered from its intrinsically statistical physics viewpoint, as was done here.
A major difference between the first and second viewpoints is the role of time-reversal. To make probabilistic sense of Feynman's commutation relation, we do not really use the time-reversed process $\rev X^\e_x$ of Corollary \ref{Cor-sde-initial-reversed}; more precisely, we use this one only to compute its (usual) drift in terms of the backward drift of Schr\"odinger--Bernstein diffusion $X^\e_x$ \cite[Eq.~(5.22)]{chung2003introduction} (where ``backward'' refers to the time increment at $t$). The coexistence of two distinct time derivatives at a given $t$ is what allowed Feynman to deduce that quantum paths are non-differentiable \cite[Eq.~(5.25)]{chung2003introduction}.

As quantum theory itself, its Euclidean analogy suggested by Schr\"odinger is perfectly invariant under time-reversal (when the potential $V$ is time-independent); the exchange of the given boundary data at the r.h.s. of equation \eqref{Sch-sys} provides another diffusion in the same class. 
This may seem surprising since the drifts involved are generally time-dependent. The formulation of the original observation of Schr\"odinger shows clearly that it is concerned with very improbable phenomena (think about the example he considered himself of the one-dimensional heat equation with a final probability far from its usual decayed evolution). It is this final conditioning which produces an optimal reversible evolution on $[0,T]$.

On the other hand, there is a true time-reversal involved in Corollary \ref{Cor-sde-initial-reversed}, which is at the origin of the results of Section \ref{app-3} regarding the irreversibility of non-equilibrium thermodynamics.
After the rediscovery of Schr\"odinger's idea by the communities of (stochastic) geometric mechanics \cite{huang2023second}, mass transportation \cite{leonard2014survey,mikami2021stochastic}, stochastic optimal control \cite{mikami2004monge,mikami2021stochastic}, we are convinced that it has still a lot to teach us in the debate about reversibility/irreversibility in statistical, quantum mechanics and other fields of science \cite{Galichon2018}. The present work was focused on the fundamental aspect of entropy production
but it has been known for a long time that there is a whole analogy between Thermodynamics and classical Mechanics ((Cf  Refs  \cite{antoniou2002caratheodory,peterson1979analogy}, for instance). This means that part of the tools introduced in  Ref \cite{huang2023second} should also be useful in elaborating the relations with stochastic Thermodynamics.
The consequences of the old and almost forgotten idea of Schr\"{o}dinger deserve clearly to be investigated further.
% The consequences of this old but almost forgotten idea deserve clearly to be investigated further.

\begin{appendices}

\section{Path space: An introduction}\label{appendixA}

Here we provide a more detailed description of the path-space operations discussed in the main text, along with further explanations of the fundamental concepts used throughout this paper.

\subsection*{Scaling} 

For an $\e>0$, the scaling map $\delta_\e : \mathcal{C}^{d,T} \to \mathcal{C}^{d,T}$ is defined as $\delta_\e \omega = \sqrt\e \omega$. 
We denote by $\mu^\e_0 := (\delta_\e)_* \mu_0 = \mu_0 \circ \delta_\e^{-1}$ the $\e$-scaling of the Wiener measure $\mu_0$. 
The scaling measure $\mu^\e_0$ is the probability distribution of the $\e$-scaled Brownian motion $\sqrt\e W$.

It should be noted that this type of $\e$-scaling notation is not applied to a general measure $\nu$ (as it relies on a special symmetry of the Wiener).

\subsection*{Shifts}

For a path $\gamma \in \mathcal{C}^{d,T}$, the shift map associated with $\gamma$ is the map $T_\gamma: \mathcal{C}^{d,T} \to \mathcal{C}^{d,T}$ defined by $T_\gamma\omega = \omega + \gamma$. For a measure $\nu$ on $\mathcal{C}^{d,T}$, the pushforward $(T_\gamma)_*\nu = \nu \circ T_\gamma^{-1}$ is called the shift measure of $\nu$ by $\gamma$.

Let $\mathcal{H}^{d,T}_0 := H_0^1([0,T]; \R^d)$ be the Hilbert space of all $\gamma\in \mathcal{C}^{d,T}_0$ such that each component of $\gamma(t) = (\gamma^1(t), \cdots, \gamma^d(t))$ is absolutely continuous in $t$ and has square-integrable derivatives, equipped with the norm
\begin{equation}\label{CM-norm}
  \|\gamma\|_{H_0^1} := \|\dot \gamma\|_{L^2} = \int_0^T |\dot \gamma(t)|^2 dt.
\end{equation}
This $\mathcal{H}^{d,T}_0$ is called the Cameron--Martin subspace of $\mathcal{C}^{d,T}_0$. For $\gamma \in \mathcal{H}^{d,T}_0$, the Wiener integral 
\begin{equation*}
  i(\gamma)(\omega) = \int_0^T \dot \gamma(t) \d\omega(t).
\end{equation*}
is well-defined and $i(\gamma) \in L^2 (\mathcal{C}^{d,T}_0,\mathcal B(\mathcal{C}^{d,T}_0), \mu_0)$. Indeed, by It\^o's isometry, we have $\E_{\mu_0} |i(\gamma)|^2 = \|\gamma\|_{H_0^1}^2$, that is, the linear mapping $i: \mathcal{H}^{d,T}_0 \to L^2 (\mathcal{C}^{d,T}_0,\mathcal B(\mathcal{C}^{d,T}_0), \mu_0)$ is an isometry.

The Cameron--Martin theorem states that the shift measure $(T_\gamma)_*\mu_0$ is absolutely continuous with respect to $\mu_0$ if and only if $\gamma \in \mathcal{H}^{d,T}_0$; moreover, the Radon--Nikodym derivative of $(T_\gamma)_*\mu_0$ with respect to $\mu_0$ is
\begin{equation*}
  \frac{\d (T_\gamma)_*\mu_0}{\d \mu_0}(\omega) = \exp \left(i(\gamma)(\omega) - \frac{1}{2} \|\gamma\|_{H_0^1}^2\right).
\end{equation*}

For each $x\in \R^d$, we define
\begin{equation*}
  \mathcal{C}^{d,T}_x := T_x \mathcal{C}^{d,T}_0, \quad \mathcal{H}^{d,T}_x := T_x \mathcal{H}^{d,T}_0,
\end{equation*}
where $T_x$ is the shift map associated with $h\equiv x$.
The norm \eqref{CM-norm} can be extended to the whole space $\mathcal{H}^{d,T} := \cup_{x\in\R^d} \mathcal{H}^{d,T}_x$ as a seminorm. We also denote $\mu_x^\e := (T_x)_* \mu^\e_0 = (T_x)_* (\delta_\e)_* \mu_0$,
that is,
\begin{equation}\label{mu0-scaling}
  \mu_x^\e(A) = \mu_0(\delta_\e^{-1} T_x^{-1} A) = \mu_0((A-x)/\sqrt\e), \quad \forall A \in \mathcal B(\mathcal{C}^{d,T}_x),
\end{equation}
or equivalently, for any measurable $f: \mathcal{C}^{d,T}_x \to \R$,
\begin{equation*}
  \int_{\mathcal{C}^{d,T}_x} f(\omega) \mu_x^\e(\d\omega) = \int_{\mathcal{C}^{d,T}_0} f(\omega_x^\e) \mu_0(\d\omega).
\end{equation*}
This means that the law of $x+\sqrt\e W$ under $\mu_0$ is $\mu_x^\e$. It also implies that the law of $(W-x)/\sqrt\e$ under $\mu_x^\e$ is $\mu_0$, or equivalently, $W = x+\sqrt\e B$ where $B$ is a standard Brownian motion under $\mu_x^\e$.

\subsection*{Time-marginals}

For any $t\in [0,T]$, let $\pi_t : \mathcal{C}^{d,T} \to \R^d$ be the projection map at time $t$, given by $\pi_t(\omega) = \omega(t)$. One can regard each $\pi_t$ as a random vector on $(\mathcal{C}^{d,T},\mathcal B(\mathcal{C}^{d,T}))$.
For a measure $\nu$ on $(\mathcal{C}^{d,T},\mathcal B(\mathcal{C}^{d,T}))$, we define its marginal at time $t$ by $\nu|_t := (\pi_t)_* \nu = \nu \circ \pi_t^{-1}$, as a measure on $\R^d$.

The time marginals of the Wiener measure $\mu_0$ have the following Lebesgue densities, known as heat kernels:
\begin{equation*}
  \rho_0(t, x) := \frac{\mathrm{d} \mu_0|_t (x)}{\mathrm{d} x} = \frac{1}{(2\pi t)^{d/2}}e^{-\frac{|x|^2}{2t}}, \quad (t,x) \in [0,T] \times \R^d.
\end{equation*}
% They give the heat semigroup $\{e^{\frac{1}{2} t\Delta}: t\in\R_+\}$, i.e., the semigroup generated by a standard Brownian motion, given by the convolution
% \begin{align*}
%   e^{\frac{1}{2} t\Delta} f(x) = \int_{\R^d} \rho_0(t, x-y) f(y) \d y, \quad f\in C_ c^\infty(\R^d).
% \end{align*}
The marginal of  $\mu_x^\e$ at time $t$ has the Lebesgue density $\rho_x^\e(t, \cdot) = \rho_0(\e t, \cdot - x)$.

\subsection*{Conditioning}

A Borel measurable map $f: \mathcal{C}^{d,T} \to \R^d$ can be regarded as a random element on $(\mathcal{C}^{d,T},\mathcal B(\mathcal{C}^{d,T}))$. One can define the conditional expectation of a $\sigma$-finite measure $\nu$ given $f$, denoted as $\E_\nu(\cdot | f) := \E_\nu(\cdot | \sigma(f))$, as well as the regular conditional measure of $\nu$ given $f=x\in\R^d$, denoted as $\nu(\cdot | f = x)$. They satisfy the relation 
\begin{equation}\label{cond-exp}
  \E_\nu(g | f) = \E_{\nu(\cdot | f = x)}(g)|_{x = f}
\end{equation}
for any $\nu$-integrable function $g: \mathcal{C}^{d,T} \to \R$. For this reason, we shall denote the expectation with respect to the regular conditional measure $\nu(\cdot | f = x)$ by
\begin{equation*}
  \E_\nu(g | f = x) := \E_{\nu(\cdot | f = x)}(g).
\end{equation*}

Recall that $\{ \nu(\cdot | f = x): x\in\R^d \}$ is a transition kernel on $\R^d \times \mathcal B(\mathcal{C}^{d,T})$ such that for all $ A \in \mathcal B(\mathcal{C}^{d,T})$ and $U\in \mathcal B(\R^d)$,
\begin{equation*}
  \nu(A \cap f^{-1}(U)) = \int_U \nu(A | f = x) (f_*\nu)(\d x).
\end{equation*}
If $X: (\mathcal{C}^{d,T},\mathcal B(\mathcal{C}^{d,T})) \to (\mathcal{C}^{d,T},\mathcal B(\mathcal{C}^{d,T}))$ is a measurable map, then
\begin{equation*}
  \begin{split}
    (X_* \nu) (A \cap f^{-1} (U)) &= \nu (X^{-1}(A) \cap (f\circ X)^{-1} (U)) \\
    &= \int_U \nu(X^{-1}(A) \mid f\circ X=x) f_* (X_*\nu)(\d x).
  \end{split}
\end{equation*}
This implies the following formula for regular conditional pushforward measures
\begin{equation}\label{reg-cond-push}
  (X_* \nu) (\cdot \mid f=x) = \nu(X^{-1}(\cdot) \mid f\circ X=x).
\end{equation}

The following lemma is adapted from \cite[Theorem 1]{Leo14b}. Cf. also \cite[Theorem D.13]{DZ98}.

\begin{lem}\label{lem-0}
  Let $\nu$ and $\eta$ be two $\sigma$-finite measures on $(\mathcal{C}^{d,T},\mathcal B(\mathcal{C}^{d,T}))$ satisfying $\nu \ll \eta$. Let $f: \mathcal{C}^{d,T} \to \R^d$ be a Borel measurable map. Then
  \begin{itemize}
    \item[(i)] $f_*\nu \ll f_*\eta$ and
    \begin{equation*}
      \frac{\d f_*\nu}{\d f_*\eta} (x) = \E_{\eta(\cdot | f = x)} \left( \frac{\d\nu}{\d\eta} \right), \quad f_*\eta \text{-a.s. } x \in \R^d;
    \end{equation*}
    \item[(ii)] for $f_*\eta$-a.s. $x\in\R^d$, $\nu(\cdot | f = x) \ll \eta(\cdot | f = x)$ and
    \begin{equation*}
      \frac{\d\nu}{\d\eta}(\omega) = \frac{\d f_*\nu}{\d f_*\eta} (f(\omega)) \frac{\d\nu(\cdot | f = x)}{\d\eta(\cdot | f = x)}(\omega) \bigg|_{x = f(\omega)}, \quad \eta \text{-a.s. } \omega \in \mathcal{C}^{d,T}.
    \end{equation*}
  \end{itemize}
\end{lem}

As the time $t$ projection $\pi_t : \mathcal{C}^{d,T} \to \R^d$, $t\in[0,T]$, is a random vector on $(\mathcal{C}^{d,T},\mathcal B(\mathcal{C}^{d,T}))$, the conditional expectation $\E_\nu(\cdot | \pi_t) = \E_\nu(\cdot | \omega(t))$ of $\nu$ given $\pi_t$ is well-defined, so is the regular conditional measure $\nu(\cdot | \omega(t) = x)$ given $\pi_t(\omega) = \omega(t) = x\in\R^d$.
In particular, for the time 0 projection $\pi_0$, we denote the regular conditional measure
\begin{equation*}
  \nu_x(\d\omega) := \nu( \d\omega | \omega(0) = x).
\end{equation*}
The disintegration theorem, implied by \eqref{cond-exp}, says that
\begin{equation}\label{disint}
  \nu(\d\omega) = \int_{\R^d} \nu_x(\d\omega) \nu|_{t=0}(\d x).
\end{equation}
where $\nu|_{t=0}$ is the marginal of $\nu$ at time $t=0$.
In other words, the measure $\nu$ is determined by its `initial measure' $\nu|_{t=0}$ and `transition measures' $\nu_x$.
Consequently, Lemma \ref{lem-0} implies that

\begin{cor}\label{cor-0}
  Let $\nu$ and $\eta$ be two $\sigma$-finite measures on $(\mathcal{C}^{d,T},\mathcal B(\mathcal{C}^{d,T}))$ satisfying $\nu \ll \eta$. Then
  \begin{itemize}
    \item[(i)] for every $t\in[0,T]$, $\nu|_t \ll \eta|_t$ and
    \begin{equation*}
      \frac{\d\nu|_t}{\d\eta|_t} (x) = \E_{\eta(\cdot | \omega(t) = x)} \left( \frac{\d\nu}{\d\eta} \right), \quad \eta|_t \text{-a.s. } x \in \R^d;
    \end{equation*}
    \item[(ii)] for $\eta|_{t=0}$-a.s. $x\in\R^d$, $\nu_x \ll \eta_x$ and
    \begin{equation*}
      \frac{\d\nu}{\d\eta}(\omega) = \frac{\d\nu|_{t=0}}{\d\eta|_{t=0}} (\omega(0)) \frac{\d\nu_x}{\d\eta_x}(\omega) \bigg|_{x = \omega(0)}, \quad \eta \text{-a.s. } \omega \in \mathcal{C}^{d,T}.
    \end{equation*}
  \end{itemize}
\end{cor}

% As a consequence, we have
% \begin{equation*}
%   \frac{\d\nu}{\d\eta}(\omega) = \E_{\eta_x} \left( \frac{\d\nu}{\d\eta} \right) (\omega) \frac{\d\nu_x}{\d\eta_x}(\omega) \bigg|_{x = \omega(0)}, \quad \eta \text{-a.s. } \omega \in \mathcal{C}^{d,T}.
% \end{equation*}

\subsection*{Kullback--Leibler divergence}

Given two measures $\nu$ and $\eta$ on $(\mathcal{C}^{d,T}, \mathcal B(\mathcal{C}^{d,T}))$, the Kullback--Leibler (KL) divergence (or relative entropy) of $\nu$ with respect to $\eta$ is defined by
\begin{equation*}
  D_{\mathrm{KL}}\left(\nu \| \eta\right) := 
  \begin{cases}
    \E_\nu\left[ \log \left(\frac{\mathrm d \nu}{\mathrm d \eta}\right) \right], & \nu \ll \eta, \\
    \infty, & \text{otherwise}.
  \end{cases}
\end{equation*}
We quote Gibbs' inequality, which states that the above KL divergence takes values in $[0, \infty]$, as
\begin{equation}\label{Gibbs-ineq}
  \E_\nu\left[ \log \left(\frac{\mathrm d \nu}{\mathrm d \eta}\right) \right] = - \E_\nu\left[ \log \left(\frac{\mathrm d \eta}{\mathrm d \nu}\right) \right] \ge \E_\nu\left( 1- \frac{\mathrm d \eta}{\mathrm d \nu} \right) = 0.
\end{equation}
Moreover, the KL divergence equals zero if and only if $\frac{\mathrm d \eta}{\mathrm d \nu} = 1$, i.e., $\nu = \eta$.

The following lemma is taken from \cite[Theorem D.13]{DZ98}. See also \cite[Eq.~(72)]{Leo14b}.

\begin{lem}\label{lemma-KL}
When $\nu \ll \eta$, 
\begin{equation*}
  D_{\mathrm{KL}}\left(\nu \| \eta\right) = D_{\mathrm{KL}}\left(\nu|_{t=0} \| \eta|_{t=0}\right) + \int_{\R^d} D_{\mathrm{KL}}\left(\nu_x \| \eta_x\right) \nu|_{t=0}(\d x).
\end{equation*}
\end{lem}

The proof follows directly from \eqref{disint} and Corollary \ref{cor-0}-(ii). Indeed, 
\begin{equation*}
  \begin{split}
    D_{\mathrm{KL}}\left(\nu \| \eta\right) &= \E_\nu\left[ \log \left( \frac{\d\nu}{\d\eta} \right) \right] \\
    &= \int_{\R^d} \E_\nu \left[ \log \left( \frac{\d\nu}{\d\eta} \right) \bigg| \omega(0)=x \right] \nu|_{t=0}(\d x) \\
    &= \int_{\R^d} \left\{ \log \frac{\d\nu|_{t=0}}{\d\eta|_{t=0}} (x) + \E_\nu \left[ \log \left( \frac{\d\nu_x}{\d\eta_x} \right) \right] \right\} \nu|_{t=0}(\d x) \\
    &= D_{\mathrm{KL}}\left(\nu|_{t=0} \| \eta|_{t=0}\right) + \int_{\R^d} D_{\mathrm{KL}}\left(\nu_x \| \eta_x\right) \nu|_{t=0}(\d x).
  \end{split}
\end{equation*}

\subsection*{Time-reversal operator} 

We define the reverse-time operator $R: \mathcal{C}^{d,T} \to \mathcal{C}^{d,T}$ by $R(\omega) = \omega(T- \cdot)$.
The operator $R$ is clearly a Banach isometry, as well as an involution, i.e., $R^2 = \id$.
For a measure $\nu$ on $(\mathcal{C}^{d,T},\mathcal B(\mathcal{C}^{d,T}))$, we define its time-reversal as the pushforward measure by $R$, 
$$\rev\nu := R_* \nu = \nu \circ R^{-1}.$$

Recall that a continuous process $X = \{X(t)\}_{t\in[0,T]}$ on $(\Omega, \mathcal F, \mathbf P)$ can be regarded as a random variable valued in $(\mathcal{C}^{d,T},\mathcal B(\mathcal{C}^{d,T}))$. We define the time-reversal of $X$ as the process 
$$\rev X := R\circ X = \{X(T - t)\}_{t\in[0,T]}.$$
If the law of $X$ is $\nu$, i.e., $\nu = \P\circ X^{-1}$, then the law of $\rev X$ is $\rev \nu$, since
\begin{equation*}
  \P\circ \rev X^{-1} = \P\circ X^{-1} \circ R^{-1} = \nu \circ R^{-1} = \rev \nu.
\end{equation*}

\subsection*{Stochastic differential equations}

Suppose we are given a  process $b: [0,T] \times \mathcal{C}^{d,T} \to \R^d$, adapted with respect to the canonical filtration $\{ \mathcal B_t(\mathcal{C}^{d,T}) \}_{t\in[0,T]}$. Consider the following functional stochastic differential equation:
\begin{equation}\label{SDE}
  \mathrm dX(t)= b(t,X)\mathrm d t + \mathrm dB(t),
\end{equation}
By a (weak) solution of SDE \eqref{SDE}, we mean a triple $X$, $B$, $(\Omega, \mathcal F, \mathbf P, \{\mathcal F_t\}_{t\in[0,T]})$, where
\begin{itemize}
  \item[(i)] $(\Omega, \mathcal F, \mathbf P, \{\mathcal F_t\}_{t\in[0,T]})$ is a filtered probability space satisfying the usual conditions, equipped with a $d$-dimensional $\{\mathcal F_t\}_{t\in[0,T]}$-Brownian motion $B$,
  \item[(ii)] $\{X(t)\}_{t\in[0,T]}$ is a continuous, adapted $\R^d$-valued process,
  \item[(iii)] $\int_0^T |b(t,X)| \mathrm dt < \infty$ $\P$-a.s.,
  \item[(iv)] the following integral version of \eqref{SDE} holds $\P$-a.s.,
  \begin{equation*}
    X(t)= X(0) + \int_0^t b(s,X)\mathrm d s + B(t), \quad t\in [0,T].
  \end{equation*}
\end{itemize}

In the case where the drift $b$ is given by $b(t, \omega) = b(t, \omega(t))$ for some time-dependent vector field $b$ on $\R^d$, equation \eqref{SDE} then has the form
\begin{equation*}
  \mathrm dX(t)= b(t,X(t))\mathrm d t + \mathrm dB(t),
\end{equation*}
and is said
to be of the Markovian-type.

%Let $(\mathcal{C}'_0[0,T],\mathcal{F}',\mu'_0)$ be an independent copy of the Wiener space $(\mathcal{C}^{d,T},\mathcal{F},\mu_0)$. Denote by $\mathbb{R}^{\mathcal{C}^{d,T}}$ all the functionals on the path space $\mathcal{C}^{d,T}$. 

The following Girsanov theorem, taken from \cite[Theorem 3.5.1]{karatzas2012brownian}, generalizes Cameron--Martin theorem to stochastic drifts.

\begin{lem}\label{Gir-thm}
Let $(\Omega, \mathcal F, \mathbf P, \{\mathcal F_t\}_{t\in[0,T]})$ be a filtered probability space satisfying the usual condition, equipped with a standard $d$-dimensional Brownian motion $\{B(t)\}_{t\in[0,T]}$. Assume: 

(i) $\beta = \{\beta(t)\}_{t\in[0,T]}$ is a measurable $\{\mathcal F_t\}_{t\in[0,T]}$-adapted process; 

(ii) the following Novikov's condition holds,
\begin{align*}
  \E_\P \left[\exp \left( \frac{1}{2} \int_{0}^{T} |\beta(t)|^2\mathrm{d} t\right)\right]<\infty.
\end{align*}
Then, the process $B -\int_0^\cdot \beta(s) \mathrm ds$ is a standard Brownian motion under the probability measure $\mathbf Q$ with density
$$
\frac{\mathrm{d} \mathbf Q}{\mathrm{d} \P}(\omega) = \exp \left(\int_{0}^{T} \beta(t, \omega) \mathrm{d}B(t,\omega)-\frac{1}{2} \int_{0}^{T} |\beta(t, \omega)|^{2} \mathrm{d} t\right).
$$
\end{lem}

The following version of Girsanov theorem, which is a partial converse of Lemma \ref{Gir-thm}, is taken from \cite[Theorem 5.72]{baudoin2014diffusion}.

\begin{lem}\label{Gir-inv}
Let $(\Omega, \mathcal F, \mathbf P, \{\mathcal F_t\}_{t\in[0,T]})$ be a filtered probability space satisfying the usual condition, equipped with a standard $d$-dimensional Brownian motion $\{B(t)\}_{t\in[0,T]}$. Assume that $\mathbf Q$ is equivalent to $\mathbf P$. Then there exists a progressively measurable process $\beta = \{\beta(t)\}_{t\in[0,T]}$ such that

(i) $\beta$ is $\mathbf P$-almost surely squared-integrable, i.e., $\int_{0}^{T} |\beta(t,\omega)|^{2} \mathrm{d} t < \infty$ for $\mathbf P$-a.s. $\omega$;

(ii) the process $B-\int_0^{\cdot}\beta(s)ds$ is a standard Brownian motion under $\mathbf Q$; 

(iii) $\mathbf Q$ has density
$$
\frac{\mathrm{d} \mathbf Q}{\mathrm{d} \mathbf P}(\omega)=\exp \left(\int_{0}^{T} \beta(t,\omega) \mathrm{d}B(t,\omega)-\frac{1}{2} \int_{0}^{T} |\beta(t,\omega)|^{2} \mathrm{d} t\right).
$$
\end{lem}

\section{Overdamped Langevin equations and nonlinear heat equation}

Now we consider the potential energy functional $\Phi^\e: \mathcal{C}^{d,T} \rightarrow \mathbb{R}$ of the form \eqref{phi}.

\begin{lem}
Under the assumptions of Lemma \ref{lem-1}, if the law of $X^\e_x$ is $\nu^\e_x$, then for $s\in [0,T]$ and $\mu_0$-a.s. $\omega\in \mathcal{C}^{d,T}_0$,
\begin{equation}\label{Sto-int1}
  \begin{split}
    & \E_{\mu_0} [D_s g^\e(\omega_x^\e(T)) \mid \mathcal B_s(\mathcal{C}^{d,T}_0)] + \sqrt\e b^\e(s, \omega_x^\e) \\
    &= \int_s^T\E_{\mu_0}[ D_s b^\e(t,\omega_x^\e) \cdot b^\e(t,\omega_x^\e) - D_sV(t, \omega_x^\e(t)) \mid \mathcal B_s(\mathcal{C}^{d,T}_0)] \mathrm{d} t.
  \end{split}
\end{equation}
\end{lem}

\begin{proof}
We take the Malliavin derivative $D_s$ \cite{nualart2006malliavin} to both sides of \eqref{eqn-13}. Since $b^\e$ is adapted to the filtration $\{ \mathcal B_t(\mathcal{C}^{d,T}) \}_{t\in[0,T]}$, we have $D_s V(t,x+ \sqrt\e \omega(t)) = 0$ and $D_s b^\e(t,x+ \sqrt\e \omega)=0$ for $s>t$ \cite[Corollary 1.2.1]{nualart2006malliavin}. The Malliavin derivative of the l.h.s. of \eqref{eqn-13} is
\begin{equation*}
  \begin{split}
    D_s [\text{l.h.s.}] &= \int_{0}^{T} D_s V(t, \omega_x^\e(t)) \mathrm d t + D_s g^\e(\omega_x^\e(T)) \\
    &= \int_{s}^{T} D_s V(t, \omega_x^\e(t)) \mathrm d t + D_s g^\e(\omega_x^\e(T)),
  \end{split}
\end{equation*}
while that of the r.h.s. is
\begin{equation*}
\begin{aligned}
D_s [\text{r.h.s.}] &= -\sqrt\e b^\e(s, \omega_x^\e) - \sqrt\e \int_0^T D_s b^\e(t, \omega_x^\e) \mathrm{d} \omega(t) + \int_0^T D_s b^\e(t, \omega_x^\e) \cdot b^\e(t, \omega_x^\e) \mathrm{d} t \\
&= - \sqrt\e b^\e(s, \omega_x^\e) - \sqrt\e \int_s^T D_s b^\e(t, \omega_x^\e) \mathrm{d} \omega(t) + \int_s^T D_s b^\e(t, \omega_x^\e) \cdot b^\e(t, \omega_x^\e) \mathrm{d} t.
\end{aligned}
\end{equation*}
Then, taking the conditional expectation $\E_{\mu_{0}}[ \cdot \mid \mathcal B_s(\mathcal{C}^{d,T}_0)]$ to the above Malliavin derivatives, we obtain the desired result.
\end{proof}

% \begin{rem}
% \normalfont
% If the process $b^\e$ in the first assertion of the above lemma is only measurable but not adapted, then the anticipating version of Girsanov theorem \cite[Theorem 4.1.2]{nualart2006malliavin} yields that $V$ satisfies the following equation:
% \begin{align*}
% \sqrt\e \E_{\mu_{0}}[D_tF^\e\mid \mathcal{F}_t](\omega_x^\e) = D_sV(t,\omega_x^\e(t)), 
% \end{align*}
% where $F^\e=\log|\det_2(I-Db^\e) |+ \delta (b^\e) -\frac{1}{2}\int_0^T |b^\e(t)|^{2} \mathrm{d} t$ and $\delta$ is Malliavin's divergence operator.
% \end{rem}

We then consider the case of the Markovian-type SDEs, particularly relevant to this paper.

\begin{lem}\label{Cor-sde-lemma}
Let Assumptions \ref{asmp-Phi-1} and \ref{asmp-Phi-cost} hold. Fix $\e>0$ and $x\in\R^d$. Let $X^\e_x$, $B$, $(\Omega, \mathcal F, \mathbf P, \{\mathcal F_t\}_{t\in[0,T]})$ be a weak solution of the following SDE
\begin{equation}\label{SDE-Markovian}
\mathrm dX^\e_x(t)= b^\e(t,X^\e_x(t))\mathrm d t + \sqrt\e \mathrm dB(t), \quad X^\e_x (0) =x,
\end{equation}
where the vector field $b^\e\in C_b^{1,2}([0,T]\times\mathbb{R}^d; \mathbb{R}^d)$ satisfies 
\begin{equation*}
   \E_{\mu_x^\e} \left[\exp \left(\frac{1}{2\e} \int_{0}^{T} |b^\e(t, \omega(t))|^{2} \mathrm{d} t\right)\right]<\infty. 
  % \E_\P \left[\exp \left(\frac{1}{2\e} \int_{0}^{T} |b^\e(t, X^\e_x(t))|^{2} \mathrm{d} t\right)\right]<\infty.
\end{equation*}
Suppose $V\in C^{0,2}_b([0,T]\times\mathbb{R}^d)$ and $g^\e\in C^2_b(\R^d)$. If the law of $X^\e_x$ is $\nu^\e_x$, then the following assertions hold:
\begin{itemize}
  \item[(i)] $b^\e$ satisfies the following time-reversed nonlinear heat (NH) equation:
\begin{equation}\label{NH}
  \begin{cases}
    \partial_t b^\e_i (t,y) + \sum_{j=1}^d b^\e_j(t,y)\pt_i b^\e_j(t,y)+\frac{\e}{2}\Delta b^\e_i(t,y)= \pt_i V(t,y), & (t,y)\in [0,T)\times\R^d, \\
    b^\e(T,y)= -\nabla g^\e(y), & y\in\R^d.
  \end{cases}
\end{equation}
  \item[(ii)] If moreover, the vector field $b^\e$ is a gradient field, i.e., $b^\e = \nabla S^\e$ for some potential function $S^\e\in C^{1,3}_b([0,T]\times\mathbb{R}^d)$. Then $S^\e$ is determined (up to a function depending only on time) by the following second-order Hamilton--Jacobi (2nd-order HJ) equation:
\begin{equation}\label{2nd-order HJ1}
  \begin{cases}
    \partial_t S^\e(t,y)-\frac{1}{2} |\nabla S^\e(t,y)|^2 + \frac{\e}{2}\Delta S^\e(t,y) = -V(t,y), & (t,y)\in [0,T)\times\R^d, \\
    S^\e(T,y) =  g^\e(y), & y\in\R^d.
  \end{cases}
\end{equation}
\end{itemize}
\end{lem}

\begin{proof}

We have $D_s g^\e(\omega_x^\e(T)) = \sqrt\e \nabla g^\e(\omega_x^\e(T))$, $D_s b^\e(t,\omega_x^\e(t)) = \sqrt\e \nabla b^\e(t,\omega_x^\e(t)) \mathbf{1}_{[s,T]}(t)$ and $D_s V(t,\omega_x^\e(t)) = \sqrt\e \nabla V(t,\omega_x^\e(t)) \mathbf{1}_{[s,T]}(t)$. Then equation \eqref{Sto-int1} becomes
\begin{align*}
0&= \E_{\mu_0} \left[ \nabla g^\e(\omega_x^\e(T)) \mid \mathcal B_s(\mathcal{C}^{d,T}_0) \right] + b^\e(s, \omega_x^\e(s)) \\
&\quad -\int_s^T\E_{\mu_0} \left[ \nabla b^\e(t,\omega_x^\e(t)) \cdot b^\e(t,\omega_x^\e(t)) - \nabla V(t, \omega_x^\e(t)) \mid \mathcal B_s(\mathcal{C}^{d,T}_0) \right] \mathrm{d} t\\
&= \E_{\mu_0} \left[ \nabla g^\e(y+\sqrt\e \omega(T-s)) \right] \big|_{y=\omega_x^\e(s)} + b^\e(s, \omega_x^\e(s)) \\
&\quad -\int_s^T\E_{\mu_0} \left[ (\nabla b^\e(t) \cdot b^\e(t) - \nabla V(t)) (y+\sqrt\e\omega(t-s)) \right] \big|_{y=\omega_x^\e(s)}\mathrm{d} t \\
&= \int_{\R^d} \nabla g^\e(y+z) \rho_0^\e(T-s,z)\mathrm{d}z \big|_{y=\omega_x^\e(s)} + b^\e(s, \omega_x^\e(s)) \\
&\quad - \int_s^T\int_{\R^d} [\nabla b^\e(t) \cdot b^\e(t) - \nabla V(t)](y+z) \rho_0^\e(t-s,z) \mathrm{d}z \mathrm{d}t \big|_{y=\omega_x^\e(s)},
\end{align*}
where $\rho_0^\e(t, \cdot)$ is the Lebesgue density of $\sqrt\e W(t)$.
Since the canonical Brownian motion $W(t)$ has full support on $\mathbb{R}^d$, we obtain that for all $y\in \R^d$,
\begin{align*}
b^\e(s,y) &= \int_s^T\int_{\R^d} [\nabla b^\e(t) \cdot b^\e(t) - \nabla V(t)](y+z) \rho_0^\e(t-s,z) \mathrm{d}z \mathrm{d}t - \int_{\R^d} \nabla g^\e(y+z) \rho_0^\e(T-s,z)\mathrm{d}z \\
&= \int_s^T \int_{\R^d} [\nabla b^\e(t) \cdot b^\e(t) - \nabla V(t)](z) \rho_0^\e(t-s,z-y) \mathrm{d}z \mathrm{d}t - \int_{\R^d} \nabla g^\e(z) \rho_0^\e(T-s,z-y)\mathrm{d}z.
% &= \int_s^T e^{\frac{\e}{2}(t-s)\Delta}[ \nabla b^\e(t) \cdot b^\e(t) - \nabla V(t)](y) \mathrm{d}t - e^{\frac{\e}{2}(T-s)\Delta} \nabla g^\e(y), 
\end{align*}
% where $\{e^{\frac{\e}{2} t\Delta}: t\in\R_+\}$ is the semigroup of $\sqrt\e W$.
It is then clear that $b^\e(T) = -\nabla g^\e$. 
% Equation \eqref{NH} follows by differentiating both sides of the above equation with respect to $s$, as follows,
As $\nabla g\in C_b^1(\R^d)$ and $\nabla b^\e \cdot b^\e$, $\nabla V\in C_b^{0,1}([0,T] \times \R^d)$, we apply Lemma \ref{lemma-heat-kernel-itg} and obtain
\begin{equation*}
  \begin{split}
    \pt_s b^\e(s,y) 
    % &= - \nabla b^\e(s,y) \cdot b^\e(s,y) + \nabla V(s,y) \\
    % &\quad - \frac{\e}{2} \Delta \int_s^Te^{\frac{\e}{2}(t-s)\Delta}[\nabla b^\e(t) \cdot b^\e(t) - \nabla V(t)](y)\mathrm{d}t + \frac{\e}{2} \Delta e^{\frac{\e}{2}(T-s)\Delta} \nabla g^\e(y) \\
    &= - \nabla b^\e(s,y) \cdot b^\e(s,y) + \nabla V(s,y) - \frac{\e}{2} \Delta b^\e(s,y).
  \end{split}
\end{equation*}
These prove (i). (ii) follows by quadrature.
\end{proof}

\begin{lem}\label{lemma-heat-kernel-itg}
(i). Let $f_1\in C_b^1(\R^d)$. Define
\begin{equation*}
  J_1(s,y) := \int_{\R^d} f_1(z) \rho_0^\e(T-s,z-y)\mathrm{d}z, \quad (s,y) \in [0,T] \times \R^d.
\end{equation*}
Then $J_1\in C^{1,\infty}([0,T] \times \R^d)$ and
\begin{equation*}
  \pt_s J_1(s,y) = - \frac{\e}{2} \Delta J_1(s,y).
\end{equation*}
(ii). Let $f_2\in C_b^{0,1}([0,T] \times \R^d)$. Define
\begin{equation*}
  J_2(s,y) :=\int_s^T \int_{\R^d} f_2(t,z) \rho_0^\e(t-s,z-y) \mathrm{d}z \mathrm{d}t, \quad (s,y) \in [0,T] \times \R^d.
\end{equation*}
Then $J_2\in C^{1,2}([0,T] \times \R^d)$ and
\begin{equation*}
  \pt_s J_1(s,y) = - \frac{\e}{2} \Delta J_1(s,y) - f_2(s,y).
\end{equation*}
\end{lem}

\begin{proof}
The first statement follows from the dominated convergence theorem. The second result follows from \cite[Chapter 1, Theorems 2, 3, 4 and 5, Section 1.6]{friedman1964partial}.
\end{proof}

\begin{rem}\label{remark:5.5}
\normalfont
(i). Under the assumption of assertion (ii), the vector field $b^\e=\nabla S^\e$ satisfies the following time-reversed viscous Burgers' equation:
\begin{equation}\label{Navier-Stokes}
  \begin{cases}
    \partial_t b^\e (t,y) + (b^\e(t,y) \cdot \nabla) b^\e(t,y) + \frac{\e}{2}\Delta b^\e(t,y)= \nabla V(t,y), & (t,y)\in [0,T)\times\R^d, \\
    b^\e(T,y)= -\nabla g^\e(y), & y\in\R^d.
  \end{cases}
\end{equation}

(ii). In Lemma \ref{Cor-sde-lemma}, we represent the Radon--Nikodym derivative \eqref{mu-tilde} as the Girsanov form \eqref{Gis-tra}, where the drift field of equation \eqref{SDE-Markovian} needs to satisfy the time-reversed nonlinear heat equations \eqref{NH}.
As a comparison, in \cite[Theorem 2.1]{truman2012link}, the Radon--Nikodym derivative  of the Girsanov form \eqref{Gis-tra} for $\e=1$ can be represented by $\exp\{S(T,Y(T))-S(0,Y(0))\}$ with a function $S\in C^{1,2}([0,\infty)\times \mathbb{R}^d)$ for every $T\ge0$, where $Y$ is a solution of the SDE 
\begin{equation*}
  \mathrm dY(t) = b^1 (t,Y(t))\mathrm d t + \mathrm dB(t), \quad Y (0) =0,
\end{equation*}
if and only if $b^1=-\nabla S$ and $S$ satisfies the following Hamilton--Jacobi equation
\begin{align*}
\partial_t S -\frac{1}{2} |\nabla S|^2 +\frac{1}{2} \Delta S=0.  
\end{align*}
They did not need to assume $b^1$ as a gradient field, because they required $T$ to vary in $[0,\infty)$ which allowed them to compare two continuous semimartingales by the uniqueness of Doob--Meyer's decomposition.
% In this case, $u$ satisfies the following homogeneous viscous Burgers' equation
% \begin{align*}
% \partial_t u(t,x) +\frac{1}{2} \Delta u(t,x)+ \left(u(t,x) \cdot \nabla \right) u(t,x)=0.
% \end{align*}

(iii). Theorem \ref{Cor-sde}, if imposing strong conditions that $V\in C^{0,2}_b([0,T]\times\mathbb{R}^d)$ and $g^\e\in C^2_b(\R^d)$, can be implied by Lemma \ref{Cor-sde-lemma} by plugging equations \eqref{2nd-order HJ1} into \eqref{eqn-20}.

(iv). When $b^\e$ is not explicitly time-dependent, the law of $X^\e_x$ is $\nu^\e_x$ implies $b^\e = -\nabla g^\e$, and $V$ is time-independent and satisfy (up to a constant for $g^\e$) the time-independent 2nd-order HJ equation \eqref{time-ind-HJ}.
\end{rem}

\section{Some remarks for Assumption \ref{asmp-Phi-3}}\label{app-rmk}

(i). A sufficient condition for Assumption \ref{asmp-Phi-3}-(i) is that there exist $r_0<\frac{1}{2}$ and $M_0 \in\R$, such that for all $\omega \in \mathcal{C}^{d,T}$, 
$$  
\Phi^0(\omega) \ge M_0 - r_0 \|\omega\|_{H_0^1}^2.
$$
In particular, a bounded below $\Phi^0$ is sufficient.
To prove the sufficiency, we first note that the function $I_{\Phi^0}^x$ defined in \eqref{rate-func} takes values in $[0,\infty]$ and is lower semicontinuous, since $I$ is a rate function and $\Phi^0$ is continuous.
Next, we show the goodness of $I_{\Phi^0}^x$, that is, for all $\beta\ge0$, the level set $\{\omega \in \mathcal{C}^{d,T}_x: I_{\Phi^0}^x (\omega) \le \beta\}$ is compact. It follows from Assumption \ref{asmp-Phi-3} that
\begin{equation*}
  \{\omega \in \mathcal{C}^{d,T}_x: I_{\Phi^0}^x (\omega) \le \beta\} \subset \{\omega \in \mathcal{H}^{d,T}_x: (\textstyle{\frac{1}{2}} - r_0) \|\omega\|_{H_0^1}^2 \le \beta - M_0 + \inf_{\omega \in \mathcal{C}^{d,T}_x} [\Phi^0(\omega) + I(\omega)]\},
\end{equation*}
where the latter set is compact in $\mathcal{C}^{d,T}_x$, and the former is closed. The compactness of the former follows.

(ii). A bounded below $\Phi^0$ is sufficient for condition \eqref{tail-cond}.
% which can be verified by virtue of Assumption \ref{asmp-Phi-1}, as follows: for some $\gamma>1$, we take $r = \frac{\e\alpha}{\gamma} >0$ in Assumption \ref{asmp-Phi-1} and use the Fernique theorem \eqref{Fernique}, 
% \begin{equation*}
%   \limsup_{\e \rightarrow 0} \e \log \E_{\mu_x^\e} \left[ \exp \left( -\gamma \frac{F+\Phi^\e}{\e} \right) \right] \le \limsup_{\e \rightarrow 0} \e \log \E_{\mu_x^\e} \left[ \exp \left( \alpha \|\omega\|_T^2 + \frac{\gamma (\|F\|_\infty - M)}{\e} \right) \right] < \infty.
% \end{equation*}

(iii). Sufficient conditions for condition \eqref{exp-lim-Phi} are that 
\begin{equation}\label{cond-Phi0}
  \lim _{M \rightarrow \infty} \limsup _{\epsilon \rightarrow 0} \epsilon \log \E_{\mu_x^\e} \left[e^{(\Phi^0 - \Phi^\e) / \epsilon} \ind_{\left\{\Phi^0  - \Phi^\e \geq M\right\}}\right] = -\infty,
\end{equation}
and either one of the following conditions holds: 
\begin{itemize}
  \item[a)] as $\e\to0$, $\Phi^\e$ exponentially good approximates $\Phi^0$ under $\mu_x^\e$, in the sense that for every $\delta>0$,
  \begin{equation*}
    \limsup _{\epsilon \rightarrow 0} \epsilon \log \mu_x^\e \left(\Phi^0 - \Phi^\e >\delta \right) = -\infty;
  \end{equation*}
  \item[b)] as $\e\to0$, $\Phi^\e$ converges compactly to $\Phi^0$.
\end{itemize}
Indeed, we note that it suffices to consider the case when $\Phi^0 - \Phi^\e < M$ for some $M>0$ and all $0<\e\ll 1$, by virtue of condition \eqref{cond-Phi0}. We first prove the sufficiency of a): for any $\delta>0$,
\begin{equation*}
  \begin{split}
    \epsilon \log \E_{\mu_x^\e} \left[e^{(\Phi^0 - \Phi^\e) / \epsilon} \right] \le &\ \epsilon \log \E_{\mu_x^\e} \left[e^{(\Phi^0 - \Phi^\e) / \epsilon} \ind_{\left\{\Phi^0  - \Phi^\e \le \delta\right\}}\right] \vee \epsilon \log \E_{\mu_x^\e} \left[e^{(\Phi^0 - \Phi^\e) / \epsilon} \ind_{\left\{\Phi^0  - \Phi^\e > \delta\right\}}\right] \\
    \le&\ \delta \vee M \epsilon \log \mu_x^\e \left(\Phi^0 - \Phi^\e >\delta \right),
  \end{split}
\end{equation*}
which implies condition \eqref{exp-lim-Phi} by taking the limits $\e\to 0$ and $\delta\to 0$. Then we verify the sufficiency of b), as follows. One the one hand, fix a compact neighborhood $K\subset \mathcal{C}^{d,T}_x$ of the constant path $\omega_x \equiv x$, we have
\begin{equation*}
  \begin{split}
    \liminf_{\e\to0} \epsilon \log \E_{\mu_x^\e} \left[e^{(\Phi^0 - \Phi^\e) / \epsilon} \right] &\ge \liminf_{\e\to0} \epsilon \log \E_{\mu_x^\e} \left[e^{(\Phi^0 - \Phi^\e) / \epsilon} \ind_K\right] \\
    &\ge -\liminf_{\e\to 0} \sup_{\omega\in K} |\Phi^0(\omega) - \Phi^\e(\omega)| + \liminf_{\e\to0} \epsilon \log \mu_x^\e(K) \\
    &\ge 0 - \inf_{\omega\in K^\circ} I(\omega) \ge 0 -I(\omega_x) = 0.
  \end{split}
\end{equation*}
On the other hand, for every $\alpha>0$, the goodness of the rate function $I$ of $\{\mu_x^\e: \e>0\}$ implies that the level set $K_\alpha := \{I\le \alpha\}$ is compact. Then
\begin{equation*}
  \begin{split}
    \limsup_{\e\to0} \epsilon \log \E_{\mu_x^\e} \left[e^{(\Phi^0 - \Phi^\e) / \epsilon} \right] &\le \limsup_{\e\to0} \epsilon \log \E_{\mu_x^\e} \left[e^{(\Phi^0 - \Phi^\e) / \epsilon} \ind_{K_\alpha}\right] \vee \limsup_{\e\to0} \epsilon \log \E_{\mu_x^\e} \left[e^{(\Phi^0 - \Phi^\e) / \epsilon} \ind_{K_\alpha^c}\right] \\
    &\le \limsup_{\e\to0} \sup_{\omega\in K_\alpha} |\Phi^0(\omega) - \Phi^\e(\omega)| \vee \left[ M + \limsup_{\e\to0} \epsilon \log \mu_x^\e(K_\alpha^c) \right] \\
    &\le 0 \vee \left[ M - \inf_{\omega\in \overline{K_\alpha^c}} I(\omega) \right] \le 0 \vee (M - \alpha),
  \end{split}
\end{equation*}
which goes to zero by taking the limit $\alpha \to \infty$.

(iv). Combining the above three remarks, one can summarize a set of sufficient conditions for Assumption \ref{asmp-Phi-3}: $\Phi^0$ is bounded and $\Phi^\e$ is bounded below uniformly in $0<\e\ll 1$, and either a) or b) in the last remark holds.

\end{appendices}

\addcontentsline{toc}{section}{Acknowledgments}
\section*{Acknowledgments}
J. Hu acknowledges support from the School of Physical and Mathematical Sciences at Nanyang Technological University, and support from MOE AcRF Tier 1 under Grant No. 04MNP004255C230OST02.
The work of Q.~Huang is supported by the National Natural Science
Foundation of China under Grant No. 12501241, the Basic Research Program of Jiangsu under Grant No. BK20251280, the Zhishan Young Scholar Program of Southeast University, the Start-up Research Fund of Southeast University under Grant No. RF1028624194 and the Jiangsu Provincial Scientific Research Center of Applied Mathematics under Grant No. BK20233002. Y. Huang would like to thank the support from his research groups in the National University of Singapore and the City University of Hong Kong during his postdoctoral period. The authors acknowledge helpful discussions with Prof. Jong-Min Park. 

\addcontentsline{toc}{section}{Statements and Declarations}
\section*{Statements and Declarations}

\paragraph{Data availability.} We do not analyze or generate any datasets, because our work proceeds within a theoretical and mathematical approach.

\paragraph{Competing interests.} The authors have no competing interests to declare that are relevant to the content of this article.

\paragraph{Declaration of generative AI and AI-assisted technologies in the manuscript preparation process.} During the preparation of this work, the authors used Gemini and ChatGPT in order to correct grammar mistakes. After using these tools, the authors reviewed and edited the content as needed and take full responsibility for the content of the published article.

\addcontentsline{toc}{section}{References}
\bibliographystyle{abbrv}
\bibliography{Refs}

@book{ikeda2014stochastic,
  title={{Stochastic Differential Equations and Diffusion Processes}},
  author={Ikeda, Nobuyuki and Watanabe, Shinzo},
  volume={24},
  year={2014},
  publisher={Elsevier}
}

@article{Schrodinger1926,
  author  = {E. Schr{\"o}dinger},
  title   = {Quantisierung als Eigenwertproblem},
  journal = {Ann. Phys.},
  volume  = {79},
  year    = {1926},
  pages   = {361--376}
}

@article{Beurling1960,
  author  = {A. Beurling},
  title   = {An automorphism of product measures},
  journal = {Ann. Math.},
  volume  = {72},
  number  = {1},
  year    = {1960},
  pages   = {189--200}
}

@article{CruzeiroZambrini1991,
  author  = {A. B. Cruzeiro and J.-C. Zambrini},
  title   = {Malliavin calculus and {Euclidean} quantum mechanics},
  journal = {J. Funct. Anal.},
  volume  = {96},
  number  = {1},
  year    = {1991},
  pages   = {62--95}
}

@incollection{CruzeiroFollmerZambrini2006,
  author    = {A. B. Cruzeiro and H. F{\"o}llmer and J.-C. Zambrini},
  title     = {Bernstein processes associated with a {Markov} process},
  booktitle = {Stochastic Analysis and Mathematical Physics},
  editor    = {R. Rebolledo},
  publisher = {Birkh{\"a}user},
  year      = {2006},
  pages     = {41--72}
}

@article{Schrodinger1932,
  author  = {E. Schr{\"o}dinger},
  title   = {Sur la th{\'e}orie relativiste de l'{\'e}lectron et l'interpr{\'e}tation de la m{\'e}canique quantique},
  journal = {Ann. Inst. H. Poincar{\'e}},
  volume  = {2},
  year    = {1932},
  pages   = {269--310}
}

@article{Zambrini1986,
  author  = {J.-C. Zambrini},
  title   = {Variational processes and stochastic versions of mechanics},
  journal = {J. Math. Phys.},
  volume  = {27},
  number  = {9},
  year    = {1986},
  pages   = {2307--2330}
}

@book{peliti2021stochastic,
  title={{Stochastic Thermodynamics: An Introduction}},
  author={Peliti, Luca and Pigolotti, Simone},
  year={2021},
  publisher={Princeton University Press}
}

@book{mikami2021stochastic,
  title={{Stochastic Optimal Transportation: Stochastic Control with Fixed Marginals}},
  author={Mikami, Toshio},
  year={2021},
  publisher={Springer Nature}
}

@inproceedings{Bernstein1932,
  author    = {S. Bernstein},
  title     = {Sur les liaisons entre les grandeurs al{\'e}atoires},
  booktitle = {Proc. Int. Congr. Math.},
  address   = {Z{\"u}rich},
  year      = {1932},
  volume    = {1},
  pages     = {288--309}
}

@article{PrivaultYangZambrini2016,
  author  = {N. Privault and X. Yang and J.-C. Zambrini},
  title   = {Large deviations for Bernstein bridges},
  journal = {Stochastic Process. Appl.},
  volume  = {126},
  number  = {5},
  year    = {2016},
  pages   = {1285--1308}
}

@book{Galichon2018,
  author    = {A. Galichon},
  title     = {Optimal Transport Methods in Economics},
  publisher = {Princeton Univ. Press},
  year      = {2018}
}

@book{hollander2000large,
  title={Large Deviations},
  author={Hollander, Frank},
  volume={14},
  year={2000},
  publisher={American Mathematical Soc.}
}

@article{selk2024smallnoise,
  title={The Small-Noise Limit of the Most Likely Element is the Most Likely Element in the Small-Noise Limit},
  author={Selk, Zachary and Honnappa, Harsha},
  journal={ALEA, Lat. Am. J. Probabil. Math. Stat.},
  volume={21},
  pages={849--862},
  year={2024},
  doi={10.30757/ALEA.v21-35}
}

@article{du2021graph,
  title={The graph limit of the minimizer of the {O}nsager-{M}achlup functional and its computation},
  author={Du, Qiang and Li, Tiejun and Li, Xiaoguang and Ren, Weiqing},
  journal={Science China Mathematics},
  volume={64},
  pages={239--280},
  year={2021},
  publisher={Springer}
}

@article{selk2021information,
  title={Information projection on {B}anach spaces with applications to state independent {KL}-weighted optimal control},
  author={Selk, Zachary and Haskell, William and Honnappa, Harsha},
  journal={Applied Mathematics and Optimization},
  volume={84},
  number={1},
  pages={805--835},
  year={2021},
  publisher={Springer}
}

@article{leonard2014survey,
  title={A survey of the {S}chr{\"o}dinger problem and some of its connections with optimal transport},
  author={L{\'e}onard, Christian},
  journal={Discrete \& Continuous Dynamical Systems},
  volume={34},
  number={4},
  pages={1533--1574},
  year={2014},
  publisher={American Institute of Mathematical Sciences (AIMS)}
}

@article{durr1978onsager,
  title={The {O}nsager-{M}achlup function as {L}agrangian for the most probable path of a diffusion process},
  author={D{\"u}rr, Detlef and Bach, Alexander},
  journal={Communications in Mathematical Physics},
  volume={60},
  number={2},
  pages={153--170},
  year={1978},
  publisher={Springer}
}

@article{dashti2013map,
  title={{MAP} estimators and their consistency in {B}ayesian nonparametric inverse problems},
  author={Dashti, Masoumeh and Law, Kody J H and Stuart, Andrew M and Voss, Jochen},
  journal={Inverse Problems},
  volume={29},
  number={9},
  pages={095017},
  year={2013},
  publisher={IOP Publishing}
}

@article{pinski2015kullback,
  title={Kullback--{L}eibler approximation for probability measures on infinite dimensional spaces},
  author={Pinski, Francis J and Simpson, Gideon and Stuart, Andrew M and Weber, Hendrik},
  journal={SIAM Journal on Mathematical Analysis},
  volume={47},
  number={6},
  pages={4091--4122},
  year={2015},
  publisher={SIAM}
}

@article{lu2017gaussian,
  title={Gaussian approximations for transition paths in {B}rownian dynamics},
  author={Lu, Yulong and Stuart, Andrew and Weber, Hendrik},
  journal={SIAM Journal on Mathematical Analysis},
  volume={49},
  number={4},
  pages={3005--3047},
  year={2017},
  publisher={SIAM}
}

@book{baudoin2014diffusion,
  title={Diffusion Processes and Stochastic Calculus},
  author={Baudoin, Fabrice},
  year={2014},
  publisher={European Mathematical Society}
}

@book{nualart2006malliavin,
  title={The Malliavin Calculus and Related Topics},
  author={Nualart, David},
  volume={1995},
  year={2006},
  publisher={Springer}
}

@book{Kallenberg2021FoundationsOM,
  title={Foundations of Modern Probability},
  author={Kallenberg, Olav},
  year={1997},
  publisher={Springer}
}

@book{friedman1964partial,
  author    = {Avner Friedman},
  title     = {Partial Differential Equations of Parabolic Type},
  publisher = {Prentice-Hall},
  year      = {1964},
  address   = {Englewood Cliffs, NJ},
}

@article{truman2012link,
  title={A link of stochastic differential equations to nonlinear parabolic equations},
  author={Truman, Aubrey and Wang, FengYu and Wu, JiangLun and Yang, Wei},
  journal={Science China Mathematics},
  volume={55},
  pages={1971--1976},
  year={2012},
  publisher={Springer}
}

@book{karatzas2012brownian,
  title={Brownian Motion and Stochastic Calculus},
  author={Karatzas, Ioannis and Shreve, Steven},
  volume={113},
  year={2012},
  publisher={Springer Science \& Business Media}
}

@article{huang2023second,
  title={From second-order differential geometry to stochastic geometric mechanics},
  author={Huang, Qiao and Zambrini, Jean-Claude},
  journal={Journal of Nonlinear Science},
  volume={33},
  number={4},
  pages={67},
  year={2023},
  publisher={Springer}
}

@inproceedings{huang2023gauge,
  title={Gauge Transformations in Stochastic Geometric Mechanics},
  author={Huang, Qiao and Zambrini, Jean-Claude},
  booktitle={International Conference on Geometric Science of Information},
  pages={583--591},
  year={2023},
  organization={Springer}
}

@book{DZ98,
   author = {Amir Dembo and Ofer Zeitouni},
   city = {Berlin, Heidelberg},
   isbn = {978-3-642-03310-0},
   publisher = {Springer Berlin Heidelberg},
   title = {Large Deviations Techniques and Applications},
   volume = {38},
   year = {2010}
}

@book{fleming2006controlled,
   author = {Wendell H. Fleming and H. Mete Soner},
   city = {New York},
   doi = {10.1007/0-387-31071-1},
   isbn = {0-387-26045-5},
   publisher = {Springer-Verlag},
   title = {Controlled Markov Processes and Viscosity Solutions},
   volume = {25},
   year = {2006}
}

@article{jamison1974reciprocal,
  title={Reciprocal processes},
  author={Jamison, Benton},
  journal={Zeitschrift f{\"u}r Wahrscheinlichkeitstheorie und Verwandte Gebiete},
  volume={30},
  number={1},
  pages={65--86},
  year={1974},
  publisher={Springer}
}

@incollection{Leo14b,
  title={Some properties of path measures},
  author={L{\'e}onard, Christian},
  booktitle={S{\'e}minaire de Probabilit{\'e}s XLVI},
  pages={207--230},
  year={2014},
  publisher={Springer}
}

@book{bogachev1998gaussian,
  title={Gaussian Measures},
  author={Bogachev, Vladimir Igorevich},
  number={62},
  year={1998},
  publisher={American Mathematical Soc.},
  series={Placeholder Series}
}

@article{seifert2005entropy,
  title={Entropy production along a stochastic trajectory and an integral fluctuation theorem},
  author={Seifert, Udo},
  journal={Physical Review Letters},
  volume={95},
  number={4},
  pages={040602},
  year={2005},
  publisher={APS}
}

@book{Nel01,
  title={Dynamical Theories of Brownian Motion},
  author={Nelson, Edward},
  edition={2nd},
  volume={106},
  year={2001},
  publisher={Princeton University Press}
}

@article{anderson1982reverse,
  title={Reverse-time diffusion equation models},
  author={Anderson, Brian DO},
  journal={Stochastic Processes and their Applications},
  volume={12},
  number={3},
  pages={313--326},
  year={1982},
  publisher={Elsevier}
}

@article{boffi2024deep,
  title={Deep learning probability flows and entropy production rates in active matter},
  author={Boffi, Nicholas M and Vanden-Eijnden, Eric},
  journal={Proceedings of the National Academy of Sciences},
  volume={121},
  number={25},
  pages={e2318106121},
  year={2024},
  publisher={National Acad Sciences}
}

@article{ge2012analytical,
  title={Analytical mechanics in stochastic dynamics: Most probable path, large-deviation rate function and {Hamilton--Jacobi} equation},
  author={Ge, Hao and Qian, Hong},
  journal={International Journal of Modern Physics B},
  volume={26},
  number={24},
  pages={1230012},
  year={2012},
  publisher={World Scientific}
}

@book{arnol2013mathematical,
  title={{Mathematical Methods of Classical Mechanics}},
  author={Arnol'd, Vladimir Igorevich},
  volume={60},
  year={2013},
  publisher={Springer Science \& Business Media}
}

@article{ao2008emerging,
  title={Emerging of stochastic dynamical equalities and steady state thermodynamics from Darwinian dynamics},
  author={Ao, Ping},
  journal={Communications in Theoretical Physics},
  volume={49},
  number={5},
  pages={1073},
  year={2008},
  publisher={IOP Publishing}
}

@article{qian2011nonlinear,
  title={Nonlinear stochastic dynamics of mesoscopic homogeneous biochemical reaction systems—an analytical theory},
  author={Qian, Hong},
  journal={Nonlinearity},
  volume={24},
  number={6},
  pages={R19},
  year={2011},
  publisher={IOP Publishing}
}

@misc{hairer2011signal,
  title={Signal processing problems on function space: Bayesian formulation, stochastic {PDEs} and effective {MCMC} methods},
  author={Hairer, Martin and Stuart, Andrew M and Voss, Jochen},
  journal={The Oxford handbook of nonlinear filtering},
  pages={833--873},
  year={2011},
  publisher={Oxford Univ. Press, Oxford}
}

@article{stuart2010inverse,
  title={Inverse problems: a {B}ayesian perspective},
  author={Stuart, Andrew M},
  journal={Acta Numerica},
  volume={19},
  pages={451--559},
  year={2010},
  publisher={Cambridge University Press}
}

@article{todorov2009efficient,
  title={Efficient computation of optimal actions},
  author={Todorov, Emanuel},
  journal={Proceedings of the National Academy of Sciences},
  volume={106},
  number={28},
  pages={11478--11483},
  year={2009},
  publisher={National Academy of Sciences}
}

@article{peng1992stochastic,
  title={Stochastic {H}amilton--{J}acobi--{B}ellman equations},
  author={Peng, Shige},
  journal={SIAM Journal on Control and Optimization},
  volume={30},
  number={2},
  pages={284--304},
  year={1992},
  publisher={SIAM}
}

@incollection{pardoux1999bsdes,
  title={{BSDEs}, weak convergence and homogenization of semilinear {PDE}s},
  author={Pardoux, {\'E}tienne},
  booktitle={Nonlinear analysis, differential equations and control},
  pages={503--549},
  year={1999},
  publisher={Springer}
}

@article{crandall1984some,
  title={Some properties of viscosity solutions of {H}amilton-{J}acobi equations},
  author={Crandall, Michael G and Evans, Lawrence C and Lions, P-L},
  journal={Transactions of the American Mathematical Society},
  volume={282},
  number={2},
  pages={487--502},
  year={1984}
}

@article{bismut1976linear,
  title={Linear quadratic optimal stochastic control with random coefficients},
  author={Bismut, Jean-Michel},
  journal={SIAM Journal on Control and Optimization},
  volume={14},
  number={3},
  pages={419--444},
  year={1976},
  publisher={SIAM}
}

@article{bismut1973conjugate,
  title={Conjugate convex functions in optimal stochastic control},
  author={Bismut, Jean-Michel},
  journal={Journal of Mathematical Analysis and Applications},
  volume={44},
  number={2},
  pages={384--404},
  year={1973},
  publisher={Academic Press}
}

@article{pardoux1990adapted,
  title={Adapted solution of a backward stochastic differential equation},
  author={Pardoux, Etienne and Peng, Shige},
  journal={Systems \& Control Letters},
  volume={14},
  number={1},
  pages={55--61},
  year={1990},
  publisher={Elsevier}
}

@article{mikami2004monge,
  title={Monge’s problem with a quadratic cost by the zero-noise limit of h-path processes},
  author={Mikami, Toshio},
  journal={Probability Theory and Related Fields},
  volume={129},
  number={2},
  pages={245--260},
  year={2004},
  publisher={Springer}
}

@article{leonard2012schrodinger,
  title={From the {Schr\"o}dinger problem to the {M}onge--{K}antorovich problem},
  author={L{\'e}onard, Christian},
  journal={Journal of Functional Analysis},
  volume={262},
  number={4},
  pages={1879--1920},
  year={2012},
  publisher={Elsevier}
}

@article{qian2001mathematical,
  title={Mathematical formalism for isothermal linear irreversibility},
  author={Qian, Hong},
  journal={Proceedings of the Royal Society of London. Series A: Mathematical, Physical and Engineering Sciences},
  volume={457},
  number={2011},
  pages={1645--1655},
  year={2001},
  publisher={The Royal Society}
}

@article{onsager1953fluctuations,
  title={Fluctuations and irreversible processes},
  author={Onsager, Lars and Machlup, Stefan},
  journal={Physical Review},
  volume={91},
  number={6},
  pages={1505},
  year={1953},
  publisher={APS}
}

@article{machlup1953fluctuations,
  title={Fluctuations and irreversible process. {II}. Systems with kinetic energy},
  author={Machlup, Stefan and Onsager, Lars},
  journal={Physical Review},
  volume={91},
  number={6},
  pages={1512},
  year={1953},
  publisher={APS}
}

@article{miao2024emergence,
  title={Emergence of Newtonian deterministic causality from stochastic motions in continuous space and time},
  author={Miao, Bing and Qian, Hong and Wu, Yong-Shi},
  journal={arXiv preprint arXiv:2406.02405},
  year={2024}
}

@article{huang2023stochastic,
  title={Stochastic geometric mechanics in nonequilibrium thermodynamics: {Schr\"odinger meets Onsager}},
  author={Huang, Qiao and Zambrini, Jean-Claude},
  journal={Journal of Physics A: Mathematical and Theoretical},
  volume={56},
  number={13},
  pages={134003},
  year={2023},
  publisher={IOP Publishing}
}

@inproceedings{huang2022hamilton,
  title={{Hamilton--Jacobi--Bellman} equations in stochastic geometric mechanics},
  author={Huang, Qiao and Zambrini, Jean-Claude},
  booktitle={Physical Sciences Forum},
  volume={5},
  pages={37},
  year={2022},
  organization={MDPI}
}

@article{huang2026entropy,
  title={Entropy production in non-Gaussian active matter: A unified fluctuation theorem and deep learning framework},
  author={Huang, Yuanfei and Liu, Chengyu and Miao, Bing and Zhou, Xiang},
  journal={Physical Review Letters},
  volume={136},
  number={6},
  pages={068302},
  year={2026},
  publisher={APS}
}

@book{giaquinta1996calculus,
  title={Calculus of Variations I, II},
  author={Giaquinta, Mariano and Hildebrandt, Stefan},
  volume={310, 311},
  year={1996},
  publisher={Springer-Verlag Berlin Heidelberg}
}

@book{chung2003introduction,
  title={Introduction to Random Time and Quantum Randomness},
  author={Chung, Kai Lai and Zambrini, Jean-Claude},
  edition={new},
  volume={1},
  year={2003},
  publisher={World Scientific}
}

@article{leonard2011stochastic,
  title={Stochastic derivatives and generalized $h$-transforms of {Markov} processes},
  author={L{\'e}onard, Christian},
  journal={arXiv preprint arXiv:1102.3172},
  year={2011}
}

@article{leonard2022feynman,
  title={{Feynman-Kac} formula under a finite entropy condition},
  author={L{\'e}onard, Christian},
  journal={Probability Theory and Related Fields},
  volume={184},
  number={3},
  pages={1029--1091},
  year={2022},
  publisher={Springer}
}

@article{antoniou2002caratheodory,
  title={Caratheodory and the foundations of thermodynamics and statistical physics},
  author={Antoniou, Ioannis E},
  journal={Foundations of Physics},
  volume={32},
  number={4},
  pages={627--641},
  year={2002},
  publisher={Springer}
}

@article{peterson1979analogy,
  title={Analogy between thermodynamics and mechanics},
  author={Peterson, Mark A},
  journal={American Journal of Physics},
  volume={47},
  number={6},
  pages={488--490},
  year={1979},
  publisher={AIP Publishing}
}

\end{document}